\documentclass[conference]{IEEEtran}

\usepackage{balance} 
\usepackage{tikz}
\usepackage{amsmath}
\usepackage{algorithm}
\usepackage{algorithmic}
\usepackage{microtype}
\usepackage{graphicx}
\usepackage{subfigure}
\usepackage{booktabs} 
\usepackage{times}
\usepackage{soul}
\usepackage{url}
\usepackage{booktabs}
\usepackage{algorithm}
\usepackage{algorithmic}
\usepackage{multirow}
\usepackage{xcolor}         
\usepackage[dvipsnames]{xcolor}
\usepackage{amssymb}
\usepackage{url}
\usepackage{mathdots}
\usepackage{amsmath}
\usepackage{amsfonts}
\usepackage{color}
\usepackage{xcolor}         
\usepackage[table]{xcolor}  
\definecolor{LightGray}{gray}{0.9} 
\usepackage{amsmath}
\usepackage{mathtools}
\usepackage{amsthm}
\usepackage{booktabs, tabularx}
\usepackage{listings}
\newcolumntype{Y}{>{\raggedright\arraybackslash}X}

\newcommand{\CQ}[1]{{\color{red}{(C.Q.: #1)}}}
\newcommand{\qiu}{\color{black}}

\newcommand{\todo}[1]{\textcolor{red}{(\textbf{TODO:~#1})}}

\newcommand{\DEL}[1]{\iffalse #1 \fi}
\newcommand{\rd}{\color{BrickRed}}
\newcommand{\rev}{\color{black}}
\newcommand{\bl}{\color{RoyalBlue}}
\newcommand{\gr}{\color{ForestGreen}}

\newcommand{\cam}{\color{black}}

\newtheorem{definition}{Definition}

\newtheorem{proposition}{\bf{Proposition}}
\newtheorem{reproposition}{\bf{Proposition}}
\newtheorem*{reproposition2.}{Proposition 2.}

\ifCLASSINFOpdf
\else
\fi
\usepackage{caption}
\begin{document}
%
\title{\textsc{DistA-LP}: Privacy-Preserving Distributed \\ Linear Programming for Metric Differential Privacy}


\author{
    \IEEEauthorblockN{
        Michael Oluwole, Ruiyao Liu, Aleksander Stevens,
        and Chenxi Qiu\textsuperscript{*}
    }
    \IEEEauthorblockA{
        University of North Texas\\
        \{MichaelOluwole, RuiyaoLiu, AleksanderStevens\}@my.unt.edu,
        Chenxi.Qiu@unt.edu\\
        \textsuperscript{*}Corresponding author
    }
}
	

%


\IEEEoverridecommandlockouts
\makeatletter\def\@IEEEpubidpullup{6.5\baselineskip}\makeatother
\IEEEpubid{\parbox{\columnwidth}{
		Network and Distributed System Security (NDSS) Symposium 2027\\
		22--26 March 2027, Seoul, Republic of Korea\\
		ISBN 978-1-970672-09-1\\  
		https://dx.doi.org/10.14722/ndss.2027.230054\\
		www.ndss-symposium.org
}
\hspace{\columnsep}\makebox[\columnwidth]{}}

\maketitle

\begin{abstract}
\emph{Linear-programming (LP)} mechanisms for \emph{metric differential privacy (mDP)} enable fine-grained, task-aware utility control, but they are computationally expensive and become privacy-fragile in distributed settings: uploading user-specific LP coefficient matrices (e.g., distance and loss blocks) to a central solver can itself leak sensitive information. We present \textsc{DistA-LP}, a privacy-preserving distributed mDP optimization framework that replaces raw coefficient uploads with compact surrogate parameterizations and releases only calibrated, perturbed parameter vectors, thereby reducing both upload-stage leakage and problem dimensionality. On the server side, \textsc{DistA-LP} reconstructs approximate coefficients and solves a scalable master-subproblem formulation via Benders decomposition. To certify end-to-end privacy, we introduce a \emph{probabilistic mDP (PmDP)} auditor that empirically quantifies the privacy loss induced by parameter release, allocates the required privacy budget prior to optimization, and composes it with the downstream mechanism to obtain an end-to-end mDP guarantee. Experiments on large real-world road-network datasets (Rome, New York City, and London) show that \textsc{DistA-LP} achieves near-zero mDP violation rates, substantially reduces computation time compared to monolithic and hybrid baselines, and matches or improves utility, thereby demonstrating stronger upload privacy and significantly improved scalability.
\end{abstract}


%
\IEEEpeerreviewmaketitle

\section{Introduction}
\emph{Local Differential Privacy (LDP)}~\cite{Duchi-FOCS2013} provides a uniform indistinguishability guarantee across the entire input domain. While robust, this one-size-fits-all guarantee can be overly conservative for applications that only require local semantic or spatial obfuscation. For example, in text analytics, enforcing strong indistinguishability between semantically distant words (e.g., ``cancer'' vs. ``pizza'') can require substantial randomization, often rendering the released text uninformative; similarly, location privacy may only require hiding a user within a neighborhood rather than an entire city~\cite{ImolaUAI2022,Carvalho2021TEMHU,Liu-CCS2025}.

To mitigate this utility loss, \emph{metric differential privacy (mDP)}~\cite{Chatzikokolakis-PETS2013} scales indistinguishability via a task-specific metric $d$. By requiring that inputs $x_i$ and $x_j$ close under $d$ induce similar output distributions, mDP avoids over-perturbation for distant inputs. This distance-aware framework significantly improves utility while maintaining formal guarantees and has proven effective in geo-location~\cite{Andres-CCS2013,Yu-NDSS2017,Shokri-PoPETs2015,Ahmed-ICDCS2016}, natural language~\cite{ImolaUAI2022,Carvalho2021TEMHU}, and perceptual data such as voice and images~\cite{han2020voice,chen2021perceptual}.

Unlike LDP, designing mechanisms for mDP is complicated by non-uniform privacy constraints and heterogeneous utility sensitivities. Under mDP, the privacy guarantee varies across input pairs, and utility loss depends on both the magnitude and the \emph{direction} of perturbation~\cite{Qiu-TMC2022}. Standard mechanisms, such as the \emph{Laplace mechanism}~\cite{Andres-CCS2013} and the \emph{exponential mechanism}~\cite{Carvalho2021TEMHU}, satisfy mDP by calibrating noise solely to perturbation magnitude. However, this isotropic approach ignores directional variations in utility across the output space, often resulting in suboptimal utility-privacy trade-offs.

To accommodate task-specific utility sensitivities, many recent works turn to \emph{optimization-based} mechanism design, most commonly formulated as \emph{linear programs (LPs)}~\cite{Bordenabe-CCS2014, ImolaUAI2022, Liu-CCS2025}. After discretizing the secret (input) domain $\mathcal{X}$ and the reported (output) domain $\mathcal{Y}$, these methods explicitly encode the utility loss of each candidate perturbation and optimize the full perturbation distribution. This yields finely tuned, often near-optimal for the chosen discretization, mechanisms, but shifts the primary challenge to computational scalability: the LP optimizes the conditional probabilities $P(y\mid x)$ for all $x\in\mathcal{X}$ and $y\in\mathcal{Y}$, resulting in $\mathcal{O}(|\mathcal{X}||\mathcal{Y}|)$ decision variables that can quickly reach millions on moderately fine grids, rendering off-the-shelf solvers impractical at scale.

\looseness = -1

To improve scalability, recent work applies decomposition techniques (e.g., Dantzig-Wolfe~\cite{Qiu-TMC2022} and Benders~\cite{Qiu-IJCAI2024}) to split a large LP into smaller subproblems, enabling parallelism and reducing per-solver complexity. Complementarily, hybrid frameworks~\cite{ImolaUAI2022} integrate the \emph{Weighted Exponential Mechanism (WEM)} with LP by using WEM for inputs where the utility-loss landscape is relatively flat (i.e., candidate outputs incur similar loss), and reserving LP optimization for utility-critical regions where small probability changes can markedly affect utility, thereby reducing computation while preserving mDP guarantees.  Nonetheless, both lines of work still rely on a \emph{centralized} LP backbone, so the maximum feasible secret-domain size $|\mathcal{X}|$ remains bounded by single-machine solver capacity (at most $\sim 500$ in our setting).

More recently, \emph{distributed} LP frameworks have been proposed to further improve scalability by offloading parts of the optimization to individual users~\cite{Qiu-PETS2025,Liu-CCS2025}. Specifically, \cite{Qiu-PETS2025} lets each user construct a local LP and upload the resulting coefficient blocks, after which the server solves a global LP that aggregates all users' contributions. In parallel, \cite{Liu-CCS2025} proposes an anchor-based two-phase framework in which each user selects a small set of anchor records, allowing the server to solve a compact LP over a reduced domain. While these designs substantially reduce runtime and retain competitive utility, their privacy treatment remains limited: \cite{Qiu-PETS2025} does not account for potential privacy leakage from the uploaded coefficient matrices, and \cite{Liu-CCS2025} must spend additional privacy budget to compensate for anchor-induced distance approximation errors, which can increase utility loss.

\DEL{
\begin{figure}
\centering
 \includegraphics[width=0.98\linewidth]{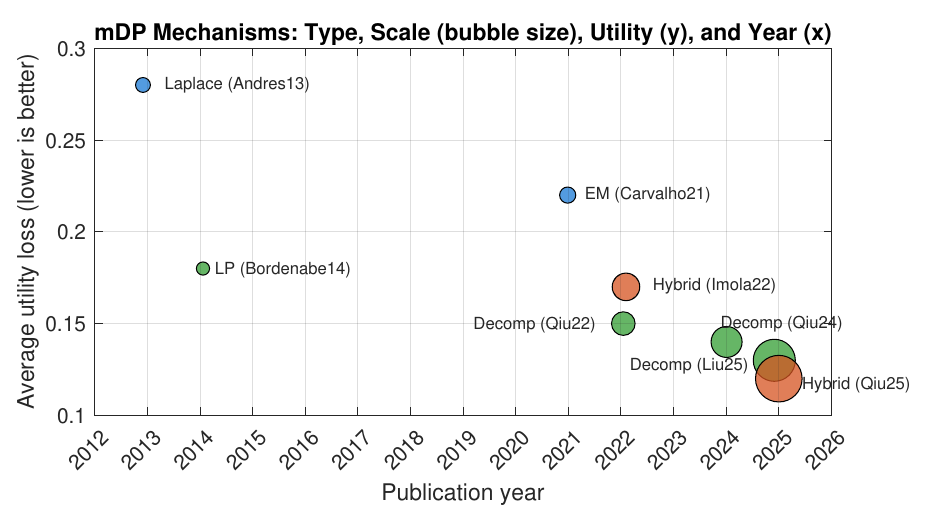}
\caption{Comparison of mDP data perturbation methods. \todo{To update}.}
\label{fig:relatedworks}
\end{figure}}

\subsection*{Our Contributions}
\noindent \textbf{Contribution 1: \underline{Dist}ributed \underline{A}pproximation \underline{LP} (\textsc{DistA-LP}) framework.} In this paper, we propose \textsc{DistA-LP}, a \emph{distributed LP framework} that optimizes the perturbation distribution and enables scalable perturbation optimization under mDP. In \textsc{DistA-LP}, each user $m$ locally constructs a set of LP coefficient blocks that capture the distance and utility-loss structure within its relevance neighborhood $\mathcal{N}_m$ ($\mathcal{N}_m \subset \mathcal{X}$ and $|\mathcal{N}_m| \ll |\mathcal{X}|
$). A key distinguishing feature is that we \emph{explicitly model and quantify the privacy loss} incurred by coefficient uploads, an aspect overlooked in prior distributed mDP formulations such as~\cite{Qiu-PETS2025}. Because these coefficient blocks are derived from users' authentic records, they can inadvertently reveal sensitive spatial or semantic correlations. To mitigate this risk, \textsc{DistA-LP} analytically characterizes the \emph{privacy loss} induced by coefficient uploads and allocates a corresponding portion of the global privacy budget using the \emph{sequential composition} property of mDP~\cite{Liu-CCS2025}, ensuring that the cumulative leakage remains rigorously bounded.


On the server side, \textsc{DistA-LP} uses the uploaded LP coefficients to assemble a global LP that integrates a \emph{WEM} structure~\cite{ImolaUAI2022, Qiu-PETS2025}, handling distant record pairs analytically while explicitly optimizing over nearby, utility-sensitive ones. To further scale, \textsc{DistA-LP} adopts \emph{Benders decomposition}: a compact master problem optimizes the global mechanism variables and WEM weights, while each user solves an independent neighborhood subproblem over $\mathcal{N}_m$ that acts as a feasibility/optimality check and returns cutting planes (from the subproblem dual) to iteratively tighten the master \cite{Qiu-IJCAI2024}.


\smallskip
\noindent \textbf{Contribution 2: Coefficient surrogate modeling, parameter perturbation, and probabilistic mDP accounting.}
Notably, LP coefficient blocks are deterministically derived from users' true records; thus, uploading them directly can leak information about the underlying records. To mitigate this risk, we replace raw coefficient uploads with a \emph{surrogate-based} approach: each user fits a low-dimensional surrogate that approximates the required coefficient blocks and uploads only \emph{perturbed surrogate parameters}. We consider both low-rank models (truncated \emph{singular value decomposition (SVD)}) and parametric/basis-function surrogates (\emph{polynomial, Gaussian, and radial basis function (RBF)}) \cite{GolubVanLoan2013,SeberLee2012,Bishop2006,Buhmann2003}, and systematically compare them in terms of approximation fidelity and privacy leakage; in our setting, \emph{low-rank SVD} offers the best fidelity--privacy trade-off. Moreover, we design the parameter perturbation to be \emph{structure-preserving} and to \emph{minimize surrogate reconstruction error}, so that the induced coefficient approximation does not degrade utility or violate LP feasibility. \looseness = -1 

Because the uploaded parameters lie in a \emph{continuous} space, tightly bounding the induced privacy-loss distribution in closed form can be challenging under structure-preserving perturbations. We therefore adopt \emph{probabilistic mDP} (Definition~\ref{def:PmDP}) accounting and certify upload-stage privacy using a sampling-based \emph{posterior leakage auditor} (based on Proposition \ref{prop:sampling}). The server then reconstructs approximate coefficient blocks from the perturbed parameters and solves the subsequent LP using the Benders decomposition framework.

\smallskip    
\noindent \textbf{Contribution 3: Empirical validation on large-scale geolocation datasets.} 
We conduct extensive experiments on large-scale real-world road-network datasets from Rome, New York City, and London~\cite{openstreetmap} to validate the effectiveness of \textsc{DistA-LP}. Our results show that \textsc{DistA-LP} substantially improves scalability over centralized LP-based mechanisms, reducing solver runtime and memory consumption while maintaining competitive, and often superior, utility. Compared with predefined-noise baselines (Laplace~\cite{Andres-CCS2013} and the exponential mechanism~\cite{Chatzikokolakis-PoPETs2015}), LP-based methods~\cite{Bordenabe-CCS2014, Liu-CCS2025}, and hybrid approaches such as COPT~\cite{ImolaUAI2022}, \textsc{DistA-LP} achieves lower expected utility loss while attaining near-zero empirical mDP violation ratios under comparable privacy budgets. We further demonstrate that the proposed surrogate-based coefficient approximation enables communication- and computation-efficient uploads, and that the probabilistic auditing procedure provides accurate empirical estimates of upload-stage privacy leakage. Collectively, these findings indicate that \textsc{DistA-LP} makes LP-based mDP optimization practical at city scale without sacrificing formal privacy guarantees.

Section~\ref{sec:prelim} reviews the fundamentals of mDP, defines privacy loss, and introduces the baseline LP formulation. Section~\ref{sec:overview} presents an overview of our approach and highlights the key research challenges. Section~\ref{sec:method} describes the distributed pipeline, including surrogate construction, parameter perturbation, and reconstruction, along with the WEM+LP formulation enhanced by Benders decomposition. Section~\ref{sec:exp} reports the experimental evaluation of our approach. Section~\ref{sec:related} discusses related work, and Section~\ref{sec:conclusion} concludes the paper.

\section{Preliminaries}
\label{sec:prelim}
In this section, we present the preliminaries of mDP, including the underlying threat model, the formal definition of mDP, the LP-based optimization approaches, and the WEM integration.

\smallskip
\noindent \textbf{Threat model.} In general, a data perturbation mechanism can be represented as a probabilistic mapping $\mathcal{M}: \mathcal{X} \rightarrow \mathcal{Y}$, where the domain $\mathcal{X}$ denotes the users’ secret data domain, and the range $\mathcal{Y}$ denotes the perturbed data domain. Based on the intrinsic characteristics of $\mathcal{X}$, a metric $d: \mathcal{X}^2 \rightarrow \mathbb{R}$ is defined to quantify the distance between any two records $x_i, x_j \in \mathcal{X}$, denoted as $d_{x_i, x_j}$.

We adopt a standard \emph{honest-but-curious} threat model~\cite{Wang-WWW2017}, where the server faithfully executes the protocol but attempts to infer sensitive information from received reports. Specifically, given the perturbation mechanism $\mathcal{M}$ and a prior distribution of $X$, the server can apply Bayes' rule to compute the posterior distribution of $X$ upon observing a perturbed record $y$~\cite{Yu-NDSS2017}. This model reflects realistic deployments where data-collecting organizations are not fully trusted. Conversely, we assume users are honest and correctly apply the perturbation protocol; collusion, either among users or between users and the server, is excluded as it falls outside the scope of LDP/mDP frameworks~\cite{To-TMC2017}.

To formally bound the server's ability to distinguish between secrets via these posterior distributions, we define the privacy loss as follows:
\begin{definition}[Metric-normalized privacy loss]
\label{def:PL}
Let $(\mathcal{X}, d)$ be a metric space of secrets, and let $\mathcal{M}$ be a perturbation mechanism with input domain $\mathcal{X}$ and output space $\mathcal{Y}$. We define the \emph{metric-normalized privacy loss} between two inputs $x_i, x_j \in \mathcal{X}$ at output $y \in \mathcal{Y}$ as
\begin{equation}
\label{eq:privacyloss}
\mathrm{PL}_{x_i, x_j}(y) = \frac{1}{d_{x_i, x_j}} \ln \left|\frac{\Pr[\mathcal{M}(x_i)=y]}{\Pr[\mathcal{M}(x_j)=y]}\right|.
\end{equation}
Intuitively, $\mathrm{PL}_{x_i,x_j}(y)$ quantifies how differently the mechanism can produce output $y$ when the input is $x_i$ rather than $x_j$, normalized by the metric distance between $x_i$ and $x_j$. Smaller values indicate stronger indistinguishability between $x_i$ and $x_j$, and hence less information leakage to an adversary (e.g., the server) through posterior inference.
\end{definition}

In the remainder of the paper, for brevity, we use ``privacy loss'' to refer to ``metric-normalized privacy loss.''

\begin{definition}[mDP \cite{Andres-CCS2013}]
\label{def:mDP}
We say that a perturbation mechanism $\mathcal{M}: \mathcal{X} \rightarrow \mathcal{Y}$ satisfies $(\epsilon, d)$-\emph{mDP} if and only if that the privacy loss is uniformly bounded by $\epsilon$, i.e.,
\begin{equation}
\label{eq:mDP}
\sup_{x_i \neq x_j} \sup_{y \in \mathcal{Y}}
\mathrm{PL}_{x_i, x_j}(y) \leq \epsilon.
\end{equation}
\end{definition}

Over more than a decade of research on mDP, a variety of mechanisms have been proposed to satisfy mDP, including Laplace-based approaches \cite{Andres-CCS2013} and exponential mechanisms \cite{chatzikokolakis2015constructing}. While these predefined mechanisms offer general-purpose solutions with clear theoretical guarantees, their design is largely driven by distance alone, and therefore they generally do not capture the heterogeneous sensitivity of utility loss under perturbations in different directions. As a result, they may be suboptimal in terms of privacy-utility tradeoff. This observation has motivated \emph{optimization-based approaches}, which seek to directly construct perturbation mechanisms that minimize utility loss subject to mDP constraints.

\smallskip
\noindent \textbf{Optimization (LP)-based approaches.}
Optimizing a perturbation mechanism $\mathcal{M}$ over continuous secret or output spaces is computationally challenging. Therefore, a common approach is to discretize $\mathcal{X}$ and $\mathcal{Y}$ into finite sets~\cite{ImolaUAI2022}. Let $\mathcal{X}=\{x_i\}_{i=1}^{N}$ and $\mathcal{Y}=\{y_k\}_{k=1}^{K}$. Then, $\mathcal{M}$ can be represented by a stochastic matrix $\mathbf{Z}=\{z_{i,k}\}_{(x_i,y_k)\in\mathcal{X}\times\mathcal{Y}}$, where $z_{i,k}$ denotes the probability of outputting $y_k$ given input $x_i$. Under this discretization, the mDP condition in Eq.~(\ref{eq:mDP}) can be written as the following linear constraints:
\begin{equation}
\label{eq:mDPdiscrete}
z_{i,k} - e^{\epsilon d_{x_i,x_j}} z_{j,k} \le 0,
\quad \forall x_i,x_j\in\mathcal{X},~\forall y_k\in\mathcal{Y}.
\end{equation}
Each row of $\mathbf{Z}$ must also satisfy the normalization constraints: 
\begin{equation}
\label{eq:unitmeasure}
\textstyle 
\sum_{k=1}^{K} z_{i,k} = 1, \quad \forall x_i \in \mathcal{X}.
\end{equation}
Let $c_{x_i,y_k}$ denote the utility loss incurred when reporting $y_k$ given the true record $x_i$, and let $\pi_{x_i}$ denote the prior probability that the true record equals $x_i$.
Define the cost matrix $\mathbf{A}^{\textsf{cy}}\in\mathbb{R}^{|\mathcal{X}|\times |\mathcal{Y}|}$ by
$[\mathbf{A}^{\textsf{cy}}]_{i,k}:=\pi_{x_i}c_{x_i,y_k}$.
Then the expected utility loss is
\begin{equation}
\label{eq:LP-loss}
\textstyle 
\mathcal{L}(\mathbf{Z};{\gr\mathbf{A}^{\textsf{cy}}})
=\sum_{i=1}^{N}\sum_{k=1}^{K} [{\gr\mathbf{A}^{\textsf{cy}}}]_{i,k} z_{i,k}.
\end{equation}
The optimal perturbation matrix is obtained by the LP
\begin{equation}
\label{eq:LPobj_origin}
\min_{\mathbf{Z}}~ \mathcal{L}(\mathbf{Z};\mathbf{A}^{\textsf{cy}})
\quad \mathrm{s.t.}~ \text{Eqs.~(\ref{eq:mDPdiscrete}) and (\ref{eq:unitmeasure}) hold.}
\end{equation}
This formulation involves $\mathcal{O}(|\mathcal{X}||\mathcal{Y}|)$ decision variables and $\mathcal{O}(|\mathcal{X}|^2|\mathcal{Y}|)$ constraints, and the pairwise mDP constraints quickly become the bottleneck for efficiency.

\begin{figure}[t]
\centering
\begin{minipage}{0.50\textwidth}
\centering
  \subfigure{
\includegraphics[width=0.70\textwidth]{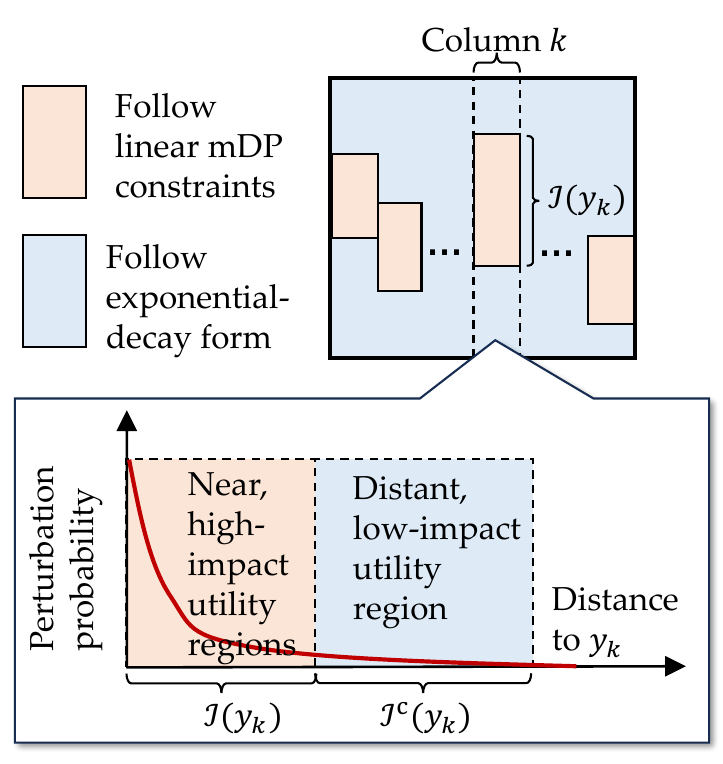}}
\vspace{-0.10in}
\caption{Idea of WEM+LP.}
\label{fig:WEM}
\end{minipage}
\vspace{-0.10in}
\end{figure}

\smallskip
\noindent \textbf{Weighted exponential mechanism (WEM) integration.}
To improve scalability, prior works such as~\cite{ImolaUAI2022,Qiu-PETS2025} impose a WEM-inspired structure on many perturbation probabilities. As illustrated in Fig.~\ref{fig:WEM}, for each fixed output $y_k$, the input domain $\mathcal{X}$ is partitioned into a \emph{nearby, high-impact utility region} $\mathcal{I}(y_k)$ and a \emph{distant, low-impact utility region} $\mathcal{I}^{\mathrm{c}}(y_k)$. Instead of optimizing every probability $z_{i,k}$ through LP, the probabilities corresponding to $x_i \in \mathcal{I}^{\mathrm{c}}(y_k)$ are constrained to share an \emph{exponential-decay form}. The rationale is that distant input-output pairs typically reside in relatively flat regions of the utility landscape and thus have limited utility impact, making a structured exponential tail sufficient.

\begin{figure*}[t]
\centering
\begin{minipage}{0.33\textwidth}
\centering
  \subfigure{
\includegraphics[width=1.00\textwidth, height=0.187\textheight]{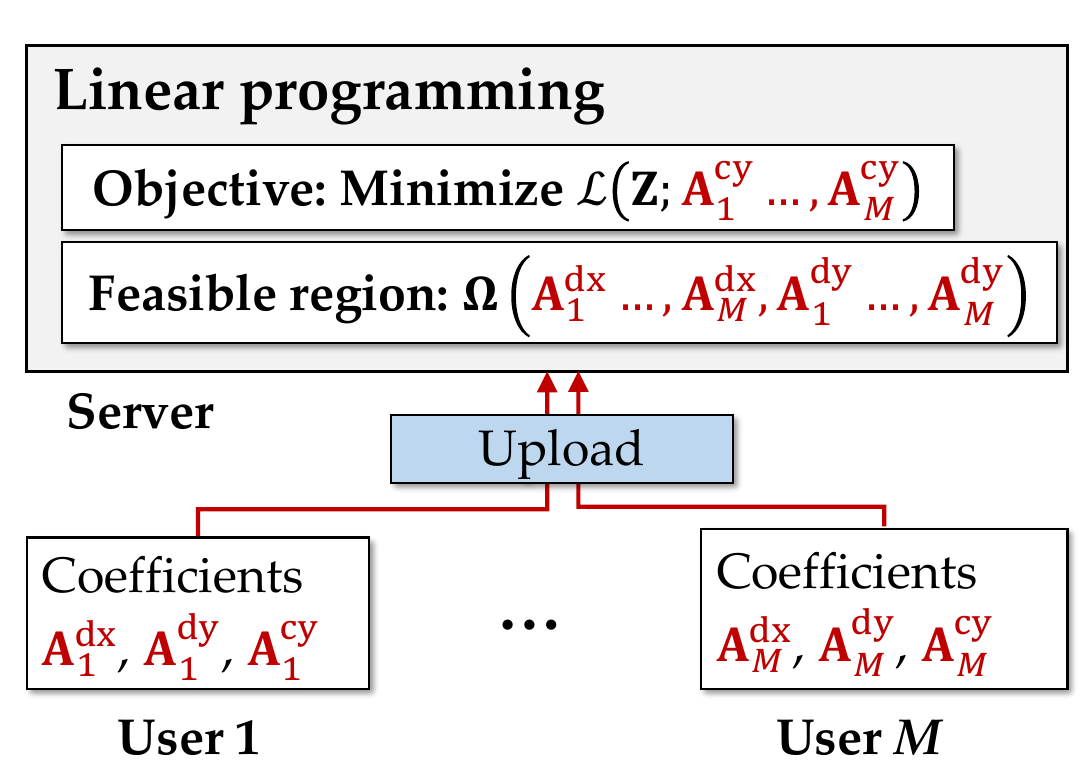}}
\caption{Distributed LP Framework: Uploading Coefficient Blocks $\mathbf{A}_m^{\textsf{dx}}, \mathbf{A}_m^{\textsf{dy}}, \mathbf{A}_m^{\textsf{cy}}$ by each user $m$ ($m = 1, ..., M$).
}
\label{fig:motivatingframework}
\end{minipage}
\hspace{0.1in}
\begin{minipage}{0.64\textwidth}
\centering
  \subfigure[ $\mathbf{A}^{\textsf{dx}}$ 
  ]{
\includegraphics[width=0.30\textwidth, height=0.17\textheight]{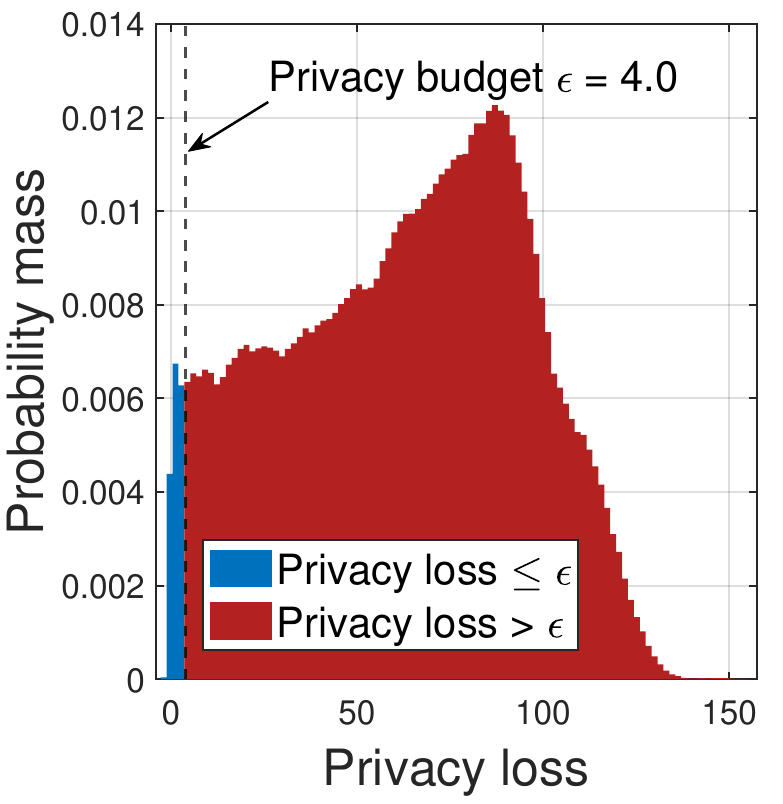}}
\hspace{0.02in}
  \subfigure[$\mathbf{A}^{\textsf{dy}}$ 
  ]{
\includegraphics[width=0.30\textwidth, height=0.17\textheight]{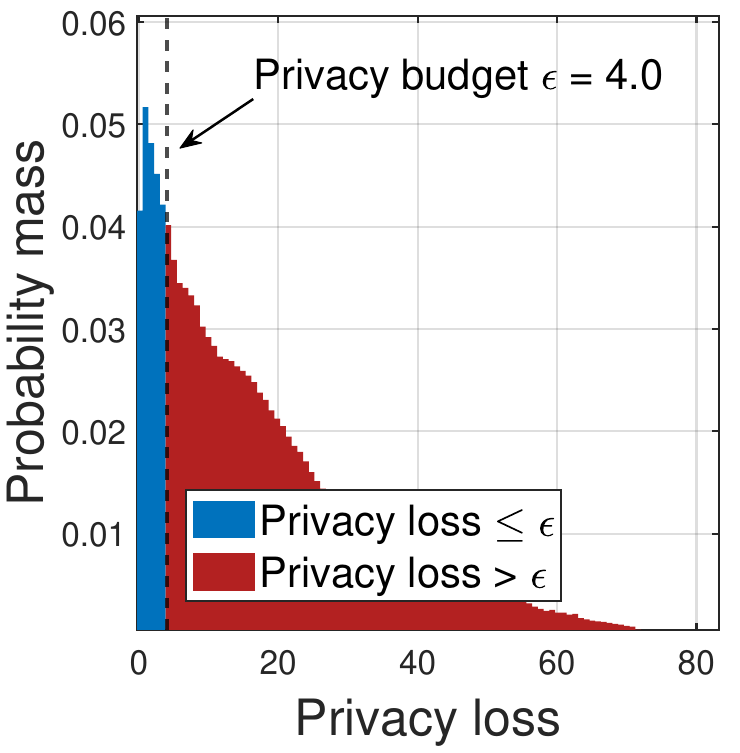}}
\hspace{0.02in}
 \subfigure[$\mathbf{A}^{\textsf{cy}}$ 
 ]{
\includegraphics[width=0.30\textwidth, height=0.17\textheight]{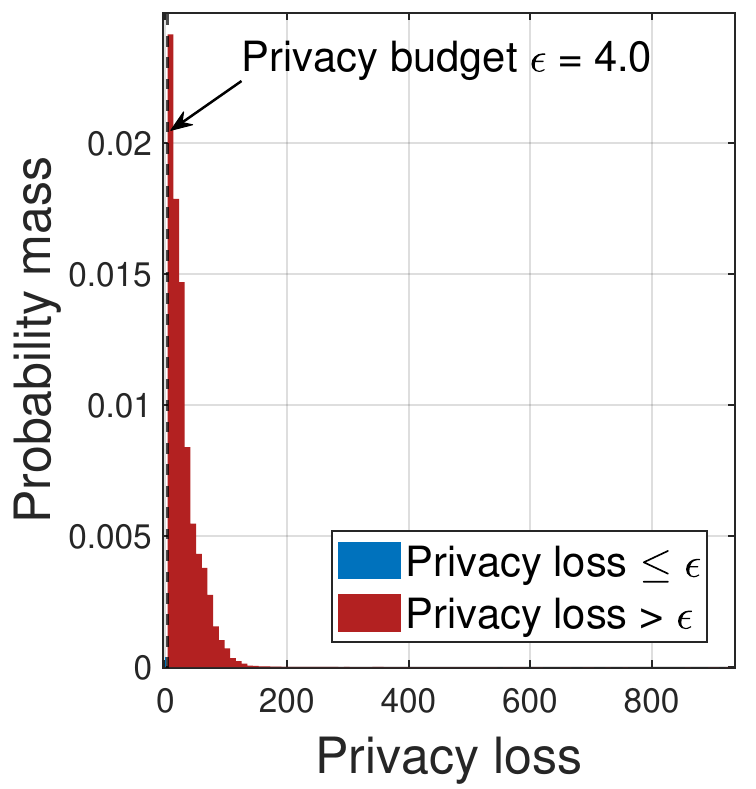}}
\caption{Empirical distributions of the privacy loss induced by uploading $\mathbf{A}_m^{\textsf{dx}}, \mathbf{A}_m^{\textsf{dy}},$ and $\mathbf{A}_m^{\textsf{cy}}$. For $\epsilon=4$, the fractions of samples exceeding the budget are $96.97\%$, $80.64\%$, and $79.73\%$, respectively (exceedances highlighted in red). 
}
\label{fig:motivatingexample}
\end{minipage}
\vspace{-0.00in}
\end{figure*}

Concretely, let $\mathbf{w}=[w_1,\ldots,w_K]$ be nonnegative output weights.
For each output $y_k$, define an \emph{exempt} set $\mathcal{I}(y_k)\subseteq\mathcal{X}$, for example, the $r$ nearest $x_i$ to $y_k$~\cite{ImolaUAI2022} or those within a distance threshold~\cite{Qiu-PETS2025}.
For every non-exempt (“far”) pair $x_i\notin\mathcal{I}(y_k)$, we impose the \emph{WEM constraints}
\begin{equation}
\label{eq:mDPexp}
z_{i,k} = w_k \exp\Big(-\tfrac{\epsilon}{2} d(x_i,y_k)\Big),
\quad \forall y_k\in\mathcal{Y},~\forall x_i\notin\mathcal{I}(y_k).
\end{equation}
Hence, the optimization no longer treats all $|\mathcal{X}|\,|\mathcal{Y}|$ probabilities as independent decision variables; only the exempt probabilities, totaling $\sum_{k=1}^{K} |\mathcal{I}(y_k)|$, together with the $K$ weight variables in $\mathbf{w}$, remain to be optimized, while the row-normalization constraints in Eq.~(\ref{eq:unitmeasure}) distribute the remaining probability mass among the exempt outputs.

To represent the WEM-integrated formulation compactly, we collect three coefficient blocks.
Let $\mathcal{T}:=\{\textsf{dx},\textsf{dy},\textsf{cy}\}$ index these blocks, and denote them collectively by $\{\mathbf{A}^{t}\}_{t\in\mathcal{T}}$, where
\begin{align}
\mathbf{A}^{\textsf{dx}} &\in \mathbb{R}^{|\mathcal{X}|\times |\mathcal{X}|},
&[\mathbf{A}^{\textsf{dx}}]_{i,j} &:= d_{x_i,x_j}, \label{eq:AdXX_A}\\
\mathbf{A}^{\textsf{dy}} &\in \mathbb{R}^{|\mathcal{X}| \times |\mathcal{Y}|},
&[\mathbf{A}^{\textsf{dy}}]_{i,k} &:= d_{x_i,y_k}, \label{eq:AdXY_A}\\
\mathbf{A}^{\textsf{cy}} &\in \mathbb{R}^{|\mathcal{X}| \times |\mathcal{Y}|},
&[\mathbf{A}^{\textsf{cy}}]_{i,k} &:= \pi_{x_i}\, c_{x_i,y_k}. \label{eq:AcXY_A}
\end{align}
Here, $\textsf{dx}$ captures pairwise distances within $\mathcal{X}$ and determines how strongly different input rows are coupled by the mDP constraints. The block $\textsf{dy}$ records input-output distances and is used to define the near/far structure around each output $y_k$ as well as the exponential-decay terms in the WEM constraints. Finally, $\textsf{cy}$ encodes the prior-weighted utility loss of reporting $y_k$ when the true input is $x_i$, and therefore supplies the coefficients of the utility objective.

With WEM constraints, the feasible region is the convex polyhedron
\begin{eqnarray} 
&& \Omega\left({\rd\mathbf{A}^{\textsf{dx}}}, {\bl\mathbf{A}^{\textsf{dy}}}\right) \\ 
\nonumber &=& \footnotesize \left\{ \begin{array}{l} (\mathbf{Z}, \mathbf{w}) \\ 
\in [0,1]^{|\mathcal{X}|\times |\mathcal{Y}|}\times \mathbb{R}^{K}_{\geq 0} \end{array} \left| \begin{array}{ll} z_{i,k} - e^{\epsilon [{\rd\mathbf{A}^{\textsf{dx}}}]_{i,j}} z_{j,k} \le 0, \\ \forall (x_i, x_j, y_k) \in \mathcal{X}^2\times\mathcal{Y}, \\ \sum_{y_k \in \mathcal{Y}} z_{i,k} = 1, \forall x_i \in \mathcal{X}, \\ z_{i,k} - w_{k}e^{-\frac{1}{2}\epsilon [{\bl\mathbf{A}^{\textsf{dy}}}]_{i,k}} = 0, \\ \forall y_k \in \mathcal{Y}, \forall x_i \in \mathcal{I}^{\mathrm{c}}(y_k) \end{array}\right. \right\}. \end{eqnarray}

\smallskip
\noindent \textbf{LP+WEM generator as an operator.}
We view the LP-based mechanism generator with WEM integration (\textsc{LP+WEM}) as a functional operator
\begin{equation}
\mathcal{F}_{\mathrm{LP+WEM}}:
\big(\mathbf{A}^{t}\big)_{t\in\mathcal{T}}
\mapsto
(\mathbf{Z},\mathbf{w}) \in [0,1]^{|\mathcal{X}|\times |\mathcal{Y}|}\times \mathbb{R}_{\ge 0}^{K},
\end{equation}
where $(\mathbf{Z},\mathbf{w})$ is obtained by solving
\begin{eqnarray}
\label{eq:lp-function_A}
&&\mathcal{F}_{\mathrm{LP+WEM}}\left({\gr\mathbf{A}^{\textsf{cy}}}, {\rd\mathbf{A}^{\textsf{dx}}}, {\bl\mathbf{A}^{\textsf{dy}}}\right)
\\
&=& 
\arg\min_{(\mathbf{Z},\mathbf{w}) \in \Omega({\rd\mathbf{A}^{\textsf{dx}}}, {\bl\mathbf{A}^{\textsf{dy}}})}
\mathcal{L}\left(\mathbf{Z};{\gr\mathbf{A}^{\textsf{cy}}}\right).
\end{eqnarray}


\section{Framework Motivation and Challenges}
\label{sec:overview}
Notably, although WEM constraints reduce the effective search space by enforcing an exponential-decay form for many ``far'' transition probabilities, its computational cost remains prohibitive on large-scale domains. This blow-up is largely driven by the mDP requirement, which introduces distance-based indistinguishability constraints for (potentially) every pair of secrets in $\mathcal{X}$. Although these constraints are defined globally, their effect is inherently \emph{local}: as $d_{x_i,x_j}$ increases, the corresponding constraint becomes increasingly permissive, so distant records exert much less influence than nearby neighbors~\cite{ImolaUAI2022}. This locality motivates scalable LP-based methods that focus on enforcing indistinguishability within local neighborhoods, rather than explicitly instantiating constraints for all secret pairs. Hence, a promising direction is to apply a \emph{distributed} framework that offloads local structure to users.

Fig.~\ref{fig:motivatingframework} illustrates a distributed LP framework adopted in prior work (e.g.,~\cite{Qiu-PETS2025}). 
Each user $m$ constructs local coefficient blocks $\mathbf{A}_m^{\textsf{dx}}, \mathbf{A}_m^{\textsf{dy}}, \mathbf{A}_m^{\textsf{cy}}$ restricted to a neighborhood $\mathcal{N}_m$ around the true record $x_m$ (e.g., the 20 nearest secrets to $x_m$), rather than spanning the entire secret domain $\mathcal{X}$. 
The server aggregates these blocks across users and solves the resulting integrated LP of the following form (the full formulation is deferred to Appendix \ref{app:algorithm_design} for completeness):
\begin{eqnarray} 
\min && \mathcal{L}(\mathbf{Z};\mathbf{A}_1^{\textsf{cy}}, ..., \mathbf{A}_M^{\textsf{cy}}) \\
\mathrm{s.t.} && \mathbf{Z} \in \Omega\left(\mathbf{A}_1^{\textsf{dx}},..., \mathbf{A}_M^{\textsf{dx}}, \mathbf{A}_1^{\textsf{dy}}, ...,\mathbf{A}_M^{\textsf{dy}}\right). 
\end{eqnarray}
This integrated program retains the constraint template of the centralized LP
$\mathcal{F}_{\mathrm{LP+WEM}}\!\left((\mathbf{A}^{t})_{t\in\mathcal{T}}\right)$ as described in Eq. (\ref{eq:lp-function_A}), but restricts it to pairs within each local neighborhood $\mathcal{N}_m$, resulting in a much smaller optimization problem.
Accordingly, the number of decision variables drops from $\mathcal{O}(|\mathcal{X}||\mathcal{Y}|)$ to
$\mathcal{O}\!\left(\sum_{m=1}^M |\mathcal{N}_m||\mathcal{Y}|\right)$, and the number of constraints from
$\mathcal{O}(|\mathcal{X}|^2|\mathcal{Y}|)$ to
$\mathcal{O}\!\left(\sum_{m=1}^M |\mathcal{N}_m|^2|\mathcal{Y}|\right)$,
providing substantial scalability gains when $\sum_{m=1}^M |\mathcal{N}_m| \ll |\mathcal{X}|$.

\begin{figure*}[t]
\centering
\begin{minipage}{1.00\textwidth}
\centering
  \subfigure{
\includegraphics[width=0.91\textwidth]{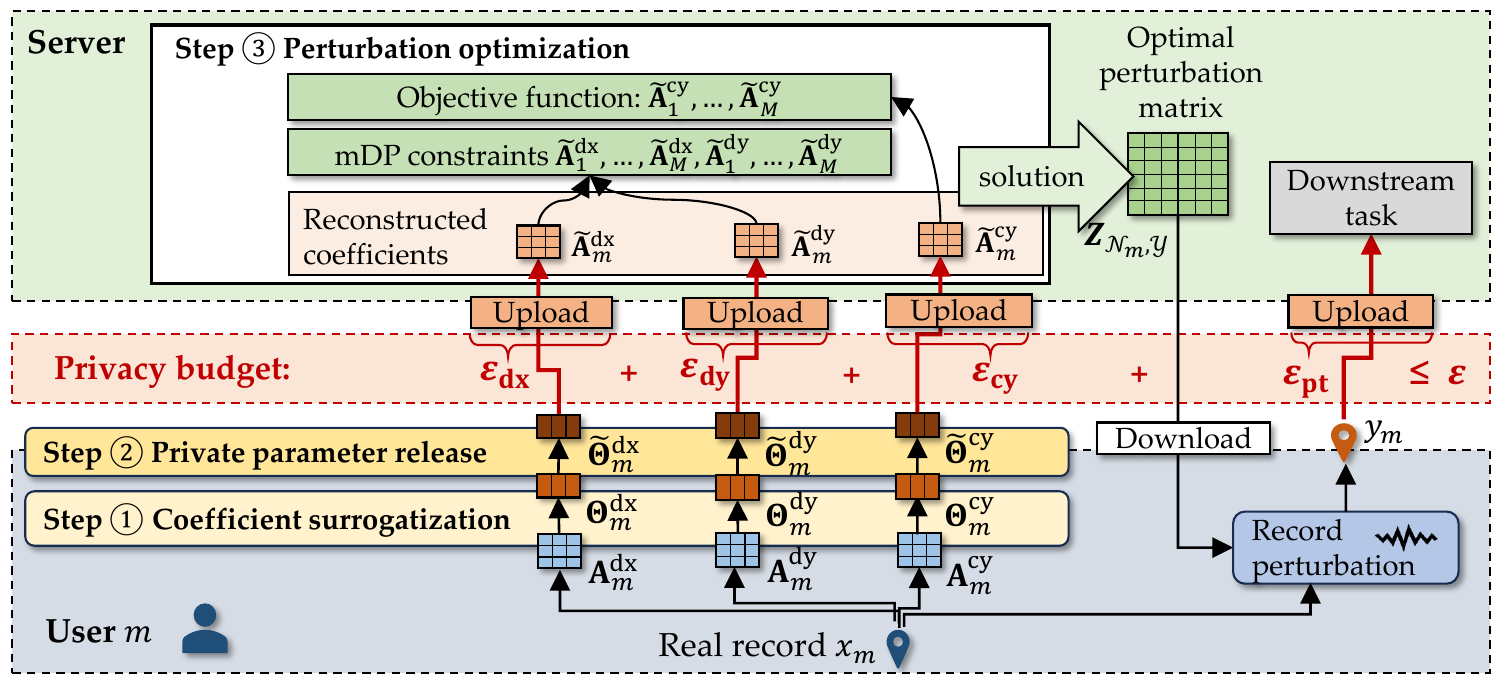}}
\caption{The DistA-LP framework. }
\label{fig:framework}
\end{minipage}
\end{figure*}

However, the distributed design in~\cite{Qiu-PETS2025} implicitly treats the upload of local LP coefficient blocks $\mathbf{A}_m^{\textsf{dx}}, \mathbf{A}_m^{\textsf{dy}}, \mathbf{A}_m^{\textsf{cy}}$ as privacy-neutral. This assumption is problematic because these blocks are deterministic functions of users' authentic records and can themselves reveal sensitive information about the underlying data. To formalize this, we define the \emph{coefficient-upload mechanism} $\mathcal{M}_{\mathrm{c}}^{t}$ as the mapping from a true record $x\in\mathcal{X}$ to the uploaded coefficient block for type $t\in\mathcal{T}$: 
\begin{equation}
\mathcal{M}_{\mathrm{c}}^{t}(x_m) \triangleq \mathbf{A}^{t}_m.
\end{equation}
If no randomization is applied, then $\mathcal{M}_{\mathrm{c}}^{t}$ is deterministic. Consequently, for two distinct records $x_i\neq x_j$, there exists an uploaded observation $\widetilde{\mathbf{A}}$ (e.g., $\widetilde{\mathbf{A}}=\mathbf{A}^{t}_i)$) such that
\begin{equation}
\Pr\left[\mathcal{M}_{\mathrm{c}}^{t}(x_i)=\widetilde{\mathbf{A}}\right]=1
\quad\text{while}\quad
\Pr\left[\mathcal{M}_{\mathrm{c}}^{t}(x_j)=\widetilde{\mathbf{A}}\right]=0.
\end{equation}
The corresponding privacy loss for observing $\widetilde{\mathbf{A}}$, 
\begin{equation}
\mathrm{PL}_{x_i, x_j}\left(\widetilde{\mathbf{A}}\right) = \frac{1}{d_{x_i, x_j}}\ln\left(\frac{\Pr[\mathcal{M}_{\mathrm{c}}^{t}(x_i)=\widetilde{\mathbf{A}}]}{\Pr[\mathcal{M}_{\mathrm{c}}^{t}(x_j)=\widetilde{\mathbf{A}}]}\right), 
\end{equation}
is infinite, showing that uploading coefficient blocks without randomization can cause unbounded leakage and thus requires explicit protection at the upload stage.


Notably, simply adding naive noise to the coefficient blocks $\mathbf{A}_m^{t}$
(for $t\in\{\textsf{dx},\textsf{dy},\textsf{cy}\}$) remains insufficient.
As a baseline, we perturb each block entrywise with Laplace noise
$\mathrm{Lap}\!\left(0,\frac{\Delta_t}{\epsilon}\right)$ (with $\epsilon=4$), where $\Delta_t$
is the metric sensitivity of $\mathbf{A}_m^{t}$ under $d(\cdot,\cdot)$ (see Appendix~\ref{app:laplace_mdp}).
We then measure the induced privacy loss $\mathrm{PL}_{x_i,x_j}\!\left(\widetilde{\mathbf{A}}\right)$.
Figure~\ref{fig:motivatingexample}(a)--(c) shows that a substantial fraction of the mass still lies \emph{above} the target budget.



Motivated by these examples, we propose the \textsc{DistA-LP} framework to explicitly account for the privacy loss incurred by coefficient uploads and to allocate a dedicated budget for correct end-to-end composition, raising several unique challenges to address:

\smallskip
\noindent \textbf{Challenge~1: \emph{How can we reduce the privacy loss incurred by uploading local LP coefficient blocks?}}
As illustrated in Fig.~\ref{fig:motivatingexample}, directly uploading coefficient blocks (even with naive perturbation) can leak highly sensitive information about a user’s locally relevant neighborhood $\mathcal{N}_m$ and, by extension, the underlying record. A key reason is that $\mathbf{A}_m^{t}$ is high-dimensional and encodes fine-grained, instance-specific neighborhood structure. To mitigate this leakage, we replace each high-dimensional block $\mathbf{A}_m^{t}$ with a low-dimensional surrogate representation and then apply a parameter-noising mechanism to the surrogate parameters $\widetilde{\boldsymbol{\Theta}}_m^{t}$, aiming to suppress fine-grained, instance-specific structure and reduce record-level distinguishability.

\smallskip
\noindent \textbf{Challenge~2: \emph{How can we minimize the impact of surrogate approximation on both mDP feasibility and utility?}}
Approximating $\mathbf{A}_m^{\textsf{dx}}, \mathbf{A}_m^{\textsf{dy}}, \mathbf{A}_m^{\textsf{cy}}$ with low-dimensional surrogates inevitably introduces reconstruction error, which can affect both the feasibility of mDP constraints and the utility of the LP solution produced by $\mathcal{F}_{\mathrm{LP}}$.
In particular, distance overestimation can inflate terms such as $e^{\epsilon d}$, enlarge the feasible region, and yield mechanisms that violate the intended mDP constraints.
Therefore, we must design calibrated asymmetric perturbation and \emph{conservative} reconstruction procedures so that the recovered coefficients remain mDP-preserving while retaining high utility.

\smallskip
\noindent \textbf{Challenge~3: \emph{How can we evaluate and account for the privacy loss of releasing continuous surrogate parameters?}}
The released surrogate parameters $\widetilde{\boldsymbol{\Theta}}_m^{t}$ lie in continuous spaces. For the corresponding parameter-noising mechanisms, deriving closed-form privacy-loss bounds or privacy-loss distributions is generally difficult, especially under asymmetric, truncated, and otherwise nonstandard perturbations. Hence, we need practical evaluation and accounting methods to quantify the induced privacy loss and incorporate it into end-to-end budget allocation and composition.



Next, we present our method and explain how it addresses each of the three challenges.

\section{Methodology}
\label{sec:method}

In this section, we detail the design of DistA-LP and explain how it addresses \textbf{Challenges~1--3} in Section~\ref{sec:overview}. As Fig.~\ref{fig:framework} shows, DistA-LP decomposes end-to-end mechanism generation into \emph{three coordinated steps}: \looseness=-1

\begin{itemize}
    \item [\textcircled{1}] \textbf{Coefficient surrogatization}. Each user $m$ locally constructs three coefficient blocks: the \emph{intra-neighborhood distance matrix} $\mathbf{A}_m^{\textsf{dx}}$, the \emph{neighborhood-output distance matrix} $\mathbf{A}_m^{\textsf{dy}}$, and the \emph{utility-loss matrix} $\mathbf{A}_m^{\textsf{cy}}$. To reduce information leakage, user $m$ fits a compact \emph{coefficient surrogate} $f(\cdot;\boldsymbol{\Theta}_m^t)$ (parameterized by $\boldsymbol{\Theta}_m^t$) to approximate $\mathbf{A}_m^t$, i.e., $[\mathbf{A}_m^t]_{i,j}\approx f(i,j;\boldsymbol{\Theta}_m^t), ~ t \in \{\textsf{dx},\textsf{dy},\textsf{cy}\}$. 
    \smallskip
    \item [\textcircled{2}] \textbf{Private parameter release}. The user privatizes the surrogate parameters using a local randomizer $\mathcal{M}_t^{\textsf{p}}$, $\widetilde{\boldsymbol{\Theta}}_m^t=\mathcal{M}_t^{\textsf{p}}(\boldsymbol{\Theta}_m^t)$, and uploads $\widetilde{\boldsymbol{\Theta}}_m^t$ ($t \in \{\textsf{dx},\textsf{dy},\textsf{cy}\}$) to the server. 
    \smallskip
    \item [\textcircled{3}] \textbf{Perturbation Optimization}. The server reconstructs each approximate coefficient block $\widetilde{\mathbf{A}}_m^{t}$ from the noised surrogate parameters via $[\widetilde{\mathbf{A}}_m^{t}]_{i,j}=f(i,j;\widetilde{\boldsymbol{\Theta}}_m^{t})$, aggregates the reconstructed blocks, and solves the resulting LP to obtain each optimal local perturbation matrix $\mathbf{Z}^{*}_{\mathcal{N}_m\times\mathcal{Y}}$. 
\end{itemize}
Finally, each user downloads the resulting perturbation matrix and applies it to map its true record $x_m$ to a randomized output $y_m$. The server then uses the perturbed records for downstream analytical tasks. 

\smallskip
\noindent \textbf{Privacy accounting across uploads.}
As shown in Fig.~\ref{fig:framework}, \textsc{DistA-LP} involves \emph{four} user-to-server releases: three uploads of the privatized surrogate parameters $\widetilde{\boldsymbol{\Theta}}^{\textsf{dx}}_m$, $\widetilde{\boldsymbol{\Theta}}^{\textsf{dy}}_m$, and $\widetilde{\boldsymbol{\Theta}}^{\textsf{cy}}_m$ in Step~\textcircled{2}, followed by a final upload of the perturbed report $y_m$ after applying the downloaded perturbation matrix. Since each release consumes privacy budget, we assign per-release budgets $\epsilon_{\textsf{dx}}, \epsilon_{\textsf{dy}}, \epsilon_{\textsf{cy}},$ and $\epsilon_{\textsf{pt}}$ to $\widetilde{\boldsymbol{\Theta}}^{\textsf{dx}}_m$, $\widetilde{\boldsymbol{\Theta}}^{\textsf{dy}}_m$, $\widetilde{\boldsymbol{\Theta}}^{\textsf{cy}}_m$, and $y_m$, respectively, 
and enforce the end-to-end constraint
\begin{equation}
\label{eqn:privacy_accounting}
\epsilon_{\textsf{dx}} + \epsilon_{\textsf{dy}} + \epsilon_{\textsf{cy}} + \epsilon_{\textsf{pt}} \le \epsilon,
\end{equation}
where $\epsilon$ is the predetermined overall privacy threshold. Since the server-side reconstruction and LP solving in Step~\textcircled{3} are pure post-processing of the uploaded privatized parameters, they incur no additional privacy cost; therefore, tracking and composing these four budgets suffices to guarantee the end-to-end privacy level of the full protocol.

\smallskip
Next, we present the details of Steps~\textcircled{1}--\textcircled{3} in Sections~\ref{subsec:step1}--\ref{subsec:step3}, respectively.

\subsection{Step \textcircled{1}: Coefficient Surrogatization}
\label{subsec:step1}
In this step, we approximate each coefficient block $\mathbf{A}_m^{t}$ $(t\in\{\textsf{dx}, \textsf{dy}, \textsf{cy}\})$ using a compact parametric surrogate $f(\cdot;\boldsymbol{\Theta}_m^{t})$, instantiated either as a \emph{predefined functional form} or a \emph{low-rank SVD} approximation, with parameters $\boldsymbol{\Theta}_m^{t}$.
Restricting $\mathbf{A}_m^{t}$ to a low-dimensional model class suppresses fine-grained, instance-specific variations within $\mathcal{N}_m$ that could otherwise be revealing. In addition, the mapping $\mathbf{A}_m^{t}\mapsto \boldsymbol{\Theta}_m^{t}$ is typically many-to-one, so multiple distinct blocks can induce the same (or similar) parameter vector; as a result, $\boldsymbol{\Theta}_m^{t}$ is generally less distinguishable than $\mathbf{A}_m^{t}$ itself, increasing an adversary's uncertainty about the underlying coefficients.
We next detail the two approximation families: \emph{predefined parametric functions} and \emph{SVD-based} surrogates.


\begin{figure}[t]
\centering
\begin{minipage}{0.50\textwidth}
\centering
  \subfigure{
\includegraphics[width=1.00\textwidth]{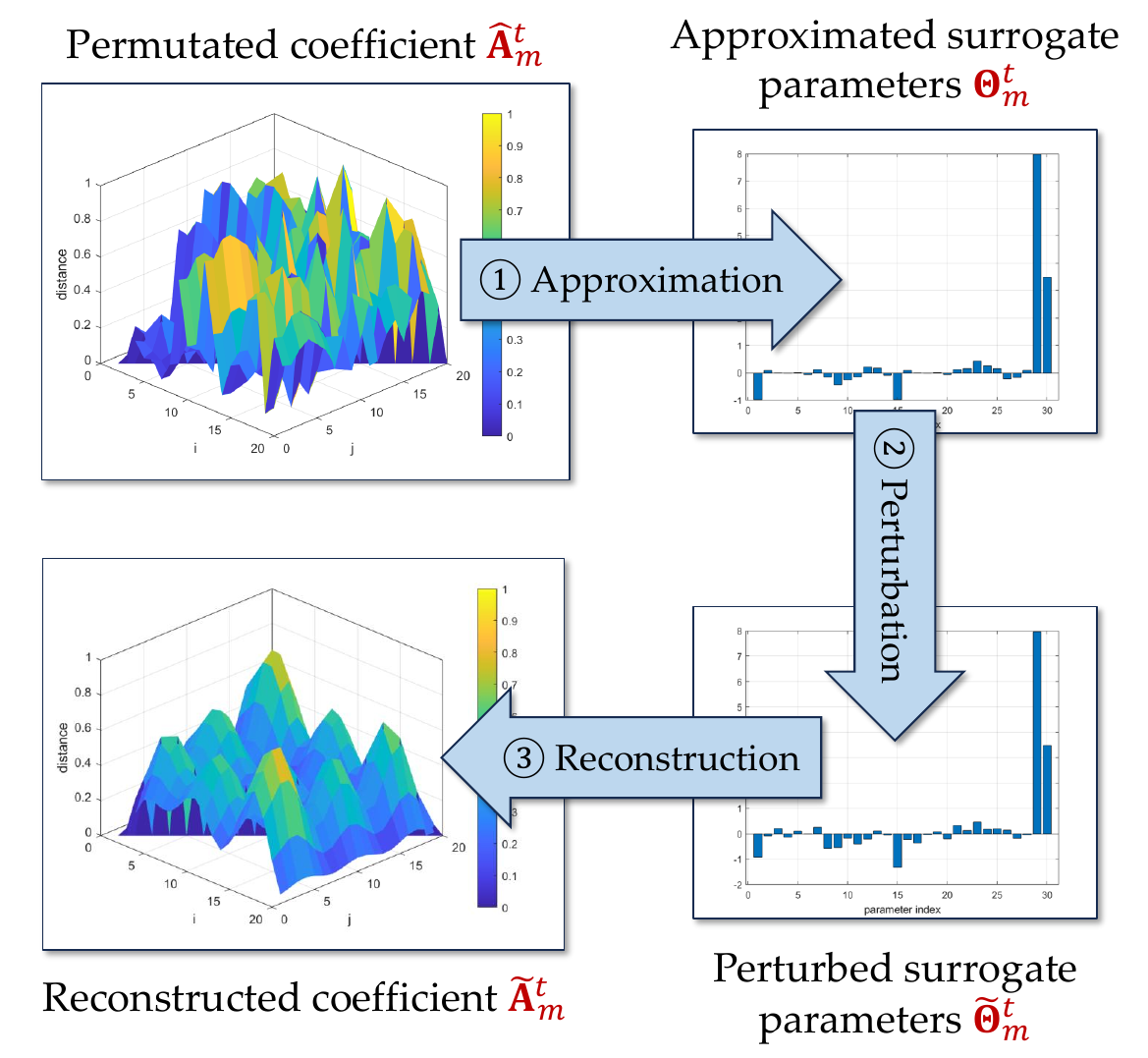}}
\caption{Surrogate-based coefficient privatization pipeline (permutation, approximation, parameter perturbation, and reconstruction): 
$\mathbf{A}_m^t \rightarrow \boldsymbol{\Theta}_m^t \rightarrow \widetilde{\boldsymbol{\Theta}}_m^t \rightarrow \widetilde{\mathbf{A}}_m^t$.}
\label{fig:3stepillustration}
\end{minipage}
\vspace{-0.00in}
\end{figure}

\smallskip
\noindent \textbf{Local permutation.}
To reduce the similarity among coefficient blocks across users, we apply a \emph{local permutation} to the ordering of locations within each neighborhood $\mathcal{N}_m$.

Let $\mathbf{E}_m\in\{0,1\}^{|\mathcal{N}_m|\times|\mathcal{N}_m|}$ be the corresponding permutation matrix, with exactly one $1$ in each row and each column (and $0$ elsewhere); hence
$\mathbf{E}_m^\top \mathbf{E}_m=\mathbf{E}_m\mathbf{E}_m^\top=\mathbf{I}$ and $\mathbf{E}_m^{-1}=\mathbf{E}_m^\top$.
The permutation for each user can be generated via a deterministic procedure (e.g., the optimization in Eq.~(17)), so the coordinator may know (or be able to reproduce) the set $\{\mathbf{E}_m\}_{m=1}^M$.
However, surrogate/coefficient uploads are received \emph{without user identifiers} in the anonymized upload phase; thus, upon receiving a permuted block, the coordinator cannot determine which $\mathbf{E}_m$ was applied and cannot undo the permutation.

\begin{table*}[t]
\caption{Comparison of different matrix approximation methods.}
\centering
\small
\begin{tabular}{p{2.1cm}p{1.6cm}p{1.2cm}p{4.0cm}p{6.8cm}}
\toprule
\rowcolor{lightgray!50} \textbf{Surrogate} &  \textbf{Size of} $\boldsymbol\Theta_m^t$ & \textbf{Accuracy} & \textbf{Notes} & \textbf{Noise Strategy} \\
\hline
2D Gaussian & 6 & Low & Compact, but only fits symmetric bell-shaped data well. & Subtract Laplace noise from amplitude ($A$) and increase spread ($\sigma$) to lower the surface everywhere. \\
\hline
2D Polynomial & 10 & Medium & Simple to fit; more effective for smooth surface trends. & Subtract Laplace noise from coefficients, especially the intercept, to bias the surface downward. \\
\hline
RBF & 25 & High & Flexible and localized, better for multi-peak or irregular matrices. & Add negative-biased noise to weights ($\alpha_i - |\text{Lap}(b)|$); optionally perturb centers for robustness. \\
\hline
Low-Rank SVD & 30--50 & Very High & Captures global structure with fewer components; efficient to store and reconstruct. & Subtract noise from singular values ($S_k$) or from the reconstructed matrix directly to ensure conservative output. \\
\bottomrule
\end{tabular}
\label{tab:approx_privacy}
\end{table*}

Notably, $\mathbf{E}_m$ permutes only the \emph{neighborhood index} in $\mathcal{N}_m$ and leaves the output index set $\mathcal{Y}$ unchanged.
Given $\mathbf{E}_m$, we permute each coefficient block $\mathbf{A}_m^{t}$ as
\begin{equation}
\widehat{\mathbf{A}}_m^{t}=
\begin{cases}
\mathbf{E}_m \mathbf{A}_m^{t} \mathbf{E}_m^{\top}, & t=\textsf{dx},\\
\mathbf{E}_m \mathbf{A}_m^{t}, & t\in\{\textsf{dy},\textsf{cy}\},
\end{cases}
\label{eq:permute_blocks}
\end{equation}
i.e., for the square block ($t=\textsf{dx}$) indexed by $\mathcal{N}_m\times\mathcal{N}_m$, we permute both neighborhood axes, while for the rectangular blocks ($t\in\{\textsf{dy},\textsf{cy}\}$) indexed by $\mathcal{N}_m\times\mathcal{Y}$, we permute only the neighborhood axis.

We do not permute $\mathcal{Y}$ because the server needs to formulate a single optimization problem in a \emph{common} output index space, so that decision variables associated with $y\in\mathcal{Y}$ are aligned across users; otherwise, cross-user consistency of output-indexed variables would be ill-defined (see Step~\textcircled{3} in Section~\ref{subsec:step3}).
After user $m$ downloads the optimized local perturbation matrix (indexed by the permuted neighborhood order), it applies the inverse permutation locally to recover the original indexing (i.e., left-multiplies by $\mathbf{E}_m^{-1}=\mathbf{E}_m^\top$).

\smallskip
\noindent\textbf{Surrogate $f(\cdot;\boldsymbol\Theta_m^t)$ derivation.} Each user then approximates their permuted blocks $\widehat{\mathbf{A}}_m^t$ ($t\in\{\textsf{dx}, \textsf{dy}, \textsf{cy}\}$) using the parametric surrogate
$f(\cdot;\boldsymbol\Theta_m^t)$:
\begin{equation}
\widehat{\mathbf{A}}_m^t \approx f(\cdot;\boldsymbol\Theta_m^t),
\qquad \forall m\in[M],\forall t\in\mathcal{T}.
\label{eq:block_surrogate}
\end{equation}
Here, we instantiate $f(\cdot;\boldsymbol{\Theta}_m^{t})$ using four parametric families that trade expressiveness for parameter count (Table~\ref{tab:approx_privacy}).
We jointly determine the permutations $\{\mathbf{E}_m\}_{m=1}^M$ and surrogate parameters
$\{\boldsymbol{\Theta}_m^{t}\}_{m\in[M],\,t\in\mathcal{T}}$ by minimizing the total entrywise $\ell_1$ fitting error:
\begin{equation}
\min_{\{\mathbf{E}_m,\boldsymbol{\Theta}_m^{t}\}_{m=1}^M}
\quad \sum_{m=1}^M \sum_{t\in\mathcal{T}} \lambda_t
\left\|\widehat{\mathbf{A}}_m^{t} - f(\cdot;\boldsymbol{\Theta}_m^{t})\right\|_1,
\label{eq:block_fit_obj}
\end{equation}
where $\lambda_t\ge 0$ controls the trade-off among approximation components, and we fix $\lambda_{\textsf{dx}}=1$ without loss of generality.
Here, we prefer $\ell_1$ over $\ell_2$ for robustness to outliers and discretization noise, which typically yields more stable fitting.

\smallskip
\noindent\textbf{Joint optimization of permutation $\{\mathbf{E}_m,\boldsymbol{\Theta}_m^{t}\}_{m=1}^M$.} We adopt an \emph{alternating scheme} between (1) fitting parameters for a fixed permutation and (2) improving the  permutation through greedy swaps:
\begin{itemize}
\item [(1)] \emph{Parameter fitting.}
Holding $\{\mathbf{E}_m\}_{m=1}^M$ fixed, for each surrogate family we fit $\boldsymbol{\Theta}_m^t$ by minimizing the weighted entrywise $\ell_1$ reconstruction loss in~\eqref{eq:block_fit_obj} using the \emph{Nelder--Mead simplex algorithm} \cite{NelderMead1965}, a robust, derivative-free optimizer well-suited for nonconvex objectives and nonsmooth losses. 


\item [(2)] \emph{Permutation search.} With $\{\boldsymbol{\Theta}_m^t\}_{m=1}^M$ fixed, we refine the local permutation $\mathbf{E}_m$ via greedy pairwise index swaps $(i,j)$ that decrease the global objective in~\eqref{eq:block_fit_obj}.
We accept only swaps that reduce the total weighted $\ell_1$ loss, ensuring monotonic descent. After each accepted swap, we refit $\{\boldsymbol{\Theta}_m^t\}_{m=1}^M$ under the updated ordering for consistency. According to the experimental results in Fig.~\ref{fig:surrogate_iteration_10}–\ref{fig:surrogate_iteration_20} this alternating scheme converges to a locally optimal configuration in practice and typically stabilizes within 10-20 iterations.
\end{itemize}



\DEL{
\smallskip
We adopt an alternating scheme between fitting parameters for a fixed permutation and improving the  permutation through greedy swaps:
\begin{itemize}
    \item \emph{Parameter fit.} Holding $\mathbf{E}_m$ fixed, for each family we minimize the weighted $\ell_1$ loss above over $\boldsymbol \Theta_m^t $ using the \textit{Nelder-Mead simplex algorithm} \CQ{Michael: add the reference for this algorithm}, a robust, derivative-free optimizer well-suited for nonconvex objectives and discontinuous gradients. This choice circumvents instabilities encountered by gradient-based methods under discretization and privacy-induced perturbations. For the Gaussian family, an implicit box constraint ensures physically meaningful spreads $(s_x, s_y)$ and nonnegative amplitudes $(A \ge 0)$, thereby avoiding degeneracy. {\rd The $\ell_1$ objective is preferred over $\ell_2$ to enhance robustness against outliers and discretization noise, promoting stable parameter recovery.} For the polynomial family, we apply mild $\ell_2$ regularization to the coefficient vector to prevent overfitting and improve conditioning, particularly under small-sample or high-noise regimes. For the RBF family, the bandwidth $\sigma$ is held within a limited range to maintain numerical stability, while the basis weights $\alpha_k$ are softly regularized to prevent over-concentration and preserve smoothness across neighboring centers. The $\ell_1$ objective is preferred over $\ell_2$ to enhance robustness against outliers and discretization noise, promoting stable and interpretable parameter recovery
    \item \emph{Permutation search.} With $\boldsymbol \Theta_m^t$ fixed, we refine the local permutation $\pi_m$ through pairwise index swaps $(i, j)$ that decrease the global objective. Each swap is evaluated in terms of its contribution to the total $\ell_1$ reconstruction loss, and only beneficial exchanges are accepted, ensuring monotonic descent. After each accepted swap, parameters $\boldsymbol \Theta_m^t$ are refitted to the updated data ordering for consistency. This alternating optimization between parameter fitting and permutation refinement converges to a locally optimal configuration that jointly minimizes approximation error and privacy loss. The greedy swap mechanism, though simple, is computationally efficient and empirically converges in a few iterations, balancing exploration of the permutation space with stability of the fitted surrogate.
\end{itemize}}

\DEL{\subsubsection{Low-Rank SVD}
We approximate each coefficient matrix $\mathbf{A}\in\mathbb{R}^{m\times n}$ with a rank-$k$ surrogate that preserves its dominant structure while drastically reducing dimensionality. Compute the singular value decomposition (SVD) $\mathbf{A}=\mathbf{U}\boldsymbol{\Sigma}\mathbf{V}^{\top}$, where $\mathbf{U}\in\mathbb{R}^{m\times m}$ and $\mathbf{V}\in\mathbb{R}^{n\times n}$ are orthogonal, and $\boldsymbol{\Sigma}=\mathrm{diag}(\sigma_1,\ldots,\sigma_r)$ with $\sigma_1\ge\cdots\ge\sigma_r>0$ and $r=\mathrm{rank}(\mathbf{A})$. The truncated SVD retains the top-$k$ components,
\begin{equation}
\mathbf{A}_k = \mathbf{U}_k\boldsymbol{\Sigma}_k\mathbf{V}_k^{\top},\qquad k\ll\min\{m,n\},
\end{equation}
where $\mathbf{U}_k\in\mathbb{R}^{m\times k}$, $\mathbf{V}_k\in\mathbb{R}^{n\times k}$, and $\boldsymbol{\Sigma}_k=\mathrm{diag}(\sigma_1,\ldots,\sigma_k)$. By Eckart-Young optimality, $\mathbf{A}_k$ minimizes both spectral and Frobenius error among all rank-$k$ matrices. For compact storage and upload, we use the factorized parameterization
\begin{equation}
\boldsymbol{\Theta}=(\mathbf{P},\mathbf{Q}),\quad \mathbf{P}:=\mathbf{U}_k\boldsymbol{\Sigma}_k\in\mathbb{R}^{m\times k},\quad \mathbf{Q}:=\mathbf{V}_k\in\mathbb{R}^{n\times k},
\end{equation}
so that $\mathbf{A}_k=\mathbf{P}\mathbf{Q}^{\top}$. This reduces the representation from $mn$ entries to $k(m{+}n)$ parameters and concentrates subsequent privacy noise on a small set of variables.

\medskip
\noindent\emph{Error guarantees and rank selection.}
By the Eckart--Young theorem, $\mathbf{A}_k$ is the best rank-$k$ approximation to $\mathbf{A}$ in both spectral and Frobenius norms:
\begin{equation}
\|\mathbf{A}-\mathbf{A}_k\|_F^2 = \sum_{i>k}\sigma_i^2, ~\|\mathbf{A}-\mathbf{A}_k\|_2 = \sigma_{k+1}.
\end{equation}
A common choice sets $k$ to capture a target energy fraction,
$\sum_{i=1}^{k}\sigma_i^2 \big/ \sum_{i=1}^{r}\sigma_i^2 \ge \tau$ (e.g., $\tau=0.95$),
trading a controlled approximation bias for substantial compression.
When privacy noise is added, the final approximation error satisfies
\begin{equation}
\|\mathbf{A}-\widetilde{\mathbf{A}}\| \le \underbrace{\|\mathbf{A}-\mathbf{A}_k\|}_{\text{truncation}} + \underbrace{\|\mathbf{A}_k-\widetilde{\mathbf{P}}\widetilde{\mathbf{Q}}^\top\|}_{\text{noise \& reconstruction}},
\end{equation}
making it explicit how $k$ (bias) and noise scale (variance) jointly determine fidelity.

Both representations reduce dimensionality (from $mn$ entries to $\dim(\boldsymbol\Theta)$ or roughly $k(m+n)$), cut storage and communication, and concentrate noise on a small parameter set, yielding scalability and inherent denoising at the utility loss of a controlled approximation bias.}

Table~\ref{tab:approx_privacy} compares surrogate parametric families in terms of parameter complexity, approximation accuracy, and the noise strategies used to ensure conservative, privacy-preserving outputs. Overall, 2D Gaussian and polynomial surrogates offer compact representations but are best suited to smooth, globally structured surfaces (and, in the Gaussian case, near-symmetric shapes), whereas RBFs provide greater expressiveness at the cost of substantially larger parameter sets. In contrast, the low-rank SVD surrogate provides the most favorable balance, preserving global structure with high fidelity while keeping the representation low-dimensional.

\subsection{Step \textcircled{2}: Private Parameter Release} 
\label{subsec:step2}
To preserve privacy in our distributed LP, we introduce three randomized \emph{parameter-perturbation mechanisms}
$\mathcal{M}_t^{\textsf{p}}$ ($t\in \{\textsf{dx},\textsf{dy},\textsf{cy}\}$) that inject calibrated and asymmetric noise into the uploaded
\emph{surrogate parameters}:
\begin{equation}
\label{eq:block_mech}
\widetilde{\boldsymbol{\Theta}}_m^{t} = \mathcal{M}_t^{\textsf{p}}\left(\boldsymbol{\Theta}_m^{t}\right),
\qquad \forall m\in[M],\forall t\in\mathcal{T}.
\end{equation}
Here, each $\mathcal{M}_t^{\textsf{p}}$ is a randomized mapping from the surrogate-parameter space to itself,
\begin{equation}
\mathcal{M}_t^{\textsf{p}}:\mathbb{R}^{d_t}\rightarrow \mathbb{R}^{d_t},
\end{equation}
where $d_t$ is the number of parameters in $\boldsymbol{\Theta}_m^{t}$ for block type $t$.
This is analogous to the perturbation mechanisms defined earlier for raw records, but with a different domain:
instead of mapping secrets $x\in\mathcal{X}$ to perturbed outputs $y\in\mathcal{Y}$, $\mathcal{M}_t^{\textsf{p}}$ maps
\emph{surrogate coefficients} $\boldsymbol{\Theta}_m^{t}$ to \emph{noised coefficients}
$\widetilde{\boldsymbol{\Theta}}_m^{t}$ prior to upload.

\DEL{
\smallskip
\noindent\textbf{(1) Mechanism $\mathcal{M}^{\textsf{p}}_{\textsf{dx}}$ (intra-neighborhood distance parameters).}
$\mathcal{M}^{\textsf{p}}_{\textsf{dx}}$ perturbs the surrogate parameters for the intra-neighborhood distance block $\boldsymbol{\Theta}_m^{\textsf{dx}}$.
It adds Laplace noise with a calibrated scale (and, when needed, uses asymmetric noise) 
$\widetilde{\boldsymbol{\Theta}}_m^{\textsf{dx}}=\boldsymbol{\Theta}_m^{\textsf{dx}} + |\eta|, \eta\sim \mathrm{Lap}\left(0, b \right)$, \CQ{Michael, can you double check whether it should be $\boldsymbol{\Theta}_m^{\textsf{dx}} - |\eta|$}
and truncates the perturbed values to a bounded range to ensure numerical stability and bounded sensitivity \CQ{E.g., distance should be positive?}. 

\smallskip
\noindent\textbf{(2) Mechanism $\mathcal{M}^{\textsf{p}}_{\textsf{dy}}$ (neighborhood-output distance parameters).}
$\mathcal{M}^{\textsf{p}}_{\textsf{dy}}$ perturbs the surrogate parameters for the cross-domain distance block $\boldsymbol{\Theta}_m^{\textsf{dy}}$, which links each local neighborhood to the shared output space $\mathcal{Y}$.
By injecting structured noise 
$\widetilde{\boldsymbol{\Theta}}_m^{\textsf{dy}}=\boldsymbol{\Theta}_m^{\textsf{dy}} + |\eta|, \eta\sim \mathrm{Lap}\left(0, b \right),$
into these parameters, $\mathcal{M}^{\textsf{p}}_{\textsf{dy}}$ reduces the risk that the server (or other users) can infer sensitive cross-domain correlations, while retaining sufficient geometric structure for the subsequent collaborative optimization.}

\smallskip
\noindent\textbf{(1)(2) Mechanisms $\mathcal{M}^{\textsf{p}}_{\textsf{dx}}$ and $\mathcal{M}^{\textsf{p}}_{\textsf{dy}}$.}
Mechanism $\mathcal{M}^{\textsf{p}}_{\textsf{dx}}$ perturbs the surrogate parameters $\boldsymbol{\Theta}_m^{\textsf{dx}}$ for the intra-neighborhood distance block.
Mechanism $\mathcal{M}^{\textsf{p}}_{\textsf{dy}}$ perturbs $\boldsymbol{\Theta}_m^{\textsf{dy}}$ for the cross-domain distance block, which links each local neighborhood to the shared output space $\mathcal{Y}$.
Both $\mathcal{M}^{\textsf{p}}_{\textsf{dx}}$ and $\mathcal{M}^{\textsf{p}}_{\textsf{dy}}$ add calibrated Laplace noise with an asymmetric form,
\[
\widetilde{\boldsymbol{\Theta}}_m^{\textsf{dx}}
= \boldsymbol{\Theta}_m^{\textsf{dx}} - |\eta|,\quad
\eta \sim \mathrm{Lap}(0,b),
\]
and truncate the perturbed values to a bounded range to ensure numerical stability and bounded sensitivity.  
Notably, the one-sided noise is intended to bias the perturbation toward \emph{reducing} the distance values, thereby tightening (rather than relaxing) the corresponding mDP constraints and avoiding infeasible solutions that could cause mDP violations.


\smallskip
\noindent\textbf{(3) Mechanism $\mathcal{M}^{\textsf{p}}_{\textsf{cy}}$ (utility-loss parameters).}
$\mathcal{M}^{\textsf{p}}_{\textsf{cy}}$ perturbs the surrogate parameters for the utility-loss block $\boldsymbol{\Theta}_m^{\textsf{cy}}$ used in the LP objective.
This prevents the released objective coefficients from inadvertently revealing sensitive information encoded in the utility-loss structure (e.g., user-specific preferences or constraints), while keeping the perturbed coefficients well-behaved so that the global LP remains feasible and numerically stable in practice, we use 
\begin{equation*}
\tilde{c}_{i,k}=c_{i,k} + |\eta_c|, \qquad \eta_c\sim \mathrm{Lap}\left(0, {b_c} \right).
\end{equation*}


The detailed description of perturbation mechanisms for both \emph{predefined parametric surrogates} and \emph{SVD-based approximations} can be found in Appendix \ref{app:paraperturbation}.

\paragraph{Privacy Loss Evaluation.} 
According to Definition \ref{def:PL}, given the observation of perturbed coefficient surrogate $\tilde{\boldsymbol\Theta}^t_m$, the \emph{privacy loss (PL)} of any two records $x_i, x_j\in \mathcal{X}$ is defined by 
\vspace{-0.15in}
\normalsize
\begin{equation}
\label{eq:privacyloss_para}
\mathrm{PL}_{x_i, x_j}(\tilde{\boldsymbol\Theta}^t_m)
= \frac{1}{d_{x_i, x_j}}
\ln \left( \frac{\Pr\left[\mathcal{M}^{\textsf{p}}_{t}(x_i)=\tilde{\boldsymbol\Theta}^t_m\right]}{\Pr\left[\mathcal{M}^{\textsf{p}}_{t}(x_j)=\tilde{\boldsymbol\Theta}^t_m\right]}\right).
\end{equation}
\normalsize
Note that $\tilde{\boldsymbol{\Theta}}^t_m$ is defined over a continuous domain. 
Because we deliberately use non-standard (\emph{asymmetric} and \emph{truncated}) parameter noise to reduce the violation rate, 
the distribution of $\mathrm{PL}_{x_i, x_j}(\tilde{\boldsymbol{\Theta}}^t_m)$ does not admit a closed-form expression. 
For computational tractability, we therefore estimate the privacy loss via Monte Carlo sampling.

Here, we define $p_{\mathcal{X}^2}$ as the probability PL is violated: 
\begin{equation}
p_{\mathcal{X}^2} = \Pr\left[\mathrm{PL}_{X_i, X_j}(Z) > \epsilon \right],
\end{equation}
where $X_i, X_j$ 
are random variables representing two secret records in $\mathcal{X}$, we require $p_{\mathcal{X}^2} \leq \delta$, i.e., the posterior leakage bound can be satisfied with a high probability $1-\delta$.


{\rev 
\begin{definition}
\label{def:PmDP}
\textbf{$(\epsilon,\delta)$-Probabilistic mDP (PmDP).}
A randomized mechanism
$\mathcal{M}_\ell:\mathcal{X}\rightarrow\mathcal{Z}$
satisfies $(\epsilon,\delta)$-PmDP if
\begin{equation}
\label{eq:AmDP}
\frac{1}{|\mathcal{X}|^2}
\sum_{x_i,x_j\in\mathcal{X}}
\Pr_{z\sim\mathcal{M}_\ell(x_i)}
\!\left[
\mathrm{PL}_{x_i,x_j}(z)>\epsilon
\right]
\leq \delta,
\end{equation}
where $\delta\in[0,1)$ is the average probability, over uniformly
selected input pairs $(x_i,x_j)$ and the randomness of
$\mathcal{M}_\ell$, that the privacy-loss bound is violated.
Equivalently, the bound
$\mathrm{PL}_{x_i,x_j}(z)\leq\epsilon$ holds with probability at
least $1-\delta$ for a uniformly selected input pair and its
corresponding randomized output. When $\delta=0$, the definition
recovers exact PL as given in Eq.~(\ref{eq:mDP}).
\end{definition}}
{\rev 
\begin{proposition}
\label{prop:composition}
(Sequential composition of PmDP, adapted from
\cite{Liu-CCS2025})
Let $\{\mathcal{M}_t^{\textsf{p}}\}_{t\in\mathcal{T}}$
be randomized mechanisms with conditionally
independent outputs given the input.
If each $\mathcal{M}_t^{\textsf{p}}$ satisfies
$(\epsilon_t,\delta_t)$-PmDP and
$\sum_{t\in\mathcal{T}}\delta_t<1$,
then the joint mechanism
$\mathcal{M}=(\mathcal{M}_t^{\textsf{p}})_{t\in\mathcal{T}}$
satisfies
$(\sum_{t\in\mathcal{T}}\epsilon_t,
\sum_{t\in\mathcal{T}}\delta_t)$-PmDP.
\end{proposition}} 

{\rev 
\begin{proposition}
\label{prop:sampling}
[Concentration Guarantee for Probabilistic mDP Sampling]  
If the empirical violation rate satisfies $\hat{p}_{\mathcal{S}_{\ell}} = \xi \delta$ for some constant $\xi < 1$, then $\Pr\left[p_{\mathcal{X}_{\ell}^2} \leq \delta\right] \geq 1 - 2\exp(-2S_\ell(1-\xi)^2\delta^2)$. The detailed proof can be found in  \textbf{Appendix~\ref{subsec:proof:sampling}}.
\end{proposition}}


\subsection{Step \textcircled{3}: Perturbation Optimization}
\label{subsec:step3}

The server-side after perturbation, depending on the surrogate family used during fitting, the server reconstructs each coefficient block through a predefined parametric mapping or through a low-rank factorization and then applies a structure-preserving post-processing step (which is privacy-free due to the post-processing immunity) to ensure feasibility and improve downstream stability.

\smallskip
\noindent\textbf{Server-side reconstruction (parametric and low-rank surrogates).}
The coordinator receives the privatized surrogate parameters $\{\widetilde{\boldsymbol{\Theta}}_m^{t}\}_{t\in\mathcal{T}}$ from each user $m$ and reconstructs each coefficient block in the permuted index space used during local fitting. The reconstruction follows the corresponding surrogate family: 
$\widetilde{\mathbf{A}}_m^{t} = f\left(\cdot;\widetilde{\boldsymbol{\Theta}}_m^{t}\right),~t\in\mathcal{T}$.

For each user $m$, we define $\mathcal{H}_m$ and $\mathcal{H}^\complement_m$ to denote the sets of real-perturbed location pairs that follow and do not follow the \emph{WEM constraints}, 
respectively. Specifically,
\normalsize
\begin{eqnarray}
\nonumber \mathcal{H}_m &=& \left\{(x_{i}, y_k) \in \mathcal{N}_m \times \mathcal{Y} \mid z_{i,k} - w_{k}e^{-\epsilon d_{x_{i},y_k}} = 0 \right\}\\ \nonumber 
\mathcal{H}^\complement_m &=& (\mathcal{N}_m \times \mathcal{Y}) \setminus \mathcal{H}_m. 
\end{eqnarray}  
\normalsize
Moreover, given each perturbed record $y_k \in \mathcal{Y}$, we define $\mathcal{H}_{m,k} = \left\{x_{i_m}\in \mathcal{N}_m \mid (x_{i_m}, y_k) \in \mathcal{H}_m \right\}$ as the set of user $m$'s peer records selecting $y_k$ by following WEM.

Following prior work \cite{ImolaUAI2022}, we apply a heuristic strategy to determine $\mathcal{H}_m$ given User $m$'s uploaded distance matrix $\tilde{\mathbf{D}}_{\mathcal{N}_m\times \mathcal{Y}}$ and utility loss matrix $\tilde{\mathbf{C}}_{\mathcal{N}_m\times \mathcal{Y}}$. 


We rewrite the objective function as 
\normalsize
\begin{eqnarray}
&& \sum_{m=1}^M\mathcal{L}\left(\mathbf{Z}_{\mathcal{N}_m}\right) \\
\nonumber 
&=&  
{\bl  \textstyle \sum_{m=1}^M\underbrace{\textstyle \sum_{(x_{i},y_k)\in \mathcal{H}_m} c_{i,k}z_{i,k}}_{\small \mbox{Following WEM}}} \\
&+& 
{\rd \textstyle \sum_{m=1}^M\underbrace{\textstyle \sum_{(x_{i},y_k)\in \mathcal{H}^\complement_m}c_{i,k}z_{i,k}}_{\small \mbox{Not following WEM}}}  \\
&=& 
\textstyle  
{\bl \sum_{k=1}^K \alpha_k w_k} + {\rd \sum_{m=1}^M \textbf{c}_{\mathcal{H}^\complement_m}\mathbf{z}_{\mathcal{H}^\complement_m}}
\end{eqnarray}
\normalsize
where each $\alpha_k = \sum_{m=1}^M\sum_{x_{i}\in \mathcal{H}_{m,k}} c_{i,l} e^{-\frac{\epsilon_{\textsf{pt}} d(x_{i}, y_\ell)}{2}}$ is a constant. The constraints described in Eq. (\ref{eq:mDPdiscrete})(\ref{eq:unitmeasure}) and Eq. (\ref{eq:mDPexp}) can be rewritten as 
\begin{equation}
\label{eq:LR-OMGconstrrewrite}
{\bl \underbrace{\mathbf{C}_{\mathcal{H}_m}\mathbf{z}_{\mathcal{H}_m}}_{\small \mbox{Following WEM}}} + {\rd \underbrace{\mathbf{C}_{\mathcal{H}^\complement_m}\mathbf{z}_{\mathcal{H}^\complement_m}}_{\small \mbox{Not following WEM}}} \geq \mathbf{b}_{\mathcal{N}_m}
\end{equation}
where $\left[\mathbf{C}_{\mathcal{H}_m}, \mathbf{C}_{\mathcal{H}^\complement_m}\right]$ includes the mDP constraints between the perturbation vectors of the records in $\mathcal{H}_m$: 
\begin{eqnarray}
\label{eq:C_H}
\nonumber && \left[\mathbf{C}_{\mathcal{H}_m}, \mathbf{C}_{\mathcal{H}^\complement_m}\right]\\ 
\nonumber 
&=& 
\left[\begin{array}{lcccr} 
\ddots & \cdots     & \cdots  & \cdots                              & \iddots \\
\cdots & 1   & \cdots  & -e^{{\epsilon_{\textsf{pt}} d_{x_i,x_j}}}    & \cdots \\
\cdots & -e^{{\epsilon_{\textsf{pt}} d_{x_i,x_j}}}    & \cdots  & 1    & \cdots \\
\iddots & \cdots     & \cdots  & \cdots \ & \ddots \\ \end{array}\right] 
\hspace{-0.06in}
\begin{array}{l}
\left\}
\begin{array}{l}
\forall (x_i, y), \\
(x_j,y) \\
\in \mathcal{N}_m 
\end{array}
\right.
\\
\end{array}
\end{eqnarray}
\normalsize
Here, the privacy budget assigned to the perturbation matrix optimization is $\epsilon_{\textsf{pt}}$. By sequential composition for PmDP (Proposition~\ref{prop:composition}), if $\epsilon_{\textsf{pt}} \le \epsilon - (\epsilon_{\textsf{dx}} + \epsilon_{\textsf{dy}} + \epsilon_{\textsf{cy}})$, then the overall mechanism satisfies the target $(\epsilon,\delta)$-PmDP guarantee.

Like \cite{Qiu-PETS2025}, we apply Benders' decomposition to decompose the LP into two simpler stages: a \emph{master program} and a set of \emph{subproblems}.

\noindent \textbf{Stage 1 - Master Program (MP):} The MP focuses on determining the weights of EM $\mathbf{w} = [w_1, \dots, w_K]$. To simplify the optimization, it replaces the utility loss terms associated with data perturbation for each user $m$, given by $\textbf{c}_{\mathcal{H}^\complement_m}\mathbf{z}_{\mathcal{H}^\complement_m}$, with a single decision variable $v_m$, i.e., $v_m = \textbf{c}_{\mathcal{H}^\complement_m}\mathbf{z}_{\mathcal{H}^\complement_m}$. The MP is formulated as the following LP problem: 
\begin{eqnarray}
\min && {\bl \underbrace{\textstyle \sum_{k=1}^K \alpha_k w_k}_{\small \mbox{Utility loss (following WEM)}}} \\
&+& {\rd \underbrace{\textstyle \sum_{m=1}^M v_m}_{\small \mbox{Substituted utility loss (not following WEM)}}}~
~\\
\text{subject to}&& \mathbf{v} \in \mathcal{V}
\end{eqnarray}
where $\mathcal{V}$ denotes the feasible region of $\mathbf{v} = [v_1, \dots, v_M]$. The MP iteratively estimates the values of $\mathbf{v}$ and checks the feasibility/optimality of $\mathbf{v}$ in the subsequent stage, adding cuts to $\mathcal{V}$ to refine the solution of MP.

\noindent \textbf{Stage 2 - Subproblems:} Each subproblem $\mathsf{sub}_m$ validates the guessed values of $v_m$ generated by the MP, subject to its own constraints: $\mathbf{C}_{\mathcal{H}_m}\mathbf{z}_{\mathcal{H}_m} \geq \mathbf{b}_{\mathcal{N}_m} - \mathbf{C}_{\mathcal{H}^\complement_m}\mathbf{z}^*_{\mathcal{H}^\complement_m}$. If the guessed value is suboptimal or infeasible, the subproblem generates a new constraint (or cut) to refine the feasible region $\mathcal{V}$, improving the MP’s future estimates of $\mathbf{v}$. The constraint matrices $\mathbf{C}_{\mathcal{H}_m}$ across different subproblems are disjoint, enabling parallel validation of the guessed values of $v_m$. This iterative process continues, with each subproblem progressively narrowing the solution space until convergence to the optimal solution is achieved.

{\qiu Finally, the coordinator restores the original ordering by applying the inverse permutation. For the square intra-neighborhood distance block,
\begin{equation}
\mathbf{Z}_m = \mathbf{E}_m\widehat{\mathbf{Z}}_m.
\label{eq:unpermute_square}
\end{equation}}

\DEL{
\begin{proposition}[mDP violation bound \cite{}]
\label{prop:privacyguarantee}
For any LR location $v_i$ of a User $n$, i.e., $v_i \in \mathcal{N}_n$, let $\mathcal{A}_i$ denote the set of obfuscated locations of $v_i$ that adhere to the exponential distribution constraints. Then:
\begin{itemize}
\item[(1)] For each pair of Users $n$ and $m$, the Geo-Ind constraint violation ratio is upper bounded by: 
\begin{equation}
1- \frac{2\sum_{(v_i, v_j) \in \mathcal{N}_n \times \mathcal{N}_m} |\mathcal{A}_i \cap \mathcal{A}_j|+(|\mathcal{N}_n|^2+|\mathcal{N}_m|^2-|\mathcal{N}_n|-|\mathcal{N}_m|)|\mathcal{V}|}{(|\mathcal{N}_n|+|\mathcal{N}_m|)(|\mathcal{N}_n|+|\mathcal{N}_m|-1)|\mathcal{V}|}.
\end{equation}
\item[(2)] For all users, the Geo-Ind constraint violation ratio is upper bounded by:
\begin{equation}
1-\frac{2\sum_{n=1}^M \sum_{m=n+1}^M \sum_{(v_i, v_j) \in \mathcal{N}_n \times \mathcal{N}_m} |\mathcal{A}_i \cap \mathcal{A}_j|+ \sum_{n=1}^M\left(|\mathcal{N}_n|^2 - |\mathcal{N}_n|\right)|\mathcal{V}|}{\sum_{n=1}^M|\mathcal{N}_n|\left(\sum_{n=1}^M|\mathcal{N}_n|-1\right)|\mathcal{V}|}.
\end{equation}
\end{itemize}
\end{proposition}}

{\cam \subsection{A Two-User Running Example}
\label{subsec:running_example}
Consider two users with local neighborhoods $\mathcal{N}_1=\{x_1,x_2,x_3\}$ and $\mathcal{N}_2=\{x_4,x_5,x_6\}$. and a shared public output domain $\mathcal{Y}=\{y_1,y_2,y_3\}$. Suppose User~1's true location is $x_2$ and User~2's is $x_5$. Neither true location is sent to the server during mechanism generation.

\noindent \textbf{Local construction and release.} Each user $m$ constructs three blocks: $\mathbf{A}_m^{\textsf{dx}}$, containing distances between locations in $\mathcal{N}_m$; $\mathbf{A}_m^{\textsf{dy}}$, containing distances from $\mathcal{N}_m$ to $\mathcal{Y}$; and $\mathbf{A}_m^{\textsf{cy}}$, containing the corresponding utility losses. All three blocks are $3\times3$ in this example. The user permutes the neighborhood indices, fits a surrogate to each permuted block (for example, a rank-$r$ SVD approximation), perturbs the fitted parameters, and uploads the three perturbed parameter sets under budgets $\epsilon_{\textsf{dx}}$, $\epsilon_{\textsf{dy}}$, and $\epsilon_{\textsf{cy}}$. The user retains the mapping between the permuted rows and the original locations.

\noindent \textbf{Reconstruction and accounting.} From each upload, the server reconstructs an approximate $3\times3$ block. The privacy auditor evaluates the disclosure associated with the parameter releases and determines whether their allocated budgets meet the stated privacy criterion. The remaining budget $\epsilon_{\textsf{pt}}$ is used for the output mechanism, with equation~\ref{eqn:privacy_accounting}.
\DEL{\[
\epsilon_{dx}+\epsilon_{dy}+\epsilon_{cy}+\epsilon_{pt}
\leq \epsilon.
\]}
Reconstruction and optimization operate on the released parameters and require no further upload of the raw blocks.

\noindent \textbf{Optimization.} Using the reconstructed distance and loss blocks, the server identifies the entries handled by the WEM structure and the entries optimized in the local LPs. Its Benders master problem chooses the shared WEM weights and an estimated cost for each user. \emph{Subproblem~1} checks the master solution against the constraints for $\mathcal{N}_1$; \emph{Subproblem~2} does the same for $\mathcal{N}_2$. An infeasible or underestimated subproblem produces a cut, and the server resolves the master problem until the stopping criterion is met. The result comprises two row-stochastic blocks, $\mathbf{Z}^{*}_{\mathcal{N}_1\times\mathcal{Y}}$ and $\mathbf{Z}^{*}_{\mathcal{N}_2\times\mathcal{Y}}$, each of size $3\times3$. For example, the constraint comparing $x_1$ and $x_2$ at output $y_1$ is
\[
z_{1,1}
\leq
\exp\bigl(\epsilon_{\textsf{pt}} \tilde d(x_1,x_2)\bigr) z_{2,1},
\]
with the reverse-direction constraint also included when that pair is enforced.

\noindent \textbf{Local reporting.} The server returns each user's mechanism block with its row order. User~1 privately selects the row corresponding to $x_2$, and User~2 selects the row corresponding to $x_5$. Each user samples one $y_k$ according to the probabilities in that row and sends only the sampled report for the downstream task. Thus, a $3\times3$ block describes possible reports for all three locations in a neighborhood; the row selection itself remains local to the user.}

\section{Performance Evaluation}
\label{sec:exp}
In this section, we empirically evaluate the effectiveness of our proposed approximation-based distributed LP framework\footnote{Artifact is available at \url{https://anonymous.4open.science/r/DistA-LP-D2F0}}. 

\smallskip
\noindent\textbf{Datasets}. 
The geographic boundaries and dataset statistics for each city are summarized below:
\begin{itemize}
    \item \textbf{Rome, Italy}, bounded by (lat = 41.66, lon = 12.24) in the southwest and (lat = 42.10, lon = 12.81) in the northeast, including 43,160 nodes and 89,739 edges.
    \item \textbf{NYC, USA}, bounded by (lat = 40.50, lon = -74.25) in the southwest and (lat = 40.91, lon = -73.70) in the northeast, including 55,353 nodes and 139,638 edges.
    \item \textbf{London, UK}, bounded by (lat = 51.29, lon = -0.51) in the southwest and (lat = 51.69, lon = 0.28) in the northeast, including 12,820 nodes and 299,524 edges.
\end{itemize}

For each city, we do not rely solely on the full metropolitan graph but also construct a family of smaller, axis-aligned subregions designed to probe the robustness of our framework under varying spatial extents and network densities. In Rome, we systematically partition the area spanning longitudes 12.24-12.81 and latitudes 41.66-42.10 into a set of contiguous rectangular windows, yielding induced subgraphs with approximately 3,000 to 15,000 nodes. These regions are concentrated around the urban core, where the underlying road network is dense and highly irregular. For New York City, we select several windows within longitudes -74.25 to -73.70 and latitudes 40.50 to 40.91, covering lower and mid-Manhattan, as well as adjacent neighborhoods. This produces subgraphs with approximately 6,000 to 29,000 nodes. Unlike Rome, these subgraphs inherit NYC's strongly anisotropic, grid-like structure. For London, we focus on central districts within longitudes -0.51 to 0.28 and latitudes 51.29 to 51.69, forming six moderately sized rectangles that encompass Westminster, the City of London, and neighboring boroughs. Across all these cities, these carefully chosen subregions provide compact yet representative benchmarks, allowing us to evaluate the effect of our obfuscation and reconstruction methods across heterogeneous spatial environments. 

\smallskip
{\cam \noindent\textbf{Experimental and privacy-auditing settings.} We evaluate overall privacy budgets $\epsilon\in\{4,7,10\}\,\mathrm{km}^{-1}$. For each coefficient type $t\in\{\textsf{dx},\textsf{dy},\textsf{cy}\}$, the upload allocation $\epsilon_t$ is set from the maximum privacy loss observed in its audit; the remaining budget is $\epsilon_{\textsf{pt}} = \epsilon_{\textsf{dx}} + \epsilon_{\textsf{dy}} + \epsilon_{\textsf{cy}}$. We set the audit threshold to $\delta=10^{-4}$ and use $S_{\ell}=100$ audit samples. Each experiment is repeated 10 times, and results are reported as mean $\pm$ standard deviation. Unless otherwise specified, each user has a neighborhood of size $|\mathcal{N}_m|=20$, and the low-rank SVD surrogate uses rank $r=5$.}

\smallskip
\noindent\textbf{Computation resource.} Our tests were conducted on a machine equipped with an Intel Core i9-13900F processor, featuring 24 cores with a base clock speed ranging from 2.00 GHz to a maximum of 5.60 GHz. The system was configured with 32 GB of DDR5 RAM (32 GB, 4800 MHz) and included an NVIDIA GeForce RTX 4090 graphics card with 24 GB of GDDR6X VRAM. We use the MATLAB optimization toolbox ``$\mathsf{linearprog}$'' to solve the LP problems.


\subsection{Matrix Approximation Evaluation}
\label{subsec:exp:approx}
\begin{table*}[t]
\caption{Performance of coefficient approximation across the three datasets (mean~$\pm$~$1.96\times$standard error).}
\centering
\footnotesize
\begin{tabular}{p{2.0cm} p{1.2cm}p{1.6cm}p{1.2cm} p{1.2cm}p{1.6cm}p{1.2cm} p{1.2cm}p{1.2cm}}
\toprule
\multicolumn{1}{l}{}& \multicolumn{3}{c}{Distance matrix $\tilde{\mathbf{A}}_m^{\textsf{dx}}$} & \multicolumn{3}{c}{Distance matrix $\tilde{\mathbf{A}}^d_{\mathcal{N}_m\times \mathcal{Y}}$} & \multicolumn{2}{c}{Utility loss matrix $\tilde{\mathbf{A}}^c_{\mathcal{N}_m\times \mathcal{Y}}$}\\
 \cline{2-9}
\textbf{Methods} & PL & RelErr & ViolRatio & PL & RelErr & ViolRatio & PL & RelErr  \\
\midrule
\multicolumn{9}{c }{Rome road map}\\
\hline
2D Gaussian & 1.34±0.93 & 0.60±0.07 & 0.38±0.04 & 1.11±0.66 & 0.50±0.09 & 0.44±0.06 & 1.93±1.91 & 0.60±0.05  \\
\hline
2D Polynomial & 5.36±3.27 & 168.64±75.42 & 0.39±0.14 & 5.59±3.04 & 257.21±219.4 & 0.45±0.16 & 6.42±4.64 & 53.83±49.01  \\
\hline
RBF & 8.02±7.15 & 0.57±0.03 & 0.41±0.03 & 9.26±8.62 & 0.54±0.04 & 0.41±0.04 & 8.06±6.87 & 0.70±0.02  \\
\hline
Low-Rank SVD & 0.95±0.16 & 0.07±0.00 & 0.13±0.01 & 0.76±0.13 & 0.04±0.00 & 0.22±0.01 & 1.25±0.23 & 0.03±0.00  \\
\midrule
\multicolumn{9}{c }{London road map}
\\
\hline
2D Gaussian & 0.96±0.28 & 0.85±0.66 & 0.42±0.07 & 0.67±0.29 & 0.48±0.03 & 0.45±0.06 & 1.00±0.56 & 0.48±0.02  \\
\hline
2D Polynomial & 3.24±2.27 & 220.31±123.4 & 0.43±0.13 & 2.32±1.15 & 189.81±116.53 & 0.45±0.17 & 3.04±0.71 & 161.41±81.89  \\
\hline
RBF & 3.86±0.90 & 0.52±0.01 & 0.43±0.03 & 3.52±1.14 & 0.54±0.02 & 0.38±0.03 & 3.36±1.10 & 0.63±0.01  \\
\hline
Low-Rank SVD & 0.91±0.09 & 0.07±0.00 & 0.13±0.01 & 0.64±0.12 & 0.04±0.00 & 0.20±0.02 & 0.92±0.58 & 0.03±0.00  \\
\midrule
\multicolumn{9}{c }{New York City road map}
\\
\hline
2D Gaussian & 0.52±0.14 & 0.66±0.03 & 0.39±0.02 & 0.51±0.23 & 0.66±0.07 & 0.44±0.08 & 0.54±0.28 & 0.54±0.00  \\
\hline
2D Polynomial & 2.43±0.49 & 253.83±141.4 & 0.42±0.22 & 2.07±0.29 & 219.51±61.59 & 0.50±0.18 & 3.32±2.63 & 48.78±31.61  \\
\hline
RBF & 3.13±2.12 & 0.64±0.03 & 0.45±0.03 & 4.19±1.91 & 0.64±0.03 & 0.42±0.08 & 3.94±2.87 & 0.67±0.02  \\
\hline
Low-Rank SVD & 1.02±0.20 & 0.08±0.01 & 0.09±0.01 & 0.63±0.10 & 0.06±0.01 & 0.15±0.02 & 1.03±0.36 & 0.03±0.00  \\
\bottomrule
\end{tabular}
\label{tab:exp:matrix_upload}
\end{table*}

\smallskip
\noindent \textbf{Comparison of Different Coefficient Approximation Methods.} We evaluated four surrogate families in Table \ref{tab:exp:matrix_upload} and found that \emph{low-rank SVD} offers the best privacy-accuracy trade-off, consistently yielding low reconstruction error and low violation ratios under small privacy budgets. Based on this comparison, our framework adopts the low-rank SVD surrogate as the default approximation model because it provides the best balance between accuracy, scalability, and privacy robustness. Detailed comparative results are reported in \textbf{Section~\ref{subsec:exp:approx}}.
 
\noindent \textbf{Performance metrics}. In this section, we evaluate two key aspects of the approximated matrices,  $\widehat{\mathbf{A}}^d_{\mathcal{N}^2_m}$, $\widehat{\mathbf{A}}^d_{\mathcal{N}_m\times \mathcal{Y}}$, and $\widehat{\mathbf{A}}^c_{\mathcal{N}_m\times \mathcal{Y}}$: 
(1) the \emph{privacy loss} (formally defined in Eq. (\ref{eq:privacyloss}) in \textbf{Definition~\ref{def:PL}}), 
and (2) the approximation accuracy. 

To quantify the overall approximation accuracy, we use the \emph{relative Frobenius norm error (or RelErr for simplicity)} 
between the original matrix $A$ and its approximation $\hat{A}$:
\begin{equation}
\text{RelErr} = \frac{\| A - \hat{A} \|_{F}}{\| A \|_{F}} 
= \frac{\sqrt{\sum_{i=1}^{n} \sum_{j=1}^{m} \left( A_{ij} - \hat{A}_{ij} \right)^{2}}}
{\sqrt{\sum_{i=1}^{n} \sum_{j=1}^{m} A_{ij}^{2}}},    
\end{equation}
where $\|\cdot\|_{F}$ denotes the Frobenius norm. 
This metric measures the overall deviation of the approximation from the original matrix 
as a fraction of the original matrix’s magnitude. 
A lower value indicates a closer match, while a higher value reflects larger deviations. 
Normalization by $\|A\|_{F}$ makes the measure scale-independent, 
allowing fair comparisons across matrices with different magnitudes.

For the distance matrices $\mathbf{A}_m^{\textsf{dx}}$ and $\mathbf{A}^d_{\mathcal{N}_m\times \mathcal{Y}}$, 
approximation accuracy is particularly critical because it directly determines 
the mDP indistinguishability constraints. 
If an approximated distance value $\hat{d}(x_i,x_j)$ exceeds the true distance $d_{x_i,x_j}$, 
the term $e^{\epsilon \hat{d}(x_i,x_j)}$ becomes larger, thereby relaxing the constraint
\begin{equation}
z_{i,k} - e^{\epsilon  d_{x_i,x_j}} z_{j,k} \le 0,
\end{equation}
which can lead to solutions that violate the intended privacy guarantee under the true distances. 
To capture this effect, we define the \emph{violation ratio} as
\begin{equation}
\text{ViolRatio} = \frac{\left| \left\{ (i,j) | \hat{d}(x_i,x_j) > d_{x_i,x_j} \right\} \right|}{n \times n}, 
\end{equation}
where $|\cdot|$ denotes set cardinality and $n \times n$ is the total number of distance entries. 
A lower violation ratio reflects a more conservative approximation that preserves 
or strengthens the original privacy constraints.

\smallskip
\noindent \textbf{Experimental results}. Table~\ref{tab:exp:matrix_upload} compares RelErr, violation ratio, and privacy loss of different approximation methods across the three datasets. The performance patterns observed in Table~\ref{tab:exp:matrix_upload} are consistent with the modeling/complexity trade-offs summarized in Table~\ref{tab:approx_privacy}.
Particularly, \emph{Low-rank SVD} captures the global structure of the distance/utility matrices with only $k(m+n)$ parameters and inherits an implicit denoising effect from truncation, so its parameter perturbation concentrates on a small, well-conditioned set, yielding both low reconstruction error and low violation ratios in Table~\ref{tab:exp:matrix_upload}.
\emph{2D Gaussian} offers a compact, smooth parametric surface; when the underlying matrices are approximately smooth and unimodal within local neighborhoods, its few parameters perturb stably, producing moderate errors and moderate violation ratios (Table~\ref{tab:exp:matrix_upload}), making it a reasonable backup where SVD is impractical.
\emph{RBF} has high expressivity but a much larger parameter set; privacy noise must be spread over many coefficients (and often asymmetrically to curb violations), which amplifies effective perturbation and explains the higher privacy loss and violation ratios observed in Table~\ref{tab:exp:matrix_upload} despite decent fit on certain regions. Finally, \emph{2D Polynomial} surrogates can be numerically unstable and poorly capture non-polynomial curvature, especially near boundaries; under perturbation, these instabilities magnify, leading to the very large relative errors and no compensating safety gains seen in Table~\ref{tab:exp:matrix_upload}.
Overall, Table~1’s complexity/regularity arguments justify the Table~\ref{tab:exp:matrix_upload} outcome: SVD provides the best privacy-utility-safety balance, with Gaussian serving as a compact and robust runner-up, while RBF and Polynomial methods are dominated once privacy noise is accounted for.

\begin{table*}[t]
\caption{Violation ratio (Mean$\pm$1.96$\times$standard error).}
\centering
\scriptsize 
\begin{tabular}{p{1.80cm} |  p{0.80cm} p{1.10cm} p{1.10cm}  p{1.10cm} | p{1.10cm} p{1.10cm} p{1.10cm} | p{1.10cm} p{1.10cm} p{1.10cm}} 
\toprule
\multicolumn{2}{c }{Number of records} & \multicolumn{3}{c }{2,000} & \multicolumn{3}{c }{4,000}& \multicolumn{3}{c }{6,000}\\
\midrule
\multicolumn{2}{c }{Privacy budget (km$^{-1}$)} & $\epsilon = 4.0$& $\epsilon = 7.0$& $\epsilon = 10.0$ & $\epsilon = 4.0$& $\epsilon = 7.0$& $\epsilon = 10.0$& $\epsilon = 4.0$& $\epsilon = 7.0$& $\epsilon = 10.0$\\
\midrule
\multicolumn{11}{c }{Rome road map}\\
\hline
Pre-defined& EM & 0.00±0.00 & 0.00±0.00 & 0.00±0.00 & 0.00±0.00 & 0.00±0.00 & 0.00±0.00 & 0.00±0.00 & 0.00±0.00 & 0.00±0.00 \\
\cline{2-2}
Noise Distribution & Laplace  & 0.00±0.00 & 0.00±0.00 & 0.00±0.00 & 0.00±0.00 & 0.00±0.00 & 0.00±0.00 & 0.00±0.00 & 0.00±0.00 & 0.00±0.00 \\
\cline{1-2}
Hybrid  & RMP & 0.00±0.00 & 0.00±0.00 & 0.00±0.00 & 0.00±0.00 & 0.00±0.00 & 0.00±0.00 & 0.00±0.00 & 0.00±0.00 & 0.00±0.00 \\
\cline{2-2}
Method &  COPT  & ------------- & ------------- & ------------- & ------------- & ------------- & ------------- & ------------- & ------------- & ------------- \\
\cline{1-2}
Optimization &  LP  & ------------- & ------------- & ------------- & ------------- & ------------- & ------------- & ------------- & ------------- & ------------- \\
\cline{2-2}
Based & LP-A & 0.59±0.04 & 0.59±0.04 & 0.59±0.0436 & 0.57±0.04 & 0.57±0.04 & 0.57±0.04 & 0.57±0.04 & 0.57±0.04 & 0.57±0.04 \\
\cline{2-2} 
Methods & PAnDA & 0.00±0.00 & 0.00±0.00 & 0.00±0.00 & 0.00±0.00 & 0.00±0.00 & 0.00±0.00 & 0.00±0.00 & 0.00±0.00 & 0.00±0.00 \\
\hline
\rowcolor{lightgray!30}
\multicolumn{2}{c}{\textsc{DistA-LP}$^\dagger$} & 0.00$\pm$0.00 & 0.00$\pm$0.00 & 0.00$\pm$0.00 & 0.00$\pm$0.00 & 0.00$\pm$0.00 & 0.00$\pm$0.00 & 0.00$\pm$0.00 & 0.00$\pm$0.00 & 0.00$\pm$0.00 \\
\midrule
\multicolumn{11}{c }{London road map}\\
 \hline
Pre-defined& EM & 0.00±0.00 & 0.00±0.00 & 0.00±0.00 & 0.00±0.00 & 0.00±0.00 & 0.00±0.00 & 0.00±0.00 & 0.00±0.00 & 0.00±0.00 \\
\cline{2-2}
Noise Distribution & Laplace  & 0.00±0.00 & 0.00±0.00 & 0.00±0.00 & 0.00±0.00 & 0.00±0.00 & 0.00±0.00 & 0.00±0.00 & 0.00±0.00 & 0.00±0.00 \\
\cline{1-2}
Hybrid  & RMP & 0.00±0.00 & 0.00±0.00 & 0.00±0.00 & 0.00±0.00 & 0.00±0.00 & 0.00±0.00 & 0.00±0.00 & 0.00±0.00 & 0.00±0.00  \\
\cline{2-2}
Method &  COPT  & ------------- & ------------- & ------------- & ------------- & ------------- & ------------- & ------------- & ------------- & ------------- \\
\cline{1-2}
Optimization &  LP  & ------------- & ------------- & ------------- & ------------- & ------------- & ------------- & ------------- & ------------- & ------------- \\
\cline{2-2}
Based & LP-A & 0.50±0.03 & 0.50±0.03 & 0.50±0.03 & 0.50±0.03 & 0.50±0.03 & 0.50±0.03 & 0.50±0.03 & 0.50±0.03 & 0.50±0.03 \\
\cline{2-2} 
 Methods & PAnDA & 0.00±0.00 & 0.00±0.00 & 0.00±0.00 & 0.00±0.00 & 0.00±0.00 & 0.00±0.00 & 0.00±0.00 & 0.00±0.00 & 0.00±0.00 \\
\hline
\rowcolor{lightgray!30}
\multicolumn{2}{c}{\textsc{DistA-LP}$^\dagger$} & 0.00$\pm$0.00 & 0.00$\pm$0.00 & 0.00$\pm$0.00 & 0.00$\pm$0.00 & 0.00$\pm$0.00 & 0.00$\pm$0.00 & 0.00$\pm$0.00 & 0.00$\pm$0.00 & 0.00$\pm$0.00 \\
\midrule
\multicolumn{11}{c }{New York City road map}\\
\hline
Pre-defined& EM & 0.00±0.00 & 0.00±0.00 & 0.00±0.00 & 0.00±0.00 & 0.00±0.00 & 0.00±0.00 & 0.00±0.00 & 0.00±0.00 & 0.00±0.00 \\
\cline{2-2}
Noise Distribution & Laplace  & 0.00±0.00 & 0.00±0.00 & 0.00±0.00 & 0.00±0.00 & 0.00±0.00 & 0.00±0.00 & 0.00±0.00 & 0.00±0.00 & 0.00±0.00 \\
\cline{1-2}
Hybrid  & RMP & 0.00±0.00 & 0.00±0.00 & 0.00±0.00 & 0.00±0.00 & 0.00±0.00 & 0.00±0.00 & 0.00±0.00 & 0.00±0.00 & 0.00±0.00 \\
\cline{2-2}
Method &  COPT  & ------------- & ------------- & ------------- & ------------- & ------------- & ------------- & ------------- & ------------- & ------------- \\
\cline{1-2}
Optimization &  LP  & ------------- & ------------- & ------------- & ------------- & ------------- & ------------- & ------------- & ------------- & ------------- \\
\cline{2-2}
Based & LP-A & 0.56±0.01 & 0.56±0.01 & 0.56±0.01 & 0.56±0.01 & 0.56±0.01 & 0.56±0.01 & 0.56±0.01 & 0.56±0.01 & 0.56±0.01 \\
\cline{2-2} 
Methods & PAnDA & 0.00±0.00 & 0.00±0.00 & 0.00±0.00 & 0.00±0.00 & 0.00±0.00 & 0.00±0.00 & 0.00±0.00 & 0.00±0.00 & 0.00±0.00 \\
\hline
\rowcolor{lightgray!30}
\multicolumn{2}{c}{\textsc{DistA-LP}$^\dagger$} &  0.00$\pm$0.00 & 0.00$\pm$0.00 & 0.00$\pm$0.00 & 0.00$\pm$0.00 & 0.00$\pm$0.00 & 0.00$\pm$0.00 & 0.00$\pm$0.00 & 0.00$\pm$0.00 & 0.00$\pm$0.00 \\
\bottomrule
\end{tabular}
\vspace{-0.15in}
\label{tab:violation}
\end{table*}

\subsection{Utility Loss}

We evaluate utility using the expected distance between the true location $\mathbf{x}$ and the reported location $\mathcal{M}(\mathbf{x})$ under the same $\ell_p$-norm. Formally, we compute $\mathbb{E}_{\mathbf{x} \sim p(\mathbf{x}), \mathbf{y} \sim \mathcal{M}(\mathbf{x})}[d_p(\mathbf{x}, \mathbf{y})]$.


We compare our approach against the following representative methods based on mDP:
\begin{itemize}
  \item \textbf{Pre-defined noise distribution}, including \textbf{Exponential Mechanism (EM)} \cite{Chatzikokolakis-PoPETs2015}, which samples $\mathbf{y}$ with $\Pr[\mathcal{M}(\mathbf{x})=\mathbf{y}] \propto \exp\big(-\epsilon d_2(\mathbf{x},\mathbf{y})\big)$, where $d_2$ is the $\ell_2$ distance, and \textbf{Planar Laplace (Laplace)} \cite{Andres-CCS2013}, which achieves geo-indistinguishability by adding 2D Laplace noise, $\Pr[\mathcal{M}(\mathbf{x})=\mathbf{y}] = \frac{\epsilon^2}{2\pi_m}\exp\big(-\epsilon d_2(\mathbf{x},\mathbf{y})\big)$. On discrete grids, Laplace uses an effective $\epsilon' \le \epsilon$ to maintain mDP.
  \item \textbf{Optimization-based methods}, including \textbf{LP} \cite{Bordenabe-CCS2014} and \textbf{PAnDA} \cite{Liu-CCS2025}. LP optimizes a mechanism that minimizes expected loss subject to mDP constraints, which scales poorly over $\mathcal{X}$. PAnDA uses a two-phase framework in which each user selects a small anchor set, allowing the server to solve a compact LP on a reduced domain. We partition the target region into a coarse $8 \times 8$ grid and apply the LP method (low computation time but inaccurate utility loss), denoted as \textbf{LP}. 

  \item \textbf{Hybrid Methods}, including \textbf{COPT} \cite{ImolaUAI2022}, which runs LP on a small, utility-sensitive subset of outputs and uses EM elsewhere—cutting cost while keeping mDP, and \textbf{Bayesian Remapping (RMP)} \cite{Chatzikokolakis-PETS2017}, a post-process of an initial output $\mathbf{y}$ to $\tilde{\mathbf{y}}=\arg\min_{\tilde{\mathbf{y}}\in\mathcal{Y}} \mathbb{E}_{\mathbf{x}\sim \Pr[\mathbf{x}\mid \mathbf{y}]}\left[\mathcal{L}(\mathbf{x},\tilde{\mathbf{y}})\right]$,     where $\Pr[\mathbf{x}\mid \mathbf{y}]$ follows from prior $p_{\mathbf{x}}$ and mechanism $\mathbf{Z}$. RMP improves utility without weakening mDP.
\end{itemize}


Table~\ref{tab:UL2norm} reports the expected utility loss across perturbation mechanisms on the Rome, London, and New York City road-map datasets, with values scaled by $10{,}000$. Across all privacy budgets $\epsilon\in\{4,7,10\}\mathrm{km}^{-1}$ and dataset sizes ($N=2,000$, $4,000$, $6,000$), \textsc{DistA-LP} consistently achieves the lowest loss with narrow confidence bands, outperforming predefined-noise baselines (EM, Laplace) and optimization/hybrid baselines (RMP, LP, PAnDA). 

Averaged over the three road-network datasets and grid sizes, DistA-LP reduces expected utility loss by 50\% compared to EM/Laplace and centralized LP, and achieves about a 49\% reduction relative to PAnDA. Compared to predefined-noise mechanisms (EM, Laplace), which calibrate noise only to the perturbation magnitude and treat all directions in the metric space isotropically, DistA-LP explicitly optimizes a utility-loss objective over the discretized domain, so it can allocate probability mass toward utility-favorable outputs while still satisfying mDP constraints.

Relative to hybrid methods such as COPT and RMP, which either optimize only on a small subset of outputs or improve utility through a post-processing remap of an existing mechanism, DistA-LP optimizes the full perturbation matrix jointly via the WEM+LP formulation, allowing it to adjust all entries coherently instead of partially correcting a predefined mechanism. Finally, compared to centralized LP optimization, DistA-LP leverages surrogate-based coefficient approximation to obtain conservative yet tight distance/utility matrices that exhibit near-zero violation ratios (as reported in Table \ref{tab:violation}) in the Appendix; this avoids the over-conservative noise that arises when centralized LP accumulates large distance overestimates, thereby yielding lower distortion at similar privacy budgets.



\begin{table*}[t]
\caption{Utility loss across different perturbation methods (10,000 meters) \newline Mean$\pm$1.96$\times$standard error; the algorithm without getting the results, we label its utility by ``--------''.}
\centering
\scriptsize 
\begin{tabular}{p{1.80cm} |  p{0.80cm} p{1.10cm} p{1.10cm}  p{1.10cm} | p{1.10cm} p{1.10cm} p{1.10cm} | p{1.10cm} p{1.10cm} p{1.10cm}} 
\toprule
\multicolumn{2}{c }{Number of records} & \multicolumn{3}{c }{2,000} & \multicolumn{3}{c }{4,000}& \multicolumn{3}{c }{6,000}\\
\midrule
\multicolumn{2}{c }{Privacy budget (km$^{-1}$)} & $\epsilon = 4.0$& $\epsilon = 7.0$& $\epsilon = 10.0$ & $\epsilon = 4.0$& $\epsilon = 7.0$& $\epsilon = 10.0$& $\epsilon = 4.0$& $\epsilon = 7.0$& $\epsilon = 10.0$\\
\midrule
\multicolumn{11}{c }{Rome road map}\\
\hline
Pre-defined& EM & 5.63±0.95 & 5.60±0.95 & 5.57±0.95 & 7.16±2.53 & 7.12±2.51 & 7.10±2.50 & 7.98±3.35 & 7.96±3.34 & 7.94±3.34 \\
\cline{2-2}
Noise Distribution & Laplace  & 5.68±0.95 & 5.68±0.95 & 5.68±0.95 & 7.30±2.66 & 7.30±2.66 & 7.30±2.66 & 8.01±3.35 & 8.01±3.35 & 8.01±3.35 \\
\cline{1-2}
Hybrid  & RMP & 5.30±0.91 & 5.26±0.91 & 5.23±0.92 & 6.81±2.54 & 6.79±2.54 & 6.78±2.54 & 7.73±3.37 & 7.71±3.36 & 7.69±3.35 \\
\cline{2-2}
Method &  COPT  & ------------- & ------------- & ------------- & ------------- & ------------- & ------------- & ------------- & ------------- & ------------- \\
\cline{1-2}
Optimization &  LP  & ------------- & ------------- & ------------- & ------------- & ------------- & ------------- & ------------- & ------------- & ------------- \\
\cline{2-2}
Based & LP-A & 5.19±0.91 & 5.15±0.91 & 5.14±0.91 & 6.76±2.55 & 6.73±2.55 & 6.72±2.55 & 7.63±3.32 & 7.60±3.32 & 7.60±3.32 \\
\cline{2-2} 
Methods & PAnDA & 4.95±1.03 & 4.95±1.03 & 4.95±1.03 & 6.43±2.60 & 6.43±2.60 & 6.43±2.60 & 7.39±3.25 & 7.39±3.25 & 7.39±3.25 \\
\hline
\rowcolor{lightgray!30}
\multicolumn{2}{c}{\textsc{DistA-LP}$^\dagger$} & 3.47$\pm$0.24 & 3.47$\pm$0.24 & 3.47$\pm$0.24 & 2.88$\pm$0.21 & 2.88$\pm$0.21 & 2.88$\pm$0.21 & 3.14$\pm$0.26 & 3.14$\pm$0.26 & 3.14$\pm$0.26 \\
\midrule
\multicolumn{11}{c }{London road map}\\
 \hline
Pre-defined& EM & 6.61±2.03 & 6.54±2.04 & 6.50±2.05 & 6.57±1.85 & 6.53±1.84 & 6.50±1.84 & 6.47±1.90 & 6.44±1.91 & 6.42±1.91 \\
\cline{2-2}
Noise Distribution & Laplace  & 6.74±2.00 & 6.74±2.01 & 6.74±2.01 & 6.66±1.86 & 6.66±1.86 & 6.66±1.86 & 6.51±1.90 & 6.51±1.90 & 6.51±1.90 \\
\cline{1-2}
Hybrid  & RMP & 6.22±1.99 & 6.17±2.00 & 6.14±2.01 & 6.30±1.87 & 6.27±1.87 & 6.25±1.87 & 6.22±1.92 & 6.20±1.92 & 6.18±1.92  \\
\cline{2-2}
Method &  COPT  & ------------- & ------------- & ------------- & ------------- & ------------- & ------------- & ------------- & ------------- & ------------- \\
\cline{1-2}
Optimization &  LP  & ------------- & ------------- & ------------- & ------------- & ------------- & ------------- & ------------- & ------------- & ------------- \\
\cline{2-2}
Based & LP-A & 6.13±2.01 & 6.08±2.02 & 6.07±2.02 & 6.22±1.88 & 6.18±1.88 & 6.17±1.88 & 6.14±1.92 & 6.11±1.92 & 6.10±1.92 \\
\cline{2-2} 
Methods & PAnDA & 6.03±2.17 & 6.02±2.17 & 6.02±2.17 & 5.93±1.80 & 5.93±1.80 & 5.93±1.80 & 5.88±1.99 & 5.88±1.99 & 5.88±1.99 \\
\hline
\rowcolor{lightgray!30}
\multicolumn{2}{c}{\textsc{DistA-LP}$^\dagger$} & 2.74$\pm$0.20 & 2.74$\pm$0.21 & 2.74$\pm$0.20 & 3.79$\pm$0.08 & 3.79$\pm$0.08 & 3.79$\pm$0.08 & 3.99$\pm$0.11 & 3.99$\pm$0.11 & 3.99$\pm$0.11 \\
\midrule
\multicolumn{11}{c }{New York City road map}\\
\hline
Pre-defined& EM & 5.64±0.91 & 5.58±0.92 & 5.54±0.93 & 5.66±0.85 & 5.62±0.85 & 5.59±0.85 & 4.98±0.81 & 4.93±0.79 & 4.91±0.79 \\
\cline{2-2}
Noise Distribution & Laplace  & 5.80±0.92 & 5.80±0.92 & 5.80±0.92 & 5.76±0.87 & 5.77±0.87 & 5.77±0.87 & 5.14±0.87 & 5.14±0.87 & 5.14±0.87 \\
\cline{1-2}
Hybrid  & RMP & 5.22±0.96 & 5.18±0.95 & 5.15±0.95 & 5.36±0.84 & 5.34±0.84 & 5.32±0.84 & 4.74±0.75 & 4.72±0.75 & 4.70±0.74 \\
\cline{2-2}
Method &  COPT  & ------------- & ------------- & ------------- & ------------- & ------------- & ------------- & ------------- & ------------- & ------------- \\
\cline{1-2}
Optimization &  LP  & ------------- & ------------- & ------------- & ------------- & ------------- & ------------- & ------------- & ------------- & ------------- \\
\cline{2-2}
Based & LP-A & 5.12±0.92 & 5.05±0.92 & 5.04±0.92 & 5.25±0.83 & 5.22±0.84 & 5.21±0.84 & 4.66±0.74 & 4.62±0.73 & 4.61±0.73 \\
\cline{2-2} 
Methods & PAnDA & 5.12±1.20 & 5.12±1.20 & 5.12±1.21 & 4.78±0.73 & 4.78±0.73 & 4.78±0.73 & 4.89±0.98 & 4.89±0.98 & 4.89±0.98 \\
\hline
\rowcolor{lightgray!30}
\multicolumn{2}{c}{\textsc{DistA-LP}$^\dagger$} &  1.60$\pm$0.09 & 1.58$\pm$0.06 & 1.54$\pm$0.03 & 2.25$\pm$0.16 & 2.25$\pm$0.17 & 2.25$\pm$0.15 & 2.47$\pm$0.09 & 2.46$\pm$0.11 & 2.45$\pm$0.09 \\
\bottomrule
    \end{tabular}
\vspace{-0.25in}
\label{tab:UL2norm}
\end{table*}

\subsection{mDP Violation Ratio}
Table~\ref{tab:violation} reports the empirical mDP violation ratio, defined as the fraction of tested mDP constraints whose privacy loss exceeds the target budget $\epsilon$. A lower violation ratio indicates stronger empirical compliance with the required mDP constraints. Across all three road-network datasets, privacy budgets, and domain sizes, \textsc{DistA-LP} achieves $0.00\pm0.00$ violation ratio, matching the predefined mechanisms and post-processing baselines that preserve mDP by construction. In contrast, LP-A exhibits substantial violation ratios, around $0.50$--$0.59$, because its coarse approximation can relax distance-based mDP constraints and produce mechanisms that violate the intended privacy bound. These results show that the proposed surrogate-based approximation and conservative reconstruction steps preserve empirical mDP validity while enabling scalable optimization.

\subsection{Computation Efficiency}

Table~\ref{tab:time} reports the end-to-end computation time for each perturbation mechanism across privacy budgets $\epsilon \in \{4,7,10\}\mathrm{km}^{-1}$ and dataset sizes ($|\mathcal{X}|=2,000$, $4,000$, $6,000$) on the road-network datasets. Predefined-noise baselines (EM, Laplace) are extremely fast ($\le 0.05$ time units) and essentially insensitive to $\epsilon$ or $|\mathcal{X}|$, since they avoid solving large optimization problems. Hybrid methods (RMP, PAnDA) incur modest overhead (roughly $0.03\text{--}0.6$ seconds), remaining lightweight compared to full LP. In contrast, centralized LP/COPT often reach the 1,800-second timeout and, when they do finish, exhibit a steep growth in runtime---from about $15.5\text{--}18.0$ seconds at $|\mathcal{X}| = 2{,}000$ to $\approx 30.2\text{--}37.7$ seconds at $|\mathcal{X}| = 4{,}000$ and $\approx 45.1\text{--}46.9$ at $|\mathcal{X}| = 6{,}000$, reflecting their poor scalability. DistA-LP lies between these extremes: it is somewhat more expensive than purely closed-form or hybrid baselines but remains in a practical range and shows substantially more stable scaling than centralized LP. Overall, the table illustrates that DistA-LP achieves a favorable trade-off between computation cost and utility, dramatically improving scalability over centralized LP-based mechanisms while retaining their optimization benefits.

\begin{table*}[t]
\caption{Computation time (Seconds) Mean$\pm$1.96$\times$standard error. }
\centering
\scriptsize 
\begin{tabular}{p{1.80cm} |  p{0.80cm} p{1.10cm} p{1.10cm}  p{1.10cm} | p{1.10cm} p{1.10cm} p{1.10cm} | p{1.10cm} p{1.10cm} p{1.10cm}} 
\toprule
\multicolumn{2}{c }{Number of records} & \multicolumn{3}{c }{2,000} & \multicolumn{3}{c }{4,000}& \multicolumn{3}{c }{6,000}\\
\midrule
\multicolumn{2}{c }{Privacy budget (km$^{-1}$)} & $\epsilon = 4.0$& $\epsilon = 7.0$& $\epsilon = 10.0$ & $\epsilon = 4.0$& $\epsilon = 7.0$& $\epsilon = 10.0$& $\epsilon = 4.0$& $\epsilon = 7.0$& $\epsilon = 10.0$\\
\midrule
\multicolumn{11}{c }{Rome road map}\\
\hline
Pre-defined& EM & $\leq 0.05$ & $\leq 0.05$ & $\leq 0.05$ & $\leq 0.05$ & $\leq 0.05$ & $\leq 0.05$ & $\leq 0.05$ & $\leq 0.05$ & $\leq 0.05$ \\
\cline{2-2}
Noise Distribution & Laplace  & $\leq 0.05$ & $\leq 0.05$ & $\leq 0.05$ & $\leq 0.05$ & $\leq 0.05$ & $\leq 0.05$ & $\leq 0.05$ & $\leq 0.05$ & $\leq 0.05$ \\
\cline{2-2}
Hybrid  & RMP & $\leq 0.05$ & $\leq 0.05$ & $\leq 0.05$ & 0.25±0.00 & 0.26±0.01 & 0.27±0.01 & 0.45±0.00 & 0.45±0.00 & 0.58±0.04 \\
\cline{2-2}
Method &  COPT  & $\geq 1,800$ & $\geq 1,800$ & $\geq 1,800$ & $\geq 1,800$ & $\geq 1,800$ & $\geq 1,800$ & $\geq 1,800$ & $\geq 1,800$ & $\geq 1,800$ \\
\cline{1-2}
Optimization &  LP  & $\geq 1,800$ & $\geq 1,800$ & $\geq 1,800$ & $\geq 1,800$ & $\geq 1,800$ & $\geq 1,800$ & $\geq 1,800$ & $\geq 1,800$ & $\geq 1,800$ \\
\cline{2-2}
Based & LP-A & 11.27±0.12 & 11.79±0.86 & 13.69±0.54 & 21.99±0.29 & 23.03±0.67 & 28.67±0.95 & 32.24±0.41 & 33.33±0.53 & 36.95±1.21 \\
\cline{2-2} 
 Methods & PAnDA & 0.083±0.014 & 0.083±0.016 & 0.082±0.015 & 0.157±0.057 & 0.163±0.070 & 0.161±0.065 & 0.33±0.30 & 0.29±0.22 & 0.29±0.21 \\
\hline
\rowcolor{lightgray!30}
\multicolumn{2}{c}{\textsc{DistA-LP}$^\dagger$} & 2.03$\pm$0.16 & 2.07$\pm$0.20 & 2.15$\pm$0.33 & 4.68$\pm$0.36 & 4.55$\pm$0.11 & 4.54$\pm$0.19 & 2.75$\pm$0.21 & 2.74$\pm$0.36 & 3.71$\pm$0.52 \\
\midrule
\multicolumn{11}{c }{London road map}\\
 \hline
Pre-defined& EM & $\leq 0.05$ & $\leq 0.05$ & $\leq 0.05$ & $\leq 0.05$ & $\leq 0.05$ & $\leq 0.05$ & $\leq 0.05$ & $\leq 0.05$ & $\leq 0.05$ \\
\cline{2-2}
Noise Distribution & Laplace  & $\leq 0.05$ & $\leq 0.05$ & $\leq 0.05$ & $\leq 0.05$ & $\leq 0.05$ & $\leq 0.05$ & $\leq 0.05$ & $\leq 0.05$ & $\leq 0.05$ \\
\cline{1-2}
Hybrid  & RMP & 0.03±0.00 & 0.03±0.00 & 0.03±0.00 & 0.20±0.01 & 0.18±0.00 & 0.22±0.01 & 0.47±0.04 & 0.58±0.07 & 0.55±0.03  \\
\cline{2-2}
Method &  COPT  & $\geq 1,800$ & $\geq 1,800$ & $\geq 1,800$ & $\geq 1,800$ & $\geq 1,800$ & $\geq 1,800$ & $\geq 1,800$ & $\geq 1,800$ & $\geq 1,800$ \\
\cline{1-2}
Optimization &  LP  & $\geq 1,800$ & $\geq 1,800$ & $\geq 1,800$ & $\geq 1,800$ & $\geq 1,800$ & $\geq 1,800$ & $\geq 1,800$ & $\geq 1,800$ & $\geq 1,800$ \\
\cline{2-2}
Based & LP-A & 41.8±0.51 & 42.9±1.99 & 42.1±0.52 & 83.9±2.06 & 80.7±0.94 & 87.2±2.87 & 127.6±9.8 & 145.8±13.7 & 144.1±2.6 \\
\cline{2-2} 
 Methods & PAnDA & 0.09±0.01 & 0.09±0.00 & 0.09±0.01 & 0.14±0.02 & 0.14±0.02 & 0.15±0.03 & 0.24±0.04 & 0.24±0.06 & 0.24±0.03 \\
\hline
\rowcolor{lightgray!30}
\multicolumn{2}{c}{\textsc{DistA-LP}$^\dagger$} & 2.08$\pm$0.05 & 3.89$\pm$2.76 & 3.21$\pm$0.12 & 0.51$\pm$0.06 & 0.77$\pm$0.15 & 0.67$\pm$0.01 & 2.22$\pm$0.23 & 2.31$\pm$0.37 & 2.23$\pm$0.28 \\
\midrule
\multicolumn{11}{c }{New York City road map}\\
\hline
Pre-defined& EM & $\leq 0.05$ & $\leq 0.05$ & $\leq 0.05$ & $\leq 0.05$ & $\leq 0.05$ & $\leq 0.05$ & $\leq 0.05$ & $\leq 0.05$ & $\leq 0.05$ \\
\cline{2-2}
Noise Distribution & Laplace  & $\leq 0.05$ & $\leq 0.05$ & $\leq 0.05$ & $\leq 0.05$ & $\leq 0.05$ & $\leq 0.05$ & $\leq 0.05$ & $\leq 0.05$ & $\leq 0.05$ \\
\cline{1-2}
Hybrid  & RMP & $\leq 0.05$ & $\leq 0.05$ & $\leq 0.05$ & 0.18±0.00 & 0.23±0.01 & 0.22±0.01 & 0.42±0.01 & 0.45±0.00 & 0.53±0.01 \\
\cline{2-2}
Method &  COPT  & $\geq 1,800$ & $\geq 1,800$ & $\geq 1,800$ & $\geq 1,800$ & $\geq 1,800$ & $\geq 1,800$ & $\geq 1,800$ & $\geq 1,800$ & $\geq 1,800$ \\
\cline{1-2}
Optimization &  LP  & $\geq 1,800$ & $\geq 1,800$ & $\geq 1,800$ & $\geq 1,800$ & $\geq 1,800$ & $\geq 1,800$ & $\geq 1,800$ & $\geq 1,800$ & $\geq 1,800$ \\
\cline{2-2}
Based & LP-A & 15.52±0.50 & 15.84±0.52 & 18.03±1.13 & 30.19±0.52 & 37.70±1.82 & 33.62±0.39 & 45.09±1.23 & 46.00±1.05 & 46.92±1.22 \\
\cline{2-2} 
Methods & PAnDA & 0.089±0.015 & 0.090±0.022 & 0.089±0.018 & 0.13±0.017 & 0.14±0.0063 & 0.15±0.021 & 0.24±0.078 & 0.23±0.073 & 0.23±0.084 \\
\hline
\rowcolor{lightgray!30}
\multicolumn{2}{c}{\textsc{DistA-LP}$^\dagger$} &  0.92$\pm$0.02 & 0.94$\pm$0.05 & 0.93$\pm$0.05 & 18.28$\pm$0.07 & 18.57$\pm$0.18 & 18.70$\pm$0.18 & 2.77$\pm$0.31 & 2.81$\pm$0.37 & 2.80$\pm$0.45 \\
\bottomrule
    \end{tabular}
\label{tab:time}
\vspace{-0.20in}
\end{table*}

\section{Related Work}
\label{sec:related}

Since the introduction of mDP by \cite{Chatzikokolakis-PETS2013}, substantial research efforts have focused on developing data perturbation mechanisms. These mechanisms generally fall into two categories: \emph{predefined noise distribution mechanisms} and \emph{optimization-based mechanisms}.

\smallskip
\noindent \textbf{Predefined Noise Distribution Mechanisms.} 
The Laplace mechanism is a classic perturbation mechanism, originally proposed for DP~\cite{Dwork-TC2006} and later adapted to mDP by Chatzikokolakis et al.~\cite{Chatzikokolakis-PETS2013}. Instead of scaling noise to global sensitivity, it uses the metric distance $d_{\mathcal{X}}$ between records and sets the noise scale to $\frac{d_{\mathcal{X}}}{\epsilon}$, so that outputs for nearby records remain statistically indistinguishable.
This adaptation underlies Geo-Indistinguishability~\cite{Andres-CCS2013} for location privacy: early work focused on sporadic locations~\cite{Chatzikokolakis-PETS2013, Andres-CCS2013}, while later extensions considered continuous tracking, real-time trajectories, and personalized budgets~\cite{chatzikokolakis2014predictive, hua2017geo, yu2023privacy, min2023personalized}. Despite its efficiency~\cite{Carvalho2021TEMHU}, additive Laplace noise is poorly suited to structured domains (e.g., graphs, embeddings), where it can distort structure or ignore density. For word embeddings, perturbation may produce invalid vectors, requiring nearest-neighbor projection~\cite{fernandes2018author, Feyisetan-ICDM2019}, which in turn assumes the embedding space itself is non-sensitive and can introduce privacy risks during neighbor retrieval.  \looseness = -1

EM~\cite{Carvalho2021TEMHU, ImolaUAI2022} addresses these limitations by selecting outputs from discrete sets via utility-driven quality scores, bypassing continuous perturbation. This proves particularly effective in location privacy (trajectory protection~\cite{Chatzikokolakis-PoPETs2015, Yu-NDSS2017}; spatial crowdsourcing~\cite{gursoy2019secure, niu2020eclipse, ren2022distpreserv, dong2019preserving}), textual data (word substitution~\cite{wang2017local, Feyisetan-WDSM2020}; embeddings~\cite{Carvalho2021TEMHU, meisenbacher20241}), and graph data (edge preservation~\cite{raskhodnikova2016lipschitz}; query responses~\cite{fioretto2019privacy, kamalaruban2020not}).  
While EM offers greater flexibility and often superior utility compared to the Laplace mechanism, it bases selection on perturbation magnitude $d_{x,y}$, which may not fully capture directional variations in utility loss across the output space. Therefore, many recent efforts have prioritized \emph{optimization-based mechanisms}, which minimize data utility loss by explicitly modeling both the magnitude and directional impact of each candidate perturbation.

\smallskip
\noindent \textbf{Optimization (or LP)-based Mechanisms.}  
Since the seminal work by Bordenabe et al.~\cite{Bordenabe-CCS2014}, which introduced the use of LP to optimize perturbation mechanisms under mDP, significant effort has been devoted to developing LP-based mechanisms. These approaches generally aim to minimize utility loss by explicitly considering both the magnitude and direction of the loss for each possible perturbation. Such mechanisms are typically formalized as stochastic matrices, with the LP objective minimizing the expected utility loss subject to mDP constraints. However, these formulations involve $O(N^2)$ decision variables, making them theoretically optimal but computationally prohibitive for large domains. Many existing implementations (e.g., \cite{Bordenabe-CCS2014, Wang-ICDM2016, Yu-NDSS2017, Wang-WWW2017}) are limited to relatively small domains with $N \leq 100$. Although Bordenabe et al.~\cite{Bordenabe-CCS2014} proposed using spanners to approximate $d_{\mathcal{X}}$ via shortest path distances in a graph, the scalability remains limited, for example, solving the LP takes over 1,800 seconds when $N = 400$ \cite{ImolaUAI2022}.

Recent research has prioritized improving the scalability of LP-based methods. Ahuja et al.~\cite{AhujaEDBT2019} proposed a hierarchical index structure to protect location data under Geo-Indistinguishability (Geo-Ind), while others leveraged decomposition techniques, such as Dantzig-Wolfe~\cite{Qiu-TMC2022,Qiu-CIKM2020} and Benders decomposition~\cite{Qiu-IJCAI2024}, to partition large-scale LP problems into tractable subproblems. Hybrid frameworks combining LP with probabilistic mechanisms have also emerged, adapting noise using local utility/sensitivity. For example, Imola et al.~\cite{ImolaUAI2022} integrated the weighted EM with selective LP optimization on sensitive inputs. Despite these advances, these state-of-the-art methods still only handle domains with $N \leq 1{,}000$. 
A more recent work by Qiu et al.~\cite{Qiu-PETS2025} introduces a locality-aware optimization approach, restricting computation to spatial neighborhoods. While scaling to $N = 1{,}600$, it does so at the cost of strict compliance with mDP, as the use of neighborhood-specific parameters introduces potential information leakage. 

\vspace{-0.05in}
\section{Conclusions}
\label{sec:conclusion}
\vspace{-0.05in}

We introduced \textsc{DistA-LP}, a privacy-preserving distributed LP framework for scalable mDP optimization with explicit upload-stage privacy accounting. \textsc{DistA-LP} replaces raw user-specific LP coefficients with noisy surrogate parameters, composes upload leakage with the downstream perturbation mechanism for end-to-end $(\epsilon,\delta)$ guarantees, and uses privacy-free post-processing plus Benders-style decomposition for scalable, geometrically consistent optimization. Experiments on Rome, London, and NYC road-network datasets show that \textsc{DistA-LP} achieves lower utility loss, near-zero reported mDP violations, and substantially faster runtime than monolithic LP and hybrid baselines. These results demonstrate that surrogate parameterization and probabilistic auditing can protect upload-stage privacy without sacrificing utility or formal guarantees.

Promising avenues include deriving tighter analytic bounds for asymmetric parameter noise to complement the PmDP auditor; exploring richer surrogate families (e.g., kernelized or learned surrogates) with privacy-aware regularization; extending the framework beyond geospatial domains to text, graphs, and multimodal metrics; handling adversarial or colluding participants in stronger threat models; and integrating adaptive privacy budgeting across time or tasks. Finally, we believe \textsc{DistA-LP} provides a principled foundation for scalable, auditable mDP optimization in distributed settings.


\section*{Ethical Considerations}
We reviewed the venue's ethics guidance, submission instructions, and ethics-document requirements. This work was conducted as a simulation-based study using public infrastructure-level data, and we believe the research process poses minimal direct risk to individuals. At the same time, because the paper concerns privacy mechanisms for sensitive records, especially location data, we explicitly consider potential downstream risks and limitations.

\paragraph{Stakeholders}
The primary stakeholders include: (i) researchers studying metric differential privacy (mDP), distributed optimization, and privacy auditing; (ii) practitioners and platform operators who may deploy mDP mechanisms in navigation, mobility analytics, spatial crowdsourcing, or related services; (iii) end users and data subjects whose locations or other sensitive records may be protected by such mechanisms; (iv) data and infrastructure providers such as OpenStreetMap contributors; and (v) society at large, including communities that may face heightened risks from location tracking, surveillance, or privacy-invasive analytics. The research process itself has limited direct impact on these stakeholders, while publication and deployment may influence how distributed mDP mechanisms are designed, audited, and used in practice.

\paragraph{Impacts of the research process and publication}
Regarding the \emph{research process}, our evaluation uses publicly available OpenStreetMap road-network graphs and synthetic perturbation experiments. We do not collect, process, or attempt to infer personal or identifiable user-level trajectories, and we do not interact with live users or deployed systems. Thus, the risk of direct tangible harms or rights-based harms to individuals during the research process is minimal.

Regarding the \emph{publication and potential deployment} of the results, the intended impact is positive. DISTA-LP is designed to make LP-based mDP optimization more scalable while explicitly accounting for privacy leakage from distributed coefficient uploads. In particular, the framework replaces raw user-specific LP coefficient blocks with perturbed surrogate parameters, uses probabilistic mDP auditing to quantify upload-stage leakage, and composes this leakage with the downstream perturbation mechanism under an end-to-end privacy budget. These contributions can benefit researchers by clarifying a previously under-accounted leakage channel, practitioners by providing an auditable design pattern for distributed mDP optimization, and data subjects by strengthening protections in settings where local distance or utility structures may themselves reveal sensitive information.

At the same time, several risks remain. First, the framework could be misused to justify weak privacy configurations, such as overly large $\epsilon$ or $\delta$ values, poorly chosen metrics, or insufficient sampling in the PmDP auditor. Second, practitioners may over-interpret the empirical guarantees outside the stated assumptions, such as applying the method under stronger adversaries, colluding users, or domains whose metric does not reflect meaningful privacy harms. Third, because the paper highlights leakage from coefficient uploads, careless readers might focus on this vulnerability without adopting the proposed mitigation. These risks are particularly important for high-risk location-based applications and for populations that may be disproportionately affected by surveillance or privacy failures.

\paragraph{Mitigations and residual risks}
To mitigate these risks, we clearly document the threat model, including the honest-but-curious server assumption, honest users, and the exclusion of collusion from the main analysis. We also distinguish exact mDP from probabilistic mDP and emphasize that the upload-stage privacy budget must be explicitly allocated and composed with the downstream perturbation budget. We recommend that practitioners use conservative privacy budgets, validate metric choices with respect to the relevant application and affected population, perform adequate privacy auditing before deployment, and avoid deploying the method in high-stakes settings without additional legal, ethical, and domain-specific review. Our released artifacts are limited to code and public or derived experimental data needed for reproducibility; they do not include personal trajectories or sensitive raw user data. Despite these mitigations, residual risk remains that third parties may ignore the assumptions, select weak parameters, or deploy the method in inappropriate contexts.

\paragraph{Decision to conduct and publish the research}
We judged that conducting this research is ethically justified because it uses public infrastructure-level data and synthetic experiments, while addressing an important privacy risk in distributed mDP optimization: leakage from user-specific coefficient uploads. Publishing the work provides a net benefit by enabling community scrutiny, improving reproducibility, and giving practitioners concrete guidance for designing and auditing privacy-preserving distributed LP mechanisms. We believe that making the assumptions, limitations, and safeguards explicit is preferable to leaving upload-stage leakage unexamined in future distributed privacy systems.






%



\normalsize
\appendix
This appendix provides supplementary details supporting the main paper, including extended derivations, detailed algorithm design, and complete results omitted from the main text due to space constraints. In particular, it presents the mathematical details of the proposed DistA-LP framework, privacy-loss accounting procedures, surrogate approximation methods, and additional empirical analyses used to validate scalability, privacy preservation, and utility performance.

\setcounter{tocdepth}{2}  

\subsection{Mathematical Notations}



\begin{table}[h]
\centering
\caption{Notation summary}
\label{tab:notation}
\small
\renewcommand{\arraystretch}{1.05}
\begin{tabularx}{\columnwidth}{@{}p{0.21\columnwidth} X@{}}
\toprule
\multicolumn{2}{@{}l}{\textbf{Sets \& Indices}} \\
\midrule
$\mathcal{X}$ & Secret/input space; records $x_i \in \mathcal{X}$. \\
$\mathcal{Y}$ & Output/reported space; $y_k \in \mathcal{Y}$. \\
$\mathcal{M}$ & Perturbation mechanism; $\mathcal{M}:\mathcal{X}\to\mathcal{Y}$ \\
$\mathcal{T}$ & Index set of coefficient-block types; $\mathcal{T}=\{\textsf{dx}, \textsf{dy}, \textsf{cy}\}$ \\
$d(\cdot,\cdot)$ & Metric on $\mathcal{X}$; $d_{i,j}=d_{x_i,x_j}$. \\
$M$ & Number of users; $m\in\{1,\dots,M\}$. \\
$\mathcal{N}_m$ & User-$m$ local neighborhood (relevance subset). \\
$\mathcal{H}_m$ & ``Far'' set for $m$ handled via WEM structure. \\
$i,j$ & Record indices in $\mathcal{X}$. \\
$k$ & Output index in $\mathcal{Y}$. \\
\midrule
\addlinespace[2pt]
\multicolumn{2}{@{}l}{\textbf{Matrices, Variables, and Parameters}} \\
\midrule
$\mathbf{Z}\in\mathbb{R}_{\ge0}^{|\mathcal{X}|\times |\mathcal{Y}|}$ & Mechanism matrix; row $i$ is $z_i^\top$ and $\sum_k z_{i,k}=1$. \\
$z_{i,k}$ & $\Pr[y_k \mid x_i]$. \\
$\mathbf{A}_m^{\textsf{dx}}$ & LP coefficients for indistinguishability over $(i,j)\in\mathcal{N}_m^2$. \\
$\mathbf{A}_m^{\textsf{dy}}$ & Distance-based LP coefficients for $(i,k)\in\mathcal{N}_m\times\mathcal{Y}$. \\
$\mathbf{A}_m^{\textsf{cy}}$ & Cost-based LP coefficients for $(i,k)\in\mathcal{N}_m\times\mathcal{Y}$. \\
$\mathbf{A}$ & Generic placeholder for a coefficient matrix above. \\
$f_{\boldsymbol{\Theta}}$ & Coefficient surrogate with parameters $\boldsymbol{\Theta}$. \\
$\boldsymbol{\Theta},\widetilde{\boldsymbol{\Theta}}$ & Unperturbed / perturbed surrogate parameters. \\
$\widetilde{\mathbf{A}}=f(\cdot;\widetilde{\boldsymbol{\Theta}})$ & Reconstructed approximate coefficients on the server. \\
$r$ & Surrogate rank (SVD/low-rank). \\
$\mathbf{U},\boldsymbol{\Sigma},\mathbf{V}$ & Truncated SVD factors: $\mathbf{A}\approx\mathbf{U}\boldsymbol{\Sigma}\mathbf{V}^\top$. \\
$\mathbf{P},\mathbf{Q}$ & Low-rank factors: $\mathbf{A}\approx\mathbf{P}\mathbf{Q}^\top$ (uploads may use $\widetilde{\mathbf{P}},\widetilde{\mathbf{Q}}$). \\
$w$ & WEM weights (master problem). \\
$v_m$ & User-$m$ subproblem value in Benders cuts. \\
\midrule
\addlinespace[2pt]
\multicolumn{2}{@{}l}{\textbf{Objectives, Constraints, and Privacy Metrics}} \\
\midrule
$\mathcal{F}_{\mathrm{LP}}(\cdot)$ & Local LP over $(\mathcal{N}_m,\mathcal{Y})$ with $(\mathbf{D},\mathbf{C})$ and $(\mathbf{A}^d,\mathbf{A}^c)$. \\
$\Omega(\cdot)$ & Feasible region (local/master LP). \\
$\epsilon$ & mDP privacy parameter. \\
$\delta$ & Failure probability for probabilistic mDP (PmDP). \\
$\mathcal{L}_{x_i,x_j}(y)$ & $\ln\frac{\Pr[\mathcal{M}(x_i)=y]}{\Pr[\mathcal{M}(x_j)=y]}\big/ d_{x_i,x_j}$ (privacy loss). \\
$\mathrm{ViolRatio}$ & Fraction of violations of $\epsilon$-mDP. \\
$\mathrm{RelErr}$ & Relative error vs.true coefficients/utilities. \\
\bottomrule
\end{tabularx}
\end{table}

\subsection{Omitted Proofs}

\subsection*{\textbf{Proof of Proposition \ref{prop:composition}}}
\label{subsec:proof:composition}

\begin{reproposition}
\label{thm:composition}
(Sequential composition of PmDP, adapted from
\cite{Liu-CCS2025})
Let $\{\mathcal{M}_t^{\textsf{p}}\}_{t\in\mathcal{T}}$
be randomized mechanisms with conditionally
independent outputs given the input.
If each $\mathcal{M}_t^{\textsf{p}}$ satisfies
$(\epsilon_t,\delta_t)$-PmDP and
$\sum_{t\in\mathcal{T}}\delta_t<1$,
then the joint mechanism
$\mathcal{M}=(\mathcal{M}_t^{\textsf{p}})_{t\in\mathcal{T}}$
satisfies
$(\sum_{t\in\mathcal{T}}\epsilon_t,
\sum_{t\in\mathcal{T}}\delta_t)$-PmDP.
\end{reproposition}

\begin{proof}
Let $\mathbf{z}=(z_t)_{t\in\mathcal{T}}$ denote
the joint output of $\mathcal{M}$.
Since the mechanisms are conditionally
independent given the input, the privacy
loss of the joint mechanism satisfies
\begin{equation}
\begin{aligned}
\mathrm{PL}^{\mathcal{M}}_{x_i,x_j}(\mathbf{z})
&=
\log\frac{\Pr[\mathcal{M}(x_i)=\mathbf{z}]}
{\Pr[\mathcal{M}(x_j)=\mathbf{z}]}\\
&=
\sum_{t\in\mathcal{T}}
\mathrm{PL}^{(t)}_{x_i,x_j}(z_t),
\end{aligned}
\end{equation}
where $\mathrm{PL}^{(t)}_{x_i,x_j}(z_t)$
denotes the privacy loss of
$\mathcal{M}_t^{\textsf{p}}$.

By the union bound, for any input pair
$(x_i,x_j)$, we have
\begin{equation}
\begin{aligned}
&\Pr\left[
\mathrm{PL}^{\mathcal{M}}_{x_i,x_j}(\mathbf{z})
>\sum_{t\in\mathcal{T}}\epsilon_t
\right]\\
&\quad\leq
\sum_{t\in\mathcal{T}}
\Pr\left[
\mathrm{PL}^{(t)}_{x_i,x_j}(z_t)
>\epsilon_t
\right].
\end{aligned}
\end{equation}

Averaging over all input pairs yields
\begin{equation}
\begin{aligned}
&\frac{1}{|\mathcal{X}|^2}
\sum_{x_i,x_j\in\mathcal{X}}
\Pr\left[
\mathrm{PL}^{\mathcal{M}}_{x_i,x_j}(\mathbf{z})
>\sum_{t\in\mathcal{T}}\epsilon_t
\right]\\
&\quad\leq
\sum_{t\in\mathcal{T}}
\frac{1}{|\mathcal{X}|^2}
\sum_{x_i,x_j\in\mathcal{X}}
\Pr\left[
\mathrm{PL}^{(t)}_{x_i,x_j}(z_t)
>\epsilon_t
\right]\\
&\quad\leq
\sum_{t\in\mathcal{T}}\delta_t,
\end{aligned}
\end{equation}
where the last inequality follows from
the $(\epsilon_t,\delta_t)$-PmDP guarantee
of each $\mathcal{M}_t^{\textsf{p}}$.

Therefore, by Definition~\ref{def:PmDP},
$\mathcal{M}$ satisfies
$(\sum_{t\in\mathcal{T}}\epsilon_t,
\sum_{t\in\mathcal{T}}\delta_t)$-PmDP.
\end{proof}

\subsection*{\textbf{Proof of Proposition \ref{prop:sampling} (Concentration Guarantees)}}
\label{subsec:proof:sampling}

\begin{reproposition}[Concentration Guarantee for Probabilistic mDP Sampling] 
\label{reprop:sampling}
 
If the empirical violation rate satisfies $\hat{p}_{\mathcal{S}_{\ell}} = \xi \delta$ for some constant $\xi < 1$, then $\Pr\left[p_{\mathcal{X}_{\ell}^2} \leq \delta\right] \geq 1 - 2\exp(-2S_\ell(1-\xi)^2\delta^2)$. 
\end{reproposition}

\begin{proof}
Suppose the empirical violation rate satisfies $\hat{p}_{\mathcal{S}} = \xi \delta$ for some constant $\xi < 1$.  
Our goal is to show that, with high probability, the true violation probability $p_{\mathcal{X}^2}$ is at most $\delta$.

First, by Hoeffding's inequality, for any $t > 0$:
\begin{equation}
\small \Pr\left[\left|\hat{p}_{\mathcal{S}} - p_{\mathcal{X}^2}\right| > t\right] \leq 2e^{-2S_\ell t^2}.
\end{equation}
Let $t = (1-\xi)\delta$. Then:
\begin{equation}
\small \Pr\left[\left|\hat{p}_{\mathcal{S}} - p_{\mathcal{X}^2}\right| > (1-\xi)\delta\right] \leq 2e^{-2S_\ell(1-\xi)^2\delta^2}.
\end{equation}
This implies:
\begin{equation}
\small \Pr\left[\left|\hat{p}_{\mathcal{S}} - p_{\mathcal{X}^2}\right| \leq (1-\xi)\delta\right] \geq 1 - 2e^{-2S_\ell(1-\xi)^2\delta^2}.
\end{equation}
Now, under the event that $\left|\hat{p}_{\mathcal{S}} - p_{\mathcal{X}^2}\right| \leq (1-\xi)\delta$, we have:
\begin{equation}
\small p_{\mathcal{X}^2} \leq \hat{p}_{\mathcal{S}} + (1-\xi)\delta = \xi\delta + (1-\xi)\delta = \delta.
\end{equation}
Thus:
\begin{equation}
\small \Pr\left[p_{\mathcal{X}^2} \leq \delta\right] \geq 1 - 2e^{-2S_\ell(1-\xi)^2\delta^2},
\end{equation}
which completes the proof.
\end{proof}

\vspace{0.2em}
\subsection{Detailed Algorithm Design}
\label{app:algorithm_design}
\vspace{0.4em}
\subsection*{\textbf{Laplace Perturbation Baseline for Coefficient Blocks}}
\label{app:laplace_mdp}
\vspace{0.4em}

\noindent We instantiate a naive additive-noise baseline by releasing
\[
\mathcal{M}_t(x) = \mathrm{vec}\!\left(\mathbf{A}_m^{t}(x)\right) + \eta,
\qquad
\eta \sim \text{i.i.d. }\mathrm{Lap}(0,b_t),
\]
for $t\in\{\textsf{dx},\textsf{dy},\textsf{cy}\}$.
For i.i.d.Laplace noise, the privacy loss between two secrets $x_i,x_j$ satisfies
\[
\mathrm{PL}_{x_i,x_j}(y)
= \log\frac{p(y\mid x_i)}{p(y\mid x_j)}
\;\le\;
\frac{\big\|\mathrm{vec}(\mathbf{A}_m^{t}(x_i))-\mathrm{vec}(\mathbf{A}_m^{t}(x_j))\big\|_1}{b_t}.
\]
To align with $\epsilon$-mDP, i.e., $\mathrm{PL}_{x_i,x_j}(y)\le \epsilon\, d(x_i,x_j)$ for all $i,j,y$,
it suffices to set $b_t=\Delta_t/\epsilon$, where the metric (Lipschitz) sensitivity is
\[
\Delta_t \triangleq \max_{i\neq j}\;
\frac{\big\|\mathrm{vec}(\mathbf{A}_m^{t}(x_i))-\mathrm{vec}(\mathbf{A}_m^{t}(x_j))\big\|_1}{d(x_i,x_j)}.
\]
In our experiments, we fix $\epsilon=4$ and apply entrywise perturbation
$\mathrm{Lap}\!\left(0,\frac{\Delta_t}{\epsilon}\right)$ to obtain $\widetilde{\mathbf{A}}$,
then compute the induced privacy loss $\mathrm{PL}_{x_i,x_j}\!\left(\widetilde{\mathbf{A}}\right)$.
The resulting privacy-loss distributions are reported in Fig.~\ref{fig:motivatingexample}(a)--(c).

\vspace{0.9em}
\subsection*{\textbf{Coefficient Matrix Approximation}}
\vspace{0.5em}

\noindent Joint Optimization of Permutation $\{\mathbf{E}_m,\boldsymbol{\Theta}_m^{t}\}_{m=1}^M$:

\begin{itemize}
\item [(1)] \emph{Parameter fitting.}
Holding $\mathbf{E}_m$ fixed, for each surrogate family we fit $\boldsymbol{\Theta}_m^t$ by minimizing the weighted entrywise $\ell_1$ reconstruction loss in~\eqref{eq:block_fit_obj} using the \emph{Nelder--Mead simplex algorithm} \cite{NelderMead1965}, a robust, derivative-free optimizer well-suited for nonconvex objectives and nonsmooth losses. This choice avoids instabilities we observed with gradient-based methods under discretization and privacy-induced perturbations. 

For the Gaussian family, we impose implicit box constraints to ensure physically meaningful spreads $(s_x,s_y)$ and nonnegative amplitude $(A\ge 0)$, avoiding degenerate fits.
For the polynomial family, we add mild ridge regularization to improve conditioning and reduce overfitting, i.e., we minimize the $\ell_1$ reconstruction loss plus $\rho\|\boldsymbol{\theta}\|_2^2$ over polynomial coefficients.
For the RBF family, we restrict the bandwidth $\sigma$ to a bounded range for numerical stability and softly regularize the basis weights $\alpha_k$ to prevent over-concentration and preserve smoothness across neighboring centers.

\item [(2)] \emph{Permutation search.} With $\{\boldsymbol{\Theta}_m^t\}$ fixed, we refine the local permutation $\mathbf{E}_m$ via greedy pairwise index swaps $(i,j)$ that decrease the global objective in~\eqref{eq:block_fit_obj}.
We accept only swaps that reduce the total weighted $\ell_1$ loss, ensuring monotonic descent.
After each accepted swap, we refit $\boldsymbol{\Theta}_m^t$ under the updated ordering for consistency.
This alternating scheme converges to a locally optimal configuration in practice and typically stabilizes within a few iterations.
\end{itemize}

\subsection*{- \, Detailed Description of Matrix Approximation Methods}
\begin{itemize}
    \item \emph{2D Gaussian.} Parameters $\Theta=\{c_x,c_y,s_x,s_y,A,b\}$, where $(c_x,c_y)$ is the center, $(s_x,s_y)$ are the spreads, $A\ge 0$ is the amplitude, and $b$ is an offset. Widths are clipped to a reasonable range to avoid degeneracy. This very compact model captures smooth, bell-shaped structure and aggressively denoises.

    \item \emph{2D Polynomial.} Parameters
    \begin{equation}
    \Theta=\{\beta_0,\beta_{1x},\beta_{1y},\beta_{2x},\beta_{2y},\beta_{xy},\beta_{3x},\beta_{3y},\beta_{2xy},\beta_{xy}^2\},
    \end{equation}
    including an intercept $\beta_0$, first-order slopes $(\beta_{1x},\beta_{1y})$, second-order curvature terms $(\beta_{2x},\beta_{2y},\beta_{xy})$, and higher-order cross-interaction coefficients $(\beta_{3x},\beta_{3y},\beta_{2xy},\beta_{xy}^2)$. We use a low-degree bivariate basis on $[0,1]^2$ (e.g., a Legendre-product basis for numerical stability) and apply $\ell_2$ regularization to mitigate overfitting. This family captures smooth global trends with a moderate parameter count.

    \item \emph{Radial Basis Function (RBF).} Parameters
    \begin{equation}
    \Theta=\{\alpha_1,\ldots,\alpha_K,\mu_{1x},\mu_{1y},\ldots,\mu_{Kx},\mu_{Ky},\sigma,b\},
    \end{equation}
    where $\alpha_k$ are basis weights, $(\mu_{kx},\mu_{ky})$ are centers, $\sigma$ is a shared bandwidth, and $b$ is an offset. Centers are fixed on a coarse grid with shared bandwidth; only the weights are optimized. This yields a localized, multi-peak approximation that remains stable under mild regularization.
\end{itemize}

\begin{table*}[t]
\caption{Comparison of different matrix approximation methods.}
\centering
\small
\begin{tabular}{p{1.4cm}p{2.4cm}p{1.5cm}p{1.2cm}p{3.8cm}p{4.8cm}}
\toprule
\rowcolor{lightgray!30} \textbf{Surrogate} & \textbf{Params in} $\boldsymbol\Theta_m^t$ &  \textbf{Size of} $\boldsymbol\Theta_m^t$ & \textbf{Accuracy} & \textbf{Notes} & \textbf{Noise Strategy} \\
\hline
2D \newline Gaussian & $c_x$, $c_y$, $s_x$, $s_y$, $a$, $b$ & 6 & Low & Compact, but only fits symmetric bell-shaped data well. & Subtract Laplace noise from amplitude ($a$) and increase spread ($\sigma$) to lower the surface everywhere. \\
\hline
2D \newline Polynomial & $\beta_0$, $\beta_{1x}$, $\beta_{1y}$, $\beta_{2x}$, $\beta_{2y}$, $\beta_{xy}$, $\beta_{3x}$, $\beta_{3y}$, $\beta_{2xy}$, $\beta_{xy}^2$ & 10 & Medium & Simple to fit; more effective for smooth surface trends. & Subtract Laplace noise from coefficients, especially the intercept, to bias the surface downward. \\
\hline
RBF & $\alpha_1, \ldots, \alpha_K$,  $\mu_{1x}, \mu_{1y}, \ldots, \mu_{Kx}, \mu_{Ky}$, $\sigma$, $b $ & 25 & High & Flexible and localized, works well for multi-peak or irregular matrices. & Add negative-biased noise to weights ($\alpha_i - |\text{Lap}(b)|$); optionally perturb centers for robustness. \\
\hline
Low-Rank SVD & $\mathbf{U}_k$, $\boldsymbol\Sigma_k$, $\mathbf{V}_k$ & 30--50 & Very High & Captures global structure with fewer components; efficient to store and reconstruct. & Subtract noise from singular values ($S_k$) or from the reconstructed matrix directly to ensure conservative output. \\
\bottomrule
\end{tabular}
\label{tab:approx_privacy}
\end{table*}

\subsection*{- \, Low-Rank SVD}
We approximate each coefficient matrix $\mathbf{A}\in\mathbb{R}^{m\times n}$ with a rank-$k$ surrogate that preserves its dominant structure while drastically reducing dimensionality. Compute the SVD $\mathbf{A}=\mathbf{U}\boldsymbol{\Sigma}\mathbf{V}^{\top}$, where $\mathbf{U}\in\mathbb{R}^{m\times m}$ and $\mathbf{V}\in\mathbb{R}^{n\times n}$ are orthogonal, and $\boldsymbol{\Sigma}=\mathrm{diag}(\sigma_1,\ldots,\sigma_r)$ with $\sigma_1\ge\cdots\ge\sigma_r>0$ and $r=\mathrm{rank}(\mathbf{A})$. The truncated SVD retains the top-$k$ components,
\begin{equation}
\mathbf{A}_k = \mathbf{U}_k\boldsymbol{\Sigma}_k\mathbf{V}_k^{\top},\qquad k\ll\min\{m,n\},
\end{equation}
where $\mathbf{U}_k\in\mathbb{R}^{m\times k}$, $\mathbf{V}_k\in\mathbb{R}^{n\times k}$, and $\boldsymbol{\Sigma}_k=\mathrm{diag}(\sigma_1,\ldots,\sigma_k)$. For compact storage and upload, we use the factorized parameterization
\begin{equation}
\boldsymbol{\Theta}=(\mathbf{P},\mathbf{Q}),\quad \mathbf{P}:=\mathbf{U}_k\boldsymbol{\Sigma}_k\in\mathbb{R}^{m\times k},\quad \mathbf{Q}:=\mathbf{V}_k\in\mathbb{R}^{n\times k},
\end{equation}
so that $\mathbf{A}_k=\mathbf{P}\mathbf{Q}^{\top}$. This reduces the representation from $mn$ entries to $k(m{+}n)$ parameters and concentrates subsequent privacy noise on a small set of variables.

\medskip
\noindent\emph{Error guarantees and rank selection.}
By the Eckart--Young theorem, $\mathbf{A}_k$ is the best rank-$k$ approximation to $\mathbf{A}$ in both spectral and Frobenius norms:
\begin{equation}
\|\mathbf{A}-\mathbf{A}_k\|_F^2=\sum_{i>k}\sigma_i^2,\qquad
\|\mathbf{A}-\mathbf{A}_k\|_2=\sigma_{k+1}.
\end{equation}
A common choice sets $k$ to capture a target energy fraction
$\sum_{i=1}^{k}\sigma_i^2 \big/ \sum_{i=1}^{r}\sigma_i^2 \ge \tau$ (e.g., $\tau=0.95$), trading a controlled approximation bias for substantial compression. When privacy noise is added, the final approximation error satisfies
\begin{equation}
\|\mathbf{A}-\widetilde{\mathbf{A}}\| \le 
\underbrace{\|\mathbf{A}-\mathbf{A}_k\|}_{\text{truncation}} +
\underbrace{\|\mathbf{A}_k-\widetilde{\mathbf{P}}\widetilde{\mathbf{Q}}^{\top}\|}_{\text{noise \& reconstruction}},
\end{equation}
explicitly showing how $k$ (bias) and noise scale (variance) jointly determine fidelity. Thus, low-rank SVD provides scalability and inherent denoising at the utility cost of a controlled approximation bias.

\vspace{0.9em}
\subsection*{\textbf{Surrogate Parameterization and Parameter Perturbation}}
\label{app:paraperturbation}
\vspace{0.5em}

We present both \emph{predefined parametric surrogates} and \emph{low-rank surrogates} under a single parameterization. 
For each local distance/cost block $\widehat{\mathbf{A}}^t_m$ (after the local permutation in~\eqref{eq:permute_blocks}), we fit a compact parameter vector/tensor $\boldsymbol{\Theta}_m$ and a corresponding reconstruction map $\mathcal{R}$ such that
\begin{equation}
\widehat{\mathbf{A}}^t_m \approx \mathcal{R}(\boldsymbol{\Theta}^t_m).
\label{eq:unified_surrogate}
\end{equation}
We then apply privacy-preserving perturbation directly in \emph{parameter space}, rather than perturbing matrix entries. This preserves the surrogate’s analytic structure and helps avoid undesirable artifacts (e.g., negative distances or loss of symmetry) that can arise from entrywise noise.

Let $\Phi_m:\widehat{\mathbf{A}}_m\mapsto \boldsymbol{\Theta}_m$ denote the deterministic parameter-estimation map, and let
$\Delta_{\Theta}=\max_{\widehat{\mathbf{A}}_m,\widehat{\mathbf{A}}_m'}\|\Phi_m(\widehat{\mathbf{A}}_m)-\Phi_m(\widehat{\mathbf{A}}_m')\|_1$
be its global $\ell_1$ sensitivity under our adjacency notion. Given privacy budget $\epsilon>0$, we release
\begin{equation}
\widetilde{\boldsymbol{\Theta}}_m
=
\mathcal{P}_{\mathcal{D}}\big(\boldsymbol{\Theta}_m + \boldsymbol{\eta}_m\big),
\qquad
\eta_{m,j}\sim \mathrm{Lap}\left(0,\frac{\Delta_{\Theta}}{\epsilon}\right),
\label{eq:unified_param_mech}
\end{equation}
where $\mathcal{P}_{\mathcal{D}}$ projects the noisy parameters to a valid domain $\mathcal{D}$ (e.g., positivity constraints, bounded coefficients, or non-negativity constraints), ensuring numerical stability and well-posed reconstruction.
The server reconstructs a perturbed coefficient block via
\begin{equation}
\widetilde{\mathbf{A}}_m = \mathcal{R}\left(\widetilde{\boldsymbol{\Theta}}_m\right).
\label{eq:unified_reconstruct}
\end{equation}

\paragraph{Instantiation A: predefined parametric functions.}
In this case, $\boldsymbol{\Theta}_m=\{\theta_{m,1},\ldots,\theta_{m,d}\}\in\mathbb{R}^d$ parameterizes a predefined family
$f(\cdot;\boldsymbol{\Theta})$ (2D Gaussian, 2D Polynomial, or RBF), and the reconstruction map is evaluation on the index grid:
\begin{equation}
\big[\mathcal{R}_{\mathrm{para}}(\boldsymbol{\Theta}_m)\big]_{ij}
\triangleq
f(i,j;\boldsymbol{\Theta}_m).
\end{equation}
The projection $\mathcal{P}_{\mathcal{D}}$ enforces family-specific validity constraints (e.g., $\sigma>0$, bounded polynomial coefficients, or $A\ge 0$) so that the released surrogate remains stable and interpretable.

\paragraph{Instantiation B: low-rank surrogates (SVD/factorization).}
In this case, $\boldsymbol{\Theta}_m$ encodes a rank-$k$ factorization, and $\mathcal{R}$ composes the low-rank matrix. Two equivalent parameterizations are convenient:
\begin{itemize}
    \item[(i)] \textbf{SVD factors:} $\boldsymbol{\Theta}_m=(\mathbf{U}_k,\boldsymbol{\Sigma}_k,\mathbf{V}_k)$ with
    $\mathcal{R}_{\mathrm{svd}}(\boldsymbol{\Theta}_m)=\mathbf{U}_k\boldsymbol{\Sigma}_k\mathbf{V}_k^\top$.
    \item[(ii)] \textbf{Rank-$k$ factors:} $\boldsymbol{\Theta}_m=(\mathbf{P},\mathbf{Q})$ with
    $\mathbf{P}:=\mathbf{U}_k\boldsymbol{\Sigma}_k\in\mathbb{R}^{m\times k}$ and $\mathbf{Q}:=\mathbf{V}_k\in\mathbb{R}^{n\times k}$, so that
    $\mathcal{R}_{\mathrm{fac}}(\boldsymbol{\Theta}_m)=\mathbf{P}\mathbf{Q}^\top$.
\end{itemize}
Because $k(m+n)\ll mn$, injecting noise into $\boldsymbol{\Theta}_m$ substantially reduces posterior leakage relative to entrywise perturbation while preserving approximation fidelity.

\paragraph{A practical perturbation rule for SVD surrogates.}
For distance and utility-loss blocks, we perturb only the singular values while keeping $\mathbf{U}_k$ and $\mathbf{V}_k$ fixed to maximize utility per unit of privacy. Let
$s=\mathrm{diag}(\boldsymbol{\Sigma}_k)\in\mathbb{R}_{\ge 0}^{k}$.
We draw i.i.d.Laplace noise $\eta_i\sim \mathrm{Lap}(0,b)$ and apply a downward-biased soft-thresholding:
\begin{equation}
\tilde{s}_i=\max\{s_i-\tau-|\eta_i|,0\},\qquad i=1,\ldots,k.
\label{eq:svd_shrink}
\end{equation}
Here $b=\frac{\Delta_s}{\epsilon}$ with $\Delta_s$ a Lipschitz bound (metric sensitivity) of the singular-value map under our adjacency notion; when a tight analytical $\Delta_s$ is unavailable, we select $b$ via a binary search in our posterior-leakage auditor to satisfy the target $(\epsilon,\delta)$-PmDP constraint. The deterministic shrinkage $\tau=\frac{c}{\epsilon}$ (with $c\in[0.3,1.0]$ tuned on a small validation set) reduces the probability of overestimation (e.g., $\tilde{d}>d$) that would otherwise relax mDP constraints, thereby lowering violation rates without additional privacy spend.

We then form $\widetilde{\boldsymbol{\Sigma}}_k=\mathrm{diag}(\tilde{s})$ and truncate to the conforming rank
$r=\#\{i:\tilde{s}_i>0\}$, yielding
$\widetilde{\mathbf{U}}=\mathbf{U}_k(:,1:r)$,
$\widetilde{\mathbf{V}}=\mathbf{V}_k(:,1:r)$,
and $\widetilde{\boldsymbol{\Sigma}}=\mathrm{diag}(\tilde{s}_{1:r})$.
The uploaded object is $\widetilde{\boldsymbol{\Theta}}_m=(\widetilde{\mathbf{U}},\widetilde{\boldsymbol{\Sigma}},\widetilde{\mathbf{V}})$ (or its equivalent $(\widetilde{\mathbf{P}},\widetilde{\mathbf{Q}})$), and the server-side reconstruction follows from~\eqref{eq:unified_reconstruct}:
$\widetilde{\mathbf{A}}_m=\widetilde{\mathbf{U}}\widetilde{\boldsymbol{\Sigma}}\widetilde{\mathbf{V}}^\top$ (or $\widetilde{\mathbf{A}}_m=\widetilde{\mathbf{P}}\widetilde{\mathbf{Q}}^\top$).

\DEL{
\subsection{Matrix Reordering (MR)}
\todo{Add detailed description of matrix alignment and coarse-grained discretization.}

Our Matrix Reordering proceeds in two stages: a \emph{coarse‑grained discretization} to reduce noise and dimensionality, followed by a \emph{matrix alignment} via permutation optimization.

\paragraph{Coarse‑Grained Discretization} 
To stabilize the subsequent permutation search, we first quantize each continuous target matrix $F_k$ into $G$ uniform bins.  For each entry $(i,j)$ and $k\in\{1,2,3\}$, define
\begin{eqnarray}
  \tilde{F}_k(i,j)
  = \mathrm{round}\Bigl(\frac{F_k(i,j)-\min F_k}{\max F_k - \min F_k}(G-1)\Bigr)\in\{0,\dots,G-1\}.
\end{eqnarray}
Likewise discretize $A_k\to\tilde{A}_k$.  Build a similarity graph on the index set $I$ with edge weights
\begin{eqnarray}
  w_{ip}
  = \sum_{k=1}^3\sum_{j=1}^m
    \bigl|\tilde{A}_k(p,j)-\tilde{F}_k(i,j)\bigr|,
\end{eqnarray}
then cluster indices into coarse blocks by spectral clustering.

\paragraph{Matrix Alignment}
Within each coarse block $g$, let $I_g$ be its original indices and $J_g$ the target positions.  Form the block utility loss matrix
\begin{eqnarray}
  C^{(g)}_{ip}
  = \sum_{j=1}^n\bigl|A_1\bigl(p,\pi_m(j)\bigr)-F_1(i,j)\bigr|
    +\lambda_2\sum_{j=1}^m\bigl|A_2(p,j)-F_2(i,j)\bigr|
    +\lambda_3\sum_{j=1}^m\bigl|A_3(p,j)-F_3(i,j)\bigr|.
\end{eqnarray}
Solve the assignment $I_g\to J_g$ via the Hungarian algorithm to update $\pi_m$.  Iterate over all blocks until the total $L_1$ error
\begin{eqnarray}
  \|PA_1P^\top - F_1\|_1
  + \lambda_2\|PA_2 - F_2\|_1
  + \lambda_3\|PA_3 - F_3\|_1
\end{eqnarray}
converges.
We consider two complementary alignment strategies:

\begin{itemize}
  \item \textbf{Method A: Iterative Pairwise Refinement}  
    We start from identity permutations $P_k = I_n$, $Q_k = I_n$, $Q'_k = I_m$ for each sample $k$.  At each iteration:
    \begin{enumerate}
      \item \emph{Pairwise Alignment:} For every pair of samples $(i,j)$, we compute the utility loss  
      \normalsize
      \begin{eqnarray} 
      &&\underbrace{\left\| P_i A_1^{(i)} Q_i - P_j A_1^{(j)} Q_j \right\|_F^2}_{\text{square matrix (original locations)}} \\
      &+& \lambda_2 \underbrace{\left\| P_i A_2^{(i)} Q_i' - P_j A_2^{(j)} Q_j' \right\|_F^2}_{\text{rectangular matrix (obfuscated)}} \\
      &+& \lambda_3 \underbrace{\left\| P_i A_3^{(i)} Q_i' - P_j A_3^{(j)} Q_j' \right\|_F^2}_{\text{rectangular matrix (utility loss cofficient)}}
      \end{eqnarray}
      \normalsize
      and solve for the updated row and column permutations that reduce pairwise discrepancies $(P_i,Q_i,Q'_i)$ vs.$(P_j,Q_j,Q'_j)$ via the Hungarian algorithm.
      \item \emph{Consensus Update:} We aggregate updated permutations across all pairs e.g.by majority vote or alternating minimization and replace each $P_k,Q_k,Q'_k$.
      \item \emph{Convergence Check:} We repeat until the total Frobenius‐norm loss no longer decreases by more than $\epsilon$.
    \end{enumerate}

  \item \textbf{Method B: Multi‑User $L_1$ Minimization}  
    We view each of $U$ users’ matrices $A^{(u)}_1,A^{(u)}_2,A^{(u)}_3$ as being permuted by $\pi_u$.  Define
    \begin{eqnarray}
      A^{(u)\prime}_1(i,j)=A^{(u)}_1(\pi_u(i),\pi_u(j)),\quad
      A^{(u)\prime}_{2,3}(i,j)=A^{(u)}_{2,3}(\pi_u(i),j).
    \end{eqnarray}
    We then solve
    \begin{eqnarray}
    \min_{\{\pi_u\}_{u=1}^U}
    \sum_{u < v} \left[
    \| A_1^{(u)\prime} - A_1^{(v)\prime} \|_1
    + \lambda_2 \| A_2^{(u)\prime} - A_2^{(v)\prime} \|_1
    + \lambda_3 \| A_3^{(u)\prime} - A_3^{(v)\prime} \|_1
    \right]
    \end{eqnarray}
    typically via block‐coordinate descent over each $\pi_u$ and solving a linear assignment for each update.
\end{itemize}

\subsection{Matrix Approximation (MA) Using Predefined Functions}
\todo{Add detailed description of Gaussian function fitting approach.}

We have a symmetric matrix $A_1 \in \mathbb{R}^{n \times n}$, indexed by a set $I = \{1, 2, \dots, n\}$ on both rows and columns, and two rectangular matrices: $A_2, A_3 \in \mathbb{R}^{n \times m}$, indexed by $I$ on rows and $J = \{1, 2, \dots, m\}$ on columns with three target functions:
    \begin{itemize}
        \item $f_1: [n] \times [n] \rightarrow \mathbb{R}$ for fitting $A_1$,
        \item $f_2: [n] \times [m] \rightarrow \mathbb{R}$ for fitting $A_2$,
        \item $f_3: [n] \times [m] \rightarrow \mathbb{R}$ for fitting $A_3$.
    \end{itemize}
We fit each matrix $A_k$ to known target functions $f_k$ by minimizing an $L_1$ fitting error over permutations. Let $\pi_m: I \rightarrow I$ be a permutation of the index set $I$. Define the permuted matrices as:
\begin{eqnarray}
A_1'(i,j) = A_1(\pi_m(i), \pi_m(j)), \quad
A_2'(i,j) = A_2(\pi_m(i), j), \quad
A_3'(i,j) = A_3(\pi_m(i), j)
\end{eqnarray}
and find the permutation $\pi_m$ that minimizes the total
L1 fitting error:
\begin{eqnarray}
  \min_{\pi_m}
  \sum_{i,j}\Bigl|A_1(\pi_m(i),\pi_m(j)) - f_1(i,j)\Bigr|
  + \lambda_2\sum_{i,j}\Bigl|A_2(\pi_m(i),j) - f_2(i,j)\Bigr|
  + \lambda_3\sum_{i,j}\Bigl|A_3(\pi_m(i),j) - f_3(i,j)\Bigr|. 
\end{eqnarray} 

Joint $L_1$ Fitting to Target Functions: Let $F_1(i,j)=f_1(i,j)$, $F_2(i,j)=f_2(i,j)$, $F_3(i,j)=f_3(i,j)$.  The optimization (permutation $P$ to minimize) becomes
\begin{eqnarray}
\| P A_1 P^\top - F_1 \|_1
+ \lambda_2 \| P A_2 - F_2 \|_1
+ \lambda_3 \| P A_3 - F_3 \|_1
\end{eqnarray}
where $\|\cdot\|_1$ is the entrywise $L_1$ norm. 
}

\subsection*{\textbf{Coefficient Reconstruction}}

The server-side after perturbation, depending on the surrogate family used during fitting, the server reconstructs each coefficient block through a predefined parametric mapping or through a low-rank factorization and then applies a structure-preserving post-processing step (which is privacy-free due to the post-processing immunity) to ensure feasibility and improve downstream stability.

\smallskip
\noindent\textbf{Server-side reconstruction (parametric and low-rank surrogates).}
The coordinator receives the privatized surrogate parameters $\{\widetilde{\boldsymbol{\Theta}}_m^{t}\}_{t\in\mathcal{T}}$ from each user $m$ and reconstructs each coefficient block in the permuted index space used during local fitting. The reconstruction follows the corresponding surrogate family: 
$\widetilde{\mathbf{A}}_m^{t} = f\left(\cdot;\widetilde{\boldsymbol{\Theta}}_m^{t}\right),~t\in\mathcal{T}$.

\smallskip
\noindent\textbf{Instantiation A: predefined parametric reconstruction.}
When a predefined parametric family is used, reconstruction is performed element-wise by evaluating the recovered surface:
\begin{equation}
\big[\mathcal{R}_{\mathrm{para}}(\widetilde{\boldsymbol{\Theta}})\big]_{ij}
\triangleq
f(i,j;\widetilde{\boldsymbol{\Theta}}).
\end{equation}
Concretely, typical instances include
\begin{equation}
f(x,y;\widetilde{\boldsymbol{\Theta}})=
\begin{cases}
\tilde{A}\exp\left(-\frac{(x-\tilde{c}_x)^2}{2\tilde{s}_x^2}-\frac{(y-\tilde{c}_y)^2}{2\tilde{s}_y^2}\right)+\tilde{b}, & \text{Gaussian},\\[0.9em]
\displaystyle \sum_{p,q}\tilde{\beta}_{pq}L_p(x)L_q(y)+\tilde{b}, & \text{Polynomial},\\[0.9em]
\displaystyle \sum_{k}\tilde{\alpha}_k\phi\left(\|[x,y]-\tilde{\mu}_k\|;\tilde{\sigma}\right)+\tilde{b}, & \text{RBF},
\end{cases}
\end{equation}
where $L_p(\cdot)$ denotes Legendre polynomials (chosen for numerical stability and orthogonality), and $\phi(\cdot;\tilde{\sigma})$ is the radial basis kernel. The resulting $\widehat{\mathbf{A}}_m^{t}$ is then mapped back to the original index order via Eq.~\eqref{eq:unpermute_square}.

\smallskip
\noindent\textbf{Instantiation B: low-rank reconstruction.}
When a low-rank surrogate is used, the coordinator receives perturbed rank-$k$ factors
$\widetilde{\boldsymbol{\Theta}}=(\widetilde{\mathbf{P}},\widetilde{\mathbf{Q}})$ (or equivalently $(\widetilde{\mathbf{U}},\widetilde{\boldsymbol{\Sigma}},\widetilde{\mathbf{V}})$) and forms the raw reconstruction
\begin{equation}
\widehat{\mathbf{A}} = \widetilde{\mathbf{P}}\widetilde{\mathbf{Q}}^\top
\quad
(\text{equivalently } \widehat{\mathbf{A}}=\widetilde{\mathbf{U}}\widetilde{\boldsymbol{\Sigma}}\widetilde{\mathbf{V}}^\top),
\end{equation}
followed by unpermutation using Eq.~\eqref{eq:unpermute_square}.

\smallskip
\noindent\textbf{Structure-preserving tightening.}
Because noise may violate algebraic or geometric invariants required by the LP constraints, the coordinator applies a structure-preserving tightening step. This yields more conservative (constraint-safe) coefficients and empirically reduces the violation ratio.

\begin{itemize}
  \item \textbf{Metric closure for intra-neighborhood distance blocks.}
  For $\widetilde{\mathbf{A}}_m^{\textsf{dx}}$, we compute the metric closure (Floyd--Warshall):
  \begin{equation}
  \widetilde{\mathbf{A}}_m^{\textsf{dx}}(i,j)
  \leftarrow
  \min_{k\in\mathcal{N}_m}
  \Big(\widetilde{\mathbf{A}}_m^{\textsf{dx}}(i,k)+\widetilde{\mathbf{A}}_m^{\textsf{dx}}(k,j)\Big).
  \end{equation}
  This enforces non-negativity, symmetry, zero diagonal, and all triangle inequalities by tightening noise-induced overestimates to the smallest path-consistent upper bounds.
  \item \textbf{Rectangular triangle tightening for cross blocks.}
  For each rectangular block $\widetilde{\mathbf{B}}_m\in\{\widetilde{\mathbf{A}}_m^{\textsf{dy}},\widetilde{\mathbf{A}}_m^{\textsf{cy}}\}$, we tighten entries using triangle-inequality-consistent upper bounds against the metric core:
  \begin{equation}
  \widetilde{\mathbf{B}}_m(i,j)
  \leftarrow
  \min_{k\in\mathcal{N}_m}
  \Big(\widetilde{\mathbf{A}}_m^{\textsf{dx}}(i,k)+\widetilde{\mathbf{B}}_m(k,j)\Big).
  \end{equation}
  This propagates the conservative geometry of $\widetilde{\mathbf{A}}_m^{\textsf{dx}}$ into the rectangular blocks, shrinking overestimates and reducing violations without additional privacy expenditure.
\end{itemize}
The tightened reconstructions $\{\widetilde{\mathbf{A}}_m^{t}\}_{t\in\mathcal{T}}$ are then used in the global LP exactly where the original coefficient blocks $\{\mathbf{A}_m^{t}\}_{t\in\mathcal{T}}$ would appear.

\DEL{
\subsection{Posterior Derivation with Laplace Noise}
Let $v \in \mathcal{V}$ be the true location. Around $v$, we have a set of neighbors 
\begin{equation}
\mathcal{N}(v) = \{v_1, v_2, \dots \}.
\end{equation}
For each pair $(v_i,v_j)$ in $\mathcal{N}(v)^2$, the true distance is
\begin{equation}
d_{i,j}(v).
\end{equation}
To protect location privacy, we add independent Laplace$(0,b)$ noise $a_{i,j}$ to each distance, yielding
\begin{equation}
d_{i,j}' = d_{i,j}(v) + a_{i,j}.
\end{equation}
We collect these noisy distances into the matrix
\begin{equation}
D'(v) = \{ d_{i,j}' \}_{(v_i,v_j)\in \mathcal{N}(v)^2}.
\end{equation}

\subsection*{Bayesian Formulation}

Let $\Pr(v)$ be the prior distribution over locations $v \in \mathcal{V}$. We want the posterior distribution
\begin{equation}
\Pr\bigl(v \mid D'(v)\bigr).
\end{equation}
By Bayes' theorem:
\begin{equation}
\Pr\bigl(v \mid D'(v)\bigr)
  = 
  \frac{
    \Pr\bigl(D'(v) \mid v\bigr)  \Pr(v)
  }{
    \displaystyle \sum_{v' \in \mathcal{V}}
      \Pr\bigl(D'(v) \mid v'\bigr)  \Pr(v')
  }.
\end{equation}

\subsection*{Likelihood under i.i.d.~Laplace Noise}

Since $d_{i,j}' = d_{i,j}(v) + a_{i,j}$ and each $a_{i,j}$ is drawn i.i.d.from $\mathrm{Laplace}(0,b)$, the density for a single observation $d_{i,j}'$ is:
\begin{equation}
f_{i,j}\bigl(d_{i,j}' \mid v\bigr)
  =
  \frac{1}{2b} \exp\Bigl(-\frac{\bigl|d_{i,j}' - d_{i,j}(v)\bigr|}{b}\Bigr).
\end{equation}
Assuming independence across all $(i,j)$ pairs, the joint likelihood for the entire matrix $D'(v)$ is:
\begin{equation}
\Pr\bigl(D'(v) \mid v\bigr)
  =
  \prod_{(i,j)\in \mathcal{N}(v)^2} 
    \frac{1}{2b} 
    \exp\Bigl(-\frac{\bigl|d_{i,j}' - d_{i,j}(v)\bigr|}{b}\Bigr).
\end{equation}
We can write this more compactly as:
\begin{equation}
\Pr\bigl(D'(v) \mid v\bigr)
  =
  \left(\frac{1}{2b}\right)^{\lvert \mathcal{N}(v)^2\rvert}
  \exp\Bigl(-\frac{1}{b}\sum_{(i,j)\in \mathcal{N}(v)^2} \bigl|d_{i,j}' - d_{i,j}(v)\bigr|\Bigr).
\end{equation}

\subsection*{Posterior}

Putting the likelihood and the prior together in Bayes' rule, we have:
\normalsize
\begin{eqnarray}
&& \Pr\bigl(v \mid D'(v)\bigr) \\
\nonumber &=& 
  \frac{
    \left(\frac{1}{2b}\right)^{\lvert \mathcal{N}(v)^2\rvert} 
    e^{-\frac{1}{b}\sum_{(i,j)\in \mathcal{N}(v)^2} \bigl|d_{i,j}' - d_{i,j}(v)\bigr|}\Pr(v)
  }{
    \sum_{v' \in \mathcal{V}} 
      \Bigl[\left(\frac{1}{2b}\right)^{\lvert \mathcal{N}(v')^2\rvert} 
      e^{-\frac{1}{b}\sum_{(i,j)\in \mathcal{N}(v')^2} \bigl|d_{i,j}' - d_{i,j}(v')\bigr|}\Bigr] 
      \Pr(v')
  }.
\end{eqnarray}
\normalsize
In many cases, the factor $\bigl(\tfrac{1}{2b}\bigr)^{\lvert \mathcal{N}(v)^2\rvert}$ does not depend on $v$ (if $|\mathcal{N}(v)|$ is fixed for all $v$), so it may partially cancel in the numerator and denominator. If $\vert \mathcal{N}(v)\vert$ differs across $v$, one must keep it explicit.}

\subsection{Additional Experimental Results}
\label{app:additional_experimental_results}

\begin{table*}[t]
\caption{Performance of coefficient approximation across the three datasets (4000 nodes).}
\centering
\scriptsize 
\begin{tabular}{p{2.0cm} p{1.2cm}p{1.6cm}p{1.2cm} p{1.2cm}p{1.56cm}p{1.2cm} p{1.2cm}p{1.2cm}}
\toprule
\multicolumn{1}{l}{}& \multicolumn{3}{c}{Distance matrix $\tilde{\mathbf{A}}_m^{\textsf{dx}}$} & \multicolumn{3}{c}{Distance matrix $\tilde{\mathbf{A}}^d_{\mathcal{N}_m\times \mathcal{Y}}$} & \multicolumn{2}{c}{Utility loss matrix $\tilde{\mathbf{A}}^c_{\mathcal{N}_m\times \mathcal{Y}}$}\\
 \cline{2-9}
\textbf{Methods} & PL & RelErr & ViolRatio & PL & RelErr & ViolRatio & PL & RelErr  \\
\midrule
\multicolumn{9}{c }{Rome road map}\\
\hline
2D Gaussian & 0.65±0.19 & 0.56±0.08 & 0.38±0.08 & 0.37±0.08 & 0.55±0.06 & 0.42±0.08 & 0.59±0.29 & 0.60±0.06  \\
\hline
2D Polynomial & 4.02±0.49 & 380.54±84.93 & 0.55±0.06 & 2.48±0.77 & 213.81±67.31 & 0.47±0.10 & 2.30±0.92 & 48.45±35.55 \\
\hline
RBF & 7.01±1.66 & 0.57±0.05 & 0.41±0.06 & 7.00±3.59 & 0.58±0.08 & 0.39±0.05 & 6.04±2.97 & 0.71±0.03 \\
\hline
Low-Rank SVD & 0.80±0.15 & 0.07±0.00 & 0.13±0.02 & 0.53±0.20 & 0.05±0.01 & 0.20±0.02 & 0.62±0.11 & 0.03±0.00  \\
\midrule
\multicolumn{9}{c }{London road map}
\\ 
\hline
2D Gaussian & 1.29±0.65 & 0.56±0.06 & 0.37±0.06 & 0.49±0.15 & 0.55±0.12 & 0.44±0.09 & 0.43±0.14 & 0.51±0.04  \\
\hline
2D Polynomial & 2.80±0.79x & 338.24±310.64 & 0.43±0.18 & 2.89±1.15 & 270.3±99.21 & 0.46±0.09 & 2.73±0.65 & 120.87±30.61 \\
\hline
RBF & 3.79±1.04 & 0.53±0.02 & 0.43±0.02 & 3.49±0.80 & 0.52±0.02 & 0.35±0.04 & 3.40±0.84 & 0.62±0.01 \\
\hline
Low-Rank SVD & 0.84±0.08 & 0.07±0.00 & 0.14±0.01 & 0.58±0.04 & 0.04±0.00 & 0.21±0.01 & 0.96±0.39 & 0.03±0.00  \\
\midrule
\multicolumn{9}{c }{New York City road map}
\\
\hline
2D Gaussian & 0.48±0.14 & 0.66±0.11 & 0.40±0.09 & 0.55±0.15 & 0.67±0.03 & 0.46±0.06 & 0.39±0.12 & 0.54±0.03  \\
\hline
2D Polynomial & 2.22±0.69 & 268.42±160.48 & 0.43±0.21 & 2.20±1.33 & 271.57±158.07 & 0.41±0.18 & 2.28±0.80 & 63.64±58.32 \\
\hline
RBF & 2.52±0.77 & 0.65±0.05 & 0.44±0.03 & 3.08±1.00 & 0.64±0.08 & 0.45±0.01 & 2.93±0.89 & 0.66±0.01 \\
\hline
Low-Rank SVD & 0.95±0.18 & 0.08±0.00 & 0.08±0.00 & 0.65±0.13 & 0.07±0.01 & 0.14±0.01 & 0.86±0.25 & 0.03±0.00  \\
\bottomrule
\end{tabular}
\label{tab:exp:matrix_upload_4000_nodes}
\end{table*}

\begin{table*}[t]
\caption{Performance of coefficient approximation across the three datasets (6000 nodes).}
\centering
\scriptsize 
\begin{tabular}{p{2.0cm} p{1.2cm}p{1.6cm}p{1.2cm} p{1.2cm}p{1.56cm}p{1.2cm} p{1.2cm}p{1.2cm}}
\toprule
\multicolumn{1}{l}{}& \multicolumn{3}{c}{Distance matrix $\tilde{\mathbf{A}}_m^{\textsf{dx}}$} & \multicolumn{3}{c}{Distance matrix $\tilde{\mathbf{A}}^d_{\mathcal{N}_m\times \mathcal{Y}}$} & \multicolumn{2}{c}{Utility loss matrix $\tilde{\mathbf{A}}^c_{\mathcal{N}_m\times \mathcal{Y}}$}\\
 \cline{2-9}
\textbf{Methods} & PL & RelErr & ViolRatio & PL & RelErr & ViolRatio & PL & RelErr  \\
\midrule

\multicolumn{9}{c }{Rome road map}\\
\hline
2D Gaussian & 1.17±1.33 & 0.60±0.12 & 0.39±0.08 & 0.83±0.62 & 0.50±0.07 & 0.47±0.07 & 0.85±0.76 & 0.56±0.04  \\
\hline
2D Polynomial & 3.43±1.05 & 350.98±217.18 & 0.50±0.14 & 2.58±1.37 & 207.81±143.53 & 0.42±0.04 & 3.02±1.87 & 6.34±10.45 \\
\hline
RBF & 3.62±1.45 & 0.58±0.01 & 0.41±0.03 & 5.55±2.90 & 0.53±0.03 & 0.36±0.05 & 3.54±1.19 & 0.69±0.02 \\
\hline
Low-Rank SVD & 0.79±0.16 & 0.07±0.00 & 0.13±0.01 & 0.65±0.32 & 0.04±0.00 & 0.19±0.02 & 0.78±0.23 & 0.04±0.00  \\
\midrule
\multicolumn{9}{c }{London road map}
\\
\hline
2D Gaussian & 0.55±0.12 & 0.57±0.04 & 0.39±0.05 & 0.72±0.22 & 1.17±1.21 & 0.47±0.06 & 0.68±0.47 & 0.45±0.04  \\
\hline
2D Polynomial & 2.68±1.25 & 393.66±152.71 & 0.51±0.16 & 2.17±0.71 & 463.08±279.57 & 0.46±0.11 & 2.97±1.67 & 85.62±45.48 \\
\hline
RBF & 4.45±2.92 & 0.53±0.02 & 0.41±0.05 & 4.45±1.72 & 0.52±0.01 & 0.39±0.03 & 3.69±1.87 & 0.62±0.01 \\
\hline
Low-Rank SVD & 0.87±0.06 & 0.07±0.00 & 0.14±0.01 & 0.56±0.06 & 0.04±0.00 & 0.22±0.02 & 0.84±0.46 & 0.03±0.00  \\
\midrule
\multicolumn{9}{c }{New York City road map}
\\
\hline
2D Gaussian & 0.54±0.12 & 0.68±0.06 & 0.44±0.05 & 0.72±0.33 & 1.74±2.24 & 0.50±0.03 & 0.65±0.36 & 0.51±0.05  \\
\hline
2D Polynomial & 2.35±0.59 & 483.97±284.02 & 0.50±0.15 & 2.56±1.40 & 352.18±168.12 & 0.43±0.09 & 2.39±0.93 & 12.02±15.26 \\
\hline
RBF & 2.44±0.67 & 0.66±0.03 & 0.47±0.05 & 3.30±1.71 & 0.65±0.04 & 0.48±0.08 & 3.00±1.37 & 0.67±0.03 \\
\hline
Low-Rank SVD & 0.96±0.23 & 0.08±0.01 & 0.09±0.02 & 0.65±0.16 & 0.06±0.01 & 0.15±0.01 & 0.91±0.25 & 0.03±0.00  \\
\bottomrule
\end{tabular}
\vspace{-0.02in}
\label{tab:exp:matrix_upload_6000_nodes}
\end{table*}

\subsection*{\textbf{Open Science}}
We provide a complete MATLAB implementation of the proposed DistA-LP framework to enable full reproducibility of all experimental results reported in the paper. The artifact is available at \url{https://anonymous.4open.science/r/DistA-LP-D2F0}. It includes a README that documents the code structure, key components, and all relevant configuration parameters.

\subsection*{\textbf{Violation Ratio}}

Table~\ref{tab:violation} summarizes the empirical mDP violation ratios for all perturbation mechanisms across datasets, privacy budgets, and domain sizes. Predefined-noise baselines (EM and Laplace) and the hybrid method RMP generally exhibit zero or near-zero violations. In contrast, centralized LP consistently incurs the largest violation ratios among all methods, as its coarse-grained distance approximation can substantially mis-estimate inter-record distances and thereby relax the enforced mDP constraints. DistA-LP, on the other hand, achieves empirically negligible violations, highlighting that privacy-aware coefficient approximation combined with probabilistic auditing can maintain near-zero mDP violations even in regimes where centralized LP accumulates substantial violations across scales and budgets. \\
{\rev
\noindent\textit{Conservative distance repair.}
Let \(d\) denote the true metric and let \(\hat d\) denote the reconstructed distance obtained from the perturbed surrogate parameters. For a repair parameter \(\gamma \geq 0\), define
\begin{equation}
\bar d_\gamma(x_i,x_j)=\max\{\hat d(x_i,x_j)-\gamma,0\}.   
\end{equation}
Let
\begin{equation*}
\Delta^+ = \max_{i,j}(\hat d(x_i,x_j)-d(x_i,x_j))_+ .
\end{equation*}
If \(\gamma \geq \Delta^+\), then
\begin{equation*}
\bar d_\gamma(x_i,x_j)\leq d(x_i,x_j), \quad \forall i,j.
\end{equation*}
Therefore, any perturbation matrix satisfying the mDP constraints under $\bar d_\gamma$ also satisfies the corresponding mDP constraints under the true metric \(d\).
\begin{proof}
For any pair \((i,j)\), if \(\hat d(x_i,x_j)\leq d(x_i,x_j)\), then
\begin{equation}
\bar d_\gamma(x_i,x_j)\leq \hat d(x_i,x_j)\leq d(x_i,x_j).
\end{equation}
If \(\hat d(x_i,x_j)>d(x_i,x_j)\), then by definition of \(\Delta^+\),
\begin{equation*}
\hat d(x_i,x_j)-d(x_i,x_j)\leq \Delta^+ \leq \gamma.
\end{equation*}
Thus,
\begin{equation}
\hat d(x_i,x_j)-\gamma \leq d(x_i,x_j).
\end{equation}
Since \(\bar d_\gamma\) is obtained by clipping \(\hat d-\gamma\) at zero, the inequality still holds. Hence \(\bar d_\gamma\leq d\) elementwise.
The mDP constraint has the form
\begin{equation*}
z_{i,k}\leq e^{\epsilon d(x_i,x_j)}z_{j,k}.
\end{equation*}
Because \(\bar d_\gamma(x_i,x_j)\leq d(x_i,x_j)\), we have
\begin{equation}
e^{\epsilon \bar d_\gamma(x_i,x_j)}
\leq
e^{\epsilon d(x_i,x_j)}.
\end{equation}
Therefore, constraints enforced using \(\bar d_\gamma\) are at least as strict as those under the true distance \(d\), giving a conservative mDP guarantee.
\end{proof}}

\subsection*{\textbf{Why local permutation helps.}}
Let $\Pi_m$ be a user-private random permutation over the neighborhood indices, applied to coefficient blocks as in~(15).
From the server's view, the uploaded surrogate parameters are drawn from a permutation-mixture distribution
$\Pr[\tilde{\Theta}_m \in S \mid A_m] = \mathbb{E}_{\Pi_m}\left[\Pr[\tilde{\Theta}_m \in S \mid \Pi_m(A_m)]\right]$.
In particular, any two coefficient tensors that differ only by a relabeling of neighborhood indices (i.e., $A_m'=\sigma(A_m)$
for some permutation $\sigma$) induce the same distribution over uploads, since $\Pi_m \circ \sigma$ is identically distributed to $\Pi_m$.
Thus, permutation primarily obscures the alignment between indices and concrete locations, weakening linkage-based inference.

\begin{table*}[t]
\caption{Performance of coefficient approximation across the three datasets (8000 nodes).}
\centering
\scriptsize 
\begin{tabular}{p{2.0cm} p{1.2cm}p{1.6cm}p{1.2cm} p{1.2cm}p{1.56cm}p{1.2cm} p{1.2cm}p{1.2cm}}
\toprule
\multicolumn{1}{l}{}& \multicolumn{3}{c}{Distance matrix $\tilde{\mathbf{A}}_m^{\textsf{dx}}$} & \multicolumn{3}{c}{Distance matrix $\tilde{\mathbf{A}}^d_{\mathcal{N}_m\times \mathcal{Y}}$} & \multicolumn{2}{c}{Utility loss matrix $\tilde{\mathbf{A}}^c_{\mathcal{N}_m\times \mathcal{Y}}$}\\
 \cline{2-9}
\textbf{Methods} & PL & RelErr & ViolRatio & PL & RelErr & ViolRatio & PL & RelErr  \\
\midrule
\multicolumn{9}{c }{Rome road map}\\
\hline
Low-Rank SVD & 0.60±0.06 & 0.07±0.00 & 0.14±0.01 & 0.58±0.15 & 0.04±0.01 & 0.21±0.01 & 1.09±0.35 & 0.04±0.00  \\
\midrule
\multicolumn{9}{c }{London road map}
\\
\hline
Low-Rank SVD & 0.93±0.07 & 0.07±0.00 & 0.14±0.01 & 0.68±0.11 & 0.04±0.00 & 0.23±0.01 & 0.79±0.26 & 0.04±0.00  \\
\midrule
\multicolumn{9}{c }{New York City road map}
\\
\hline
Low-Rank SVD & 0.86±0.09 & 0.08±0.00 & 0.08±0.01 & 0.67±0.16 & 0.06±0.00 & 0.14±0.03 & 0.73±0.10 & 0.03±0.00  \\
\bottomrule
\end{tabular}
\label{tab:exp:matrix_upload_8000_nodes}
\end{table*}

\begin{table*}[t]
\caption{Performance of coefficient approximation across the three datasets (10000 nodes).}
\centering
\scriptsize 
\begin{tabular}{p{2.0cm} p{1.2cm}p{1.6cm}p{1.2cm} p{1.2cm}p{1.56cm}p{1.2cm} p{1.2cm}p{1.2cm}}
\toprule
\multicolumn{1}{l}{}& \multicolumn{3}{c}{Distance matrix $\tilde{\mathbf{A}}_m^{\textsf{dx}}$} & \multicolumn{3}{c}{Distance matrix $\tilde{\mathbf{A}}^d_{\mathcal{N}_m\times \mathcal{Y}}$} & \multicolumn{2}{c}{Utility loss matrix $\tilde{\mathbf{A}}^c_{\mathcal{N}_m\times \mathcal{Y}}$}\\
 \cline{2-9}
\textbf{Methods} & PL & RelErr & ViolRatio & PL & RelErr & ViolRatio & PL & RelErr  \\
\midrule
\multicolumn{9}{c }{Rome road map}\\
\hline
Low-Rank SVD & 0.59±0.09 & 0.07±0.00 & 0.14±0.01 & 0.62±0.29 & 0.04±0.00 & 0.21±0.01 & 1.09±0.21 & 0.04±0.00  \\
\midrule
\multicolumn{9}{c }{London road map} \\
\hline
Low-Rank SVD & 0.89±0.10 & 0.07±0.00 & 0.13±0.01 & 0.56±0.04 & 0.04±0.00 & 0.20±0.01 & 0.80±0.27 & 0.04±0.00  \\
\midrule
\multicolumn{9}{c }{New York City road map} \\
\hline
Low-Rank SVD & 0.94±0.12 & 0.08±0.01 & 0.07±0.01 & 0.54±0.12 & 0.06±0.01 & 0.15±0.03 & 0.81±0.27 & 0.03±0.00  \\
\bottomrule
\end{tabular}
\label{tab:exp:matrix_upload_10000_nodes}
\end{table*}

\begin{table*}[t]
\caption{{\cam Memory cost.}}
\centering
\scriptsize 
\begin{tabular}{p{1.80cm} |  p{0.80cm} p{1.10cm} p{1.10cm}  p{1.10cm} | p{1.13cm} p{1.13cm} p{1.13cm} | p{1.22cm} p{1.22cm} p{1.22cm}} 
\toprule
\multicolumn{2}{c }{Number of records} & \multicolumn{3}{c }{2,000} & \multicolumn{3}{c }{4,000}& \multicolumn{3}{c }{6,000}\\
\midrule
\multicolumn{2}{c }{Privacy budget (km$^{-1}$)} & $\epsilon = 4.0$& $\epsilon = 7.0$& $\epsilon = 10.0$ & $\epsilon = 4.0$& $\epsilon = 7.0$& $\epsilon = 10.0$& $\epsilon = 4.0$& $\epsilon = 7.0$& $\epsilon = 10.0$\\
\midrule
\multicolumn{11}{c }{Rome road map}\\
\hline
Pre-defined& EM & 0.64 MB &  0.64 MB &  0.64 MB & 1.28 MB & 1.28 MB & 1.28 MB & 1.92 MB & 1.92 MB & 1.92 MB \\
\cline{2-2}
Noise Distribution & Laplace  & 0.96 MB & 0.96 MB & 0.96 MB & 1.92 MB & 1.92 MB & 1.92 MB & 2.88 MB & 2.88 MB & 2.88 MB \\
\cline{1-2}
Hybrid  & RMP & 0.96 MB & 0.96 MB & 0.96 MB & 1.92 MB & 1.92 MB & 1.92 MB & 2.88 MB & 2.88 MB & 2.88 MB \\
\cline{2-2}
Method &  COPT  & ------------- & ------------- & ------------- & ------------- & ------------- & ------------- & ------------- & ------------- & ------------- \\
\cline{1-2}
Optimization &  LP  & $>$25.6 TB & $>$25.6 TB & $>$25.6 TB & $>$204.7 TB & $>$204.7 TB & $>$204.7 TB & $>$691.1 TB & $>$691.1 TB & $>$691.1 TB \\
\cline{2-2}
Based & LP-A & 392 MB & 392 MB & 392 MB & 392 MB & 392 MB & 392 MB & 392 MB & 392 MB & 392 MB \\
\cline{2-2} 
Methods & PAnDA & 232.39 MB & 232.39 MB & 232.39 MB & 933.56 MB & 933.56 MB & 933.56 MB & 2313.32 MB & 2313.32 MB & 2313.32 MB \\
\hline
\rowcolor{lightgray!30}
\multicolumn{2}{c}{\textsc{DistA-LP}$^\dagger$} & 1.0375 MB & 7.7995 MB & 14.58 MB & 1.4009 MB & 6.6318 MB & 19.645 MB & 1.6471 MB & 8.4639 MB & 25.628 MB \\

\midrule
\multicolumn{11}{c }{London road map}\\
 \hline
Pre-defined& EM & 0.64 MB &  0.64 MB &  0.64 MB & 1.28 MB & 1.28 MB & 1.28 MB & 1.92 MB & 1.92 MB & 1.92 MB \\
\cline{2-2}
Noise Distribution & Laplace  & 0.96 MB & 0.96 MB & 0.96 MB & 1.92 MB & 1.92 MB & 1.92 MB & 2.88 MB & 2.88 MB & 2.88 MB  \\
\cline{1-2}
Hybrid  & RMP & 0.96 MB & 0.96 MB & 0.96 MB & 1.92 MB & 1.92 MB & 1.92 MB & 2.88 MB & 2.88 MB & 2.88 MB  \\
\cline{2-2}
Method &  COPT  & ------------- & ------------- & ------------- & ------------- & ------------- & ------------- & ------------- & ------------- & ------------- \\
\cline{1-2}
Optimization &  LP  & $>$25.6 TB & $>$25.6 TB & $>$25.6 TB & $>$204.7 TB & $>$204.7 TB & $>$204.7 TB & $>$691.1 TB & $>$691.1 TB & $>$691.1 TB \\
\cline{2-2}
Based & LP-A & 522.72 MB & 522.72 MB & 522.72 MB & 522.72 MB & 522.72 MB & 522.72 MB & 522.72 MB & 522.72 MB & 522.72 MB \\
\cline{2-2} 
 Methods & PAnDA & 401.42 MB & 401.42 MB & 401.42 MB & 923.15 MB & 923.15 MB & 923.15 MB & 2309.97 MB & 2309.97 MB & 2309.97 MB \\
\hline
\rowcolor{lightgray!30}
\multicolumn{2}{c}{\textsc{DistA-LP}$^\dagger$} & 1.7227 MB & 7.8758 MB & 22.355 MB & 2.3711 MB & 8.7055 MB & 22.77 MB & 2.7703 MB & 11.913 MB & 23.616 MB \\

\midrule
\multicolumn{11}{c }{New York City road map}\\
\hline
Pre-defined& EM & 0.64 MB &  0.64 MB &  0.64 MB & 1.28 MB & 1.28 MB & 1.28 MB & 1.92 MB & 1.92 MB & 1.92 MB \\
\cline{2-2}
Noise Distribution & Laplace  & 0.96 MB & 0.96 MB & 0.96 MB & 1.92 MB & 1.92 MB & 1.92 MB & 2.88 MB & 2.88 MB & 2.88 MB  \\
\cline{1-2}
Hybrid  & RMP & 0.96 MB & 0.96 MB & 0.96 MB & 1.92 MB & 1.92 MB & 1.92 MB & 2.88 MB & 2.88 MB & 2.88 MB \\
\cline{2-2}
Method &  COPT  & ------------- & ------------- & ------------- & ------------- & ------------- & ------------- & ------------- & ------------- & ------------- \\
\cline{1-2}
Optimization &  LP  &$>$25.6 TB & $>$25.6 TB & $>$25.6 TB & $>$204.7 TB & $>$204.7 TB & $>$204.7 TB & $>$691.1 TB & $>$691.1 TB & $>$691.1 TB \\
\cline{2-2}
Based & LP-A & 184.95 MB & 184.95 MB & 184.95 MB & 184.95 MB & 184.95 MB & 184.95 MB & 184.95 MB & 184.95 MB & 184.95 MB \\
\cline{2-2} 
Methods & PAnDA & 229.57 MB & 229.57 MB & 229.57 MB & 922.61 MB & 922.61 MB & 922.61 MB & 2312.96 MB & 2312.96 MB & 2312.96 MB\\
\hline
\rowcolor{lightgray!30}
\multicolumn{2}{c}{\textsc{DistA-LP}$^\dagger$} & 1.5049 MB & 4.0234 MB & 16.244 MB & 1.6511 MB & 5.335 MB & 20.137 MB & 6.4273 MB & 11.796 MB & 23.774 MB \\
\bottomrule
\end{tabular}
\label{tab:memory_cost}
\end{table*}

\begin{table*}[t]
\caption{{\cam Communication cost.}}
\centering
\scriptsize 
\begin{tabular}{
p{1.6cm}
>{\centering\arraybackslash}p{1.7cm}
>{\centering\arraybackslash}p{2.3cm}
>{\centering\arraybackslash}p{2.3cm}
>{\centering\arraybackslash}p{1.7cm}
>{\centering\arraybackslash}p{1.7cm}
}
\toprule
\textbf{Surrogate} &  
Upload to Server (Bytes) & 
Download from Server (Bytes) & 
Upload to Downstream (Bytes) & 
Total End-to-End (Bytes) &  
Total End-to-End (kb) \\
\midrule
Raw Matrices & 9600 & 3200 & 8 & 12808 & 12.5078 \\
\hline
2D Gaussian & 304 & 3200 & 8 & 3512 & 3.4297 \\
\hline
2D Polynomial & 400 & 3200 & 8 & 3608 & 3.5234 \\
\hline
RBF & 760 & 3200 & 8 & 3968 & 3.875 \\
\hline
Low-Rank SVD & 1640 & 3200 & 8 & 4848 & 4.7344 \\
\bottomrule
\end{tabular}
\label{tab:communication_cost}
\end{table*}

\begin{figure*}[t]
\centering
\begin{minipage}{0.95\textwidth}
\centering
\includegraphics[width=0.29\textwidth]{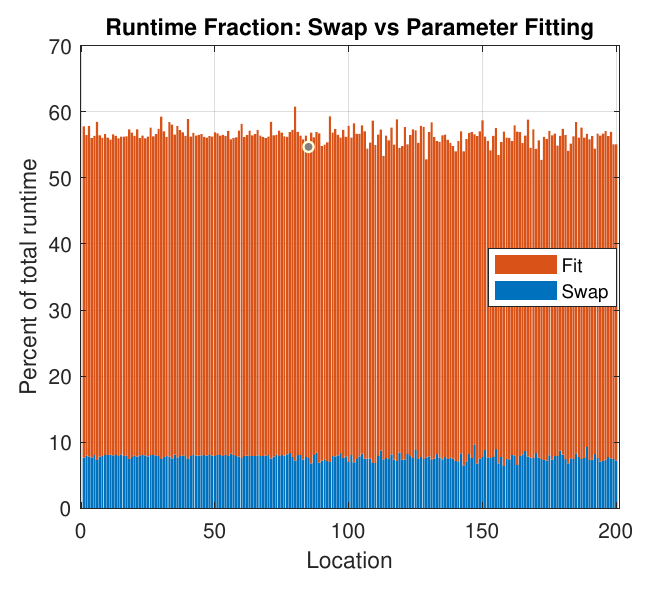}
\hspace{0.0in}
  \subfigure[Gaussian surrogate]{
\includegraphics[width=0.30\textwidth]{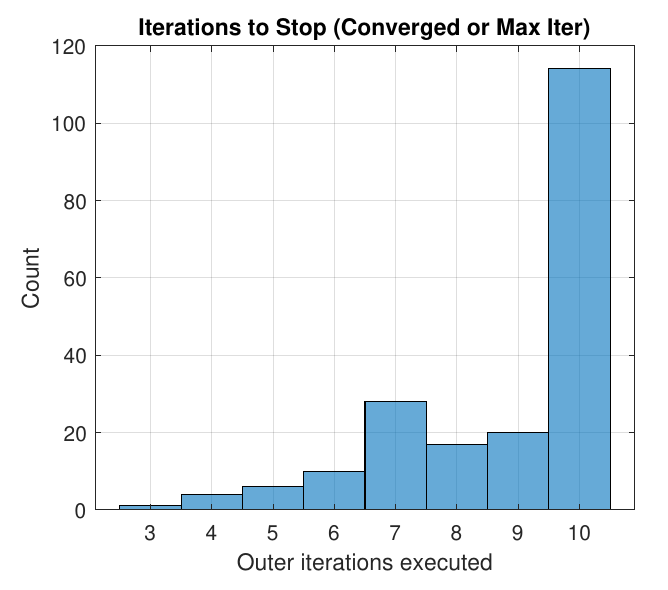}}
\hspace{0.0in}
\includegraphics[width=0.30\textwidth]{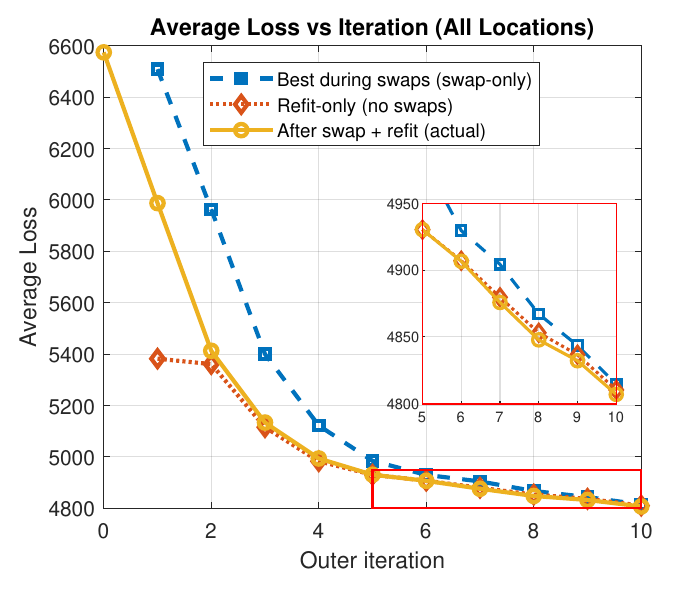}
\end{minipage}

\begin{minipage}{0.95\textwidth}
\centering
\includegraphics[width=0.30\textwidth]{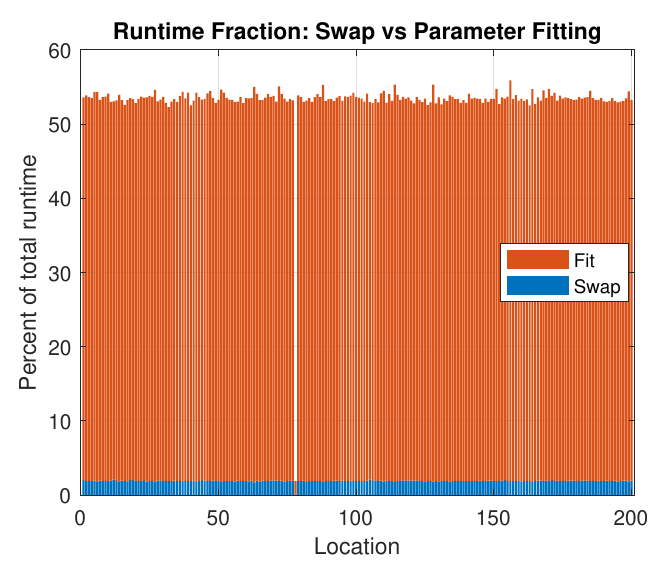}
\hspace{0.0in}
  \subfigure[Polynomial surrogate]{
\includegraphics[width=0.30\textwidth]{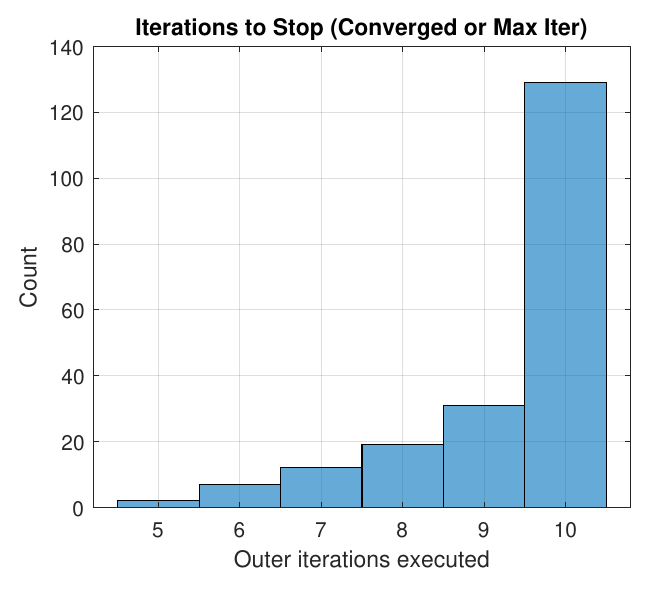}}
\hspace{0.0in}
\includegraphics[width=0.30\textwidth]{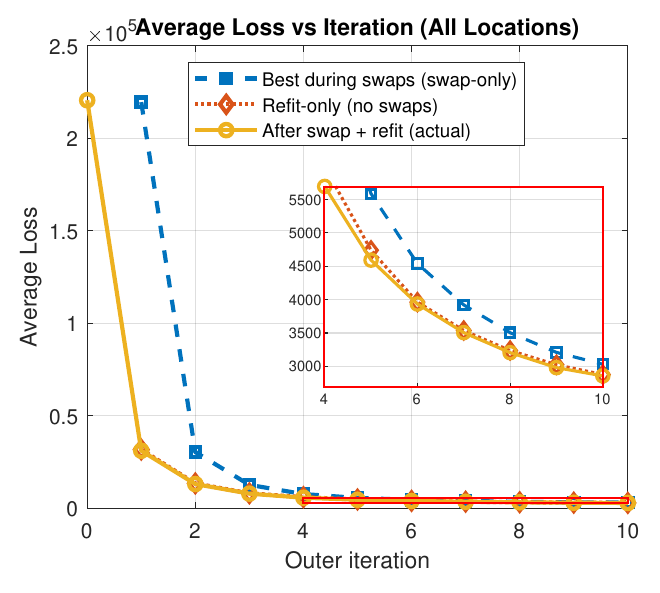}
\end{minipage}

\begin{minipage}{0.95\textwidth}
\centering
\includegraphics[width=0.30\textwidth]{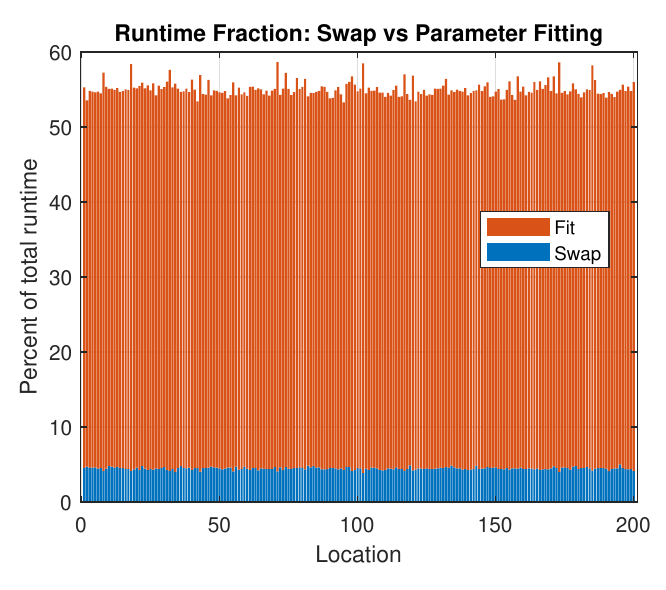}
\hspace{0.0in}
  \subfigure[RBF surrogate]{
\includegraphics[width=0.29\textwidth]{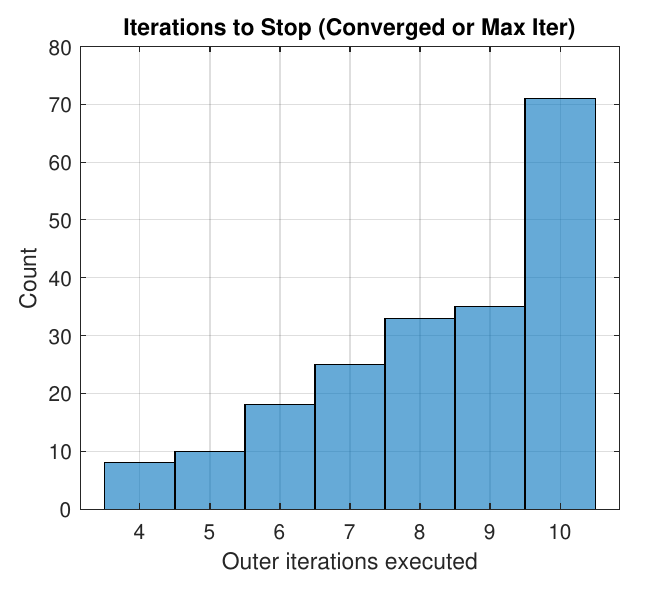}}
\hspace{0.0in}
\includegraphics[width=0.30\textwidth]{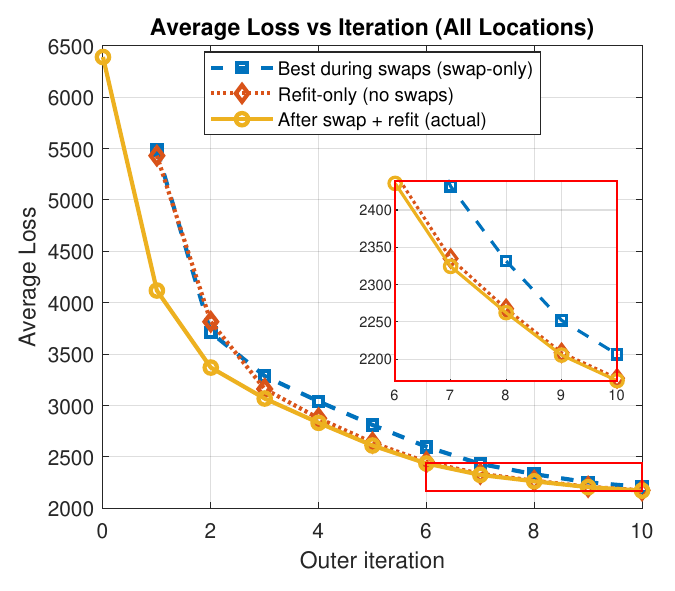}
\end{minipage}
\caption{Surrogates approximation performance with a maximum of 10 optimization iterations, showing the runtime split between parameter fitting and swap search; convergence/stopping iterations; and the decrease in average loss under swap-only, refit-only, and combined swap–refit updates. }
\label{fig:surrogate_iteration_10}
\vspace{-0.00in}
\end{figure*}

\begin{figure*}[t]
\centering
\begin{minipage}{0.95\textwidth}
\centering
\includegraphics[width=0.30\textwidth]{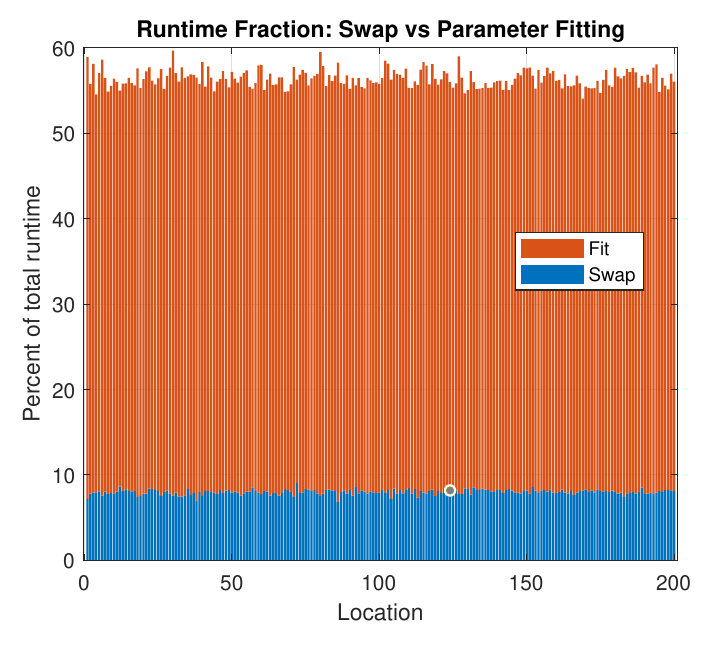}
\hspace{0.00in}
  \subfigure[Gaussian surrogate]{
\includegraphics[width=0.30\textwidth]{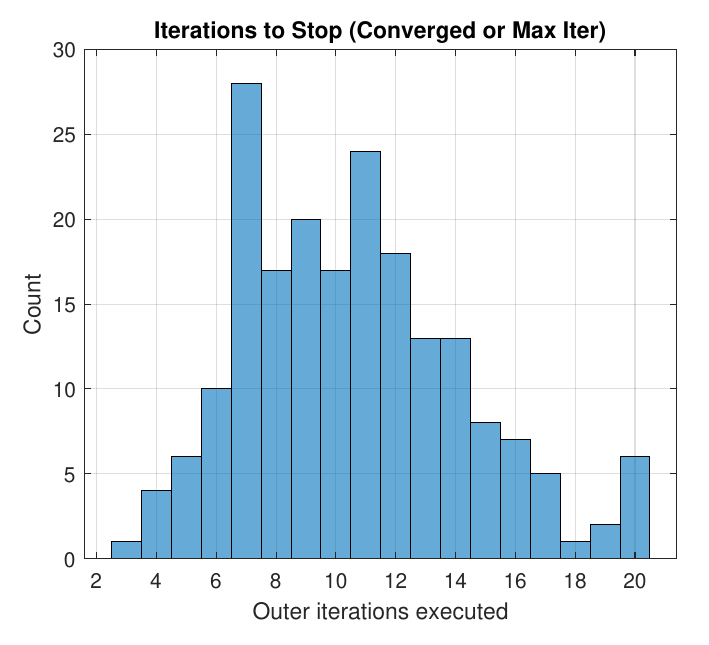}}
\hspace{0.00in}
\includegraphics[width=0.31\textwidth]{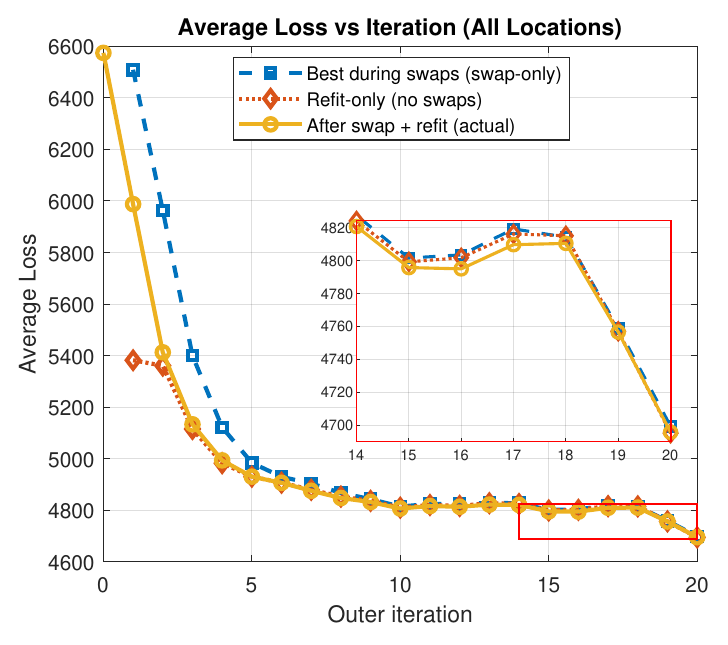}
\end{minipage}
\end{figure*}

\begin{figure*}[t]
\centering
\begin{minipage}{0.95\textwidth}
\centering
\includegraphics[width=0.31\textwidth]{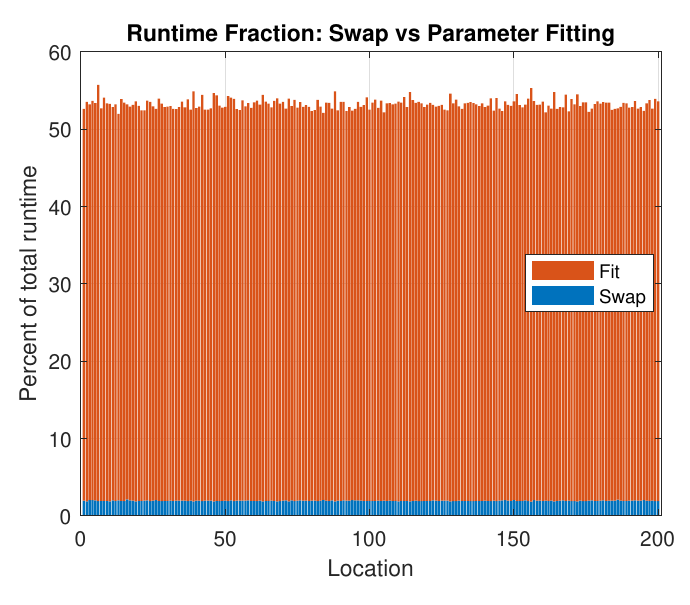}
\hspace{0.00in}
  \subfigure[Polynomial surrogate]{
\includegraphics[width=0.30\textwidth]{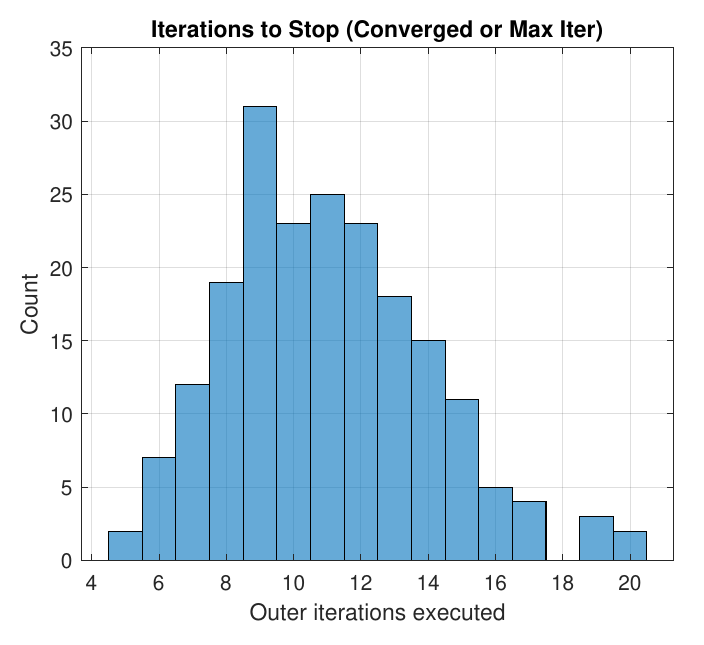}}
\hspace{0.00in}
\includegraphics[width=0.30\textwidth]{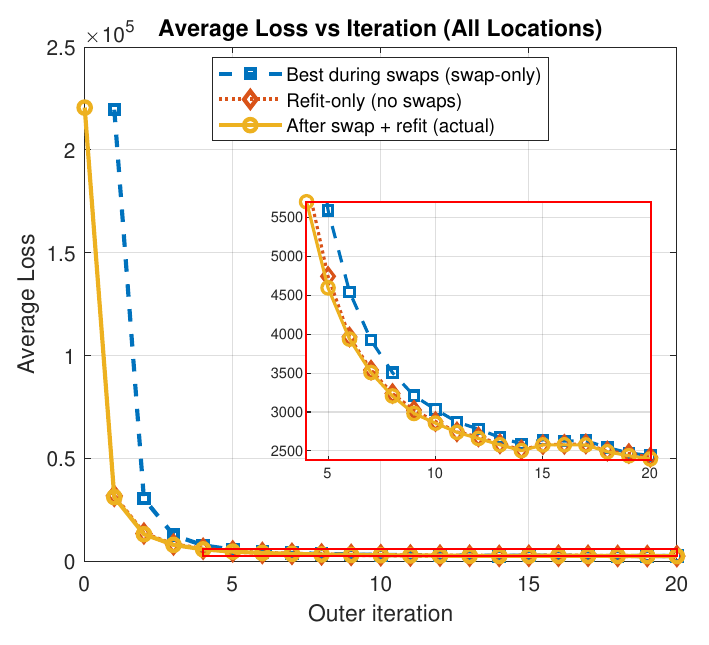}
\end{minipage}

\begin{minipage}{0.95\textwidth}
\centering
\includegraphics[width=0.32\textwidth]{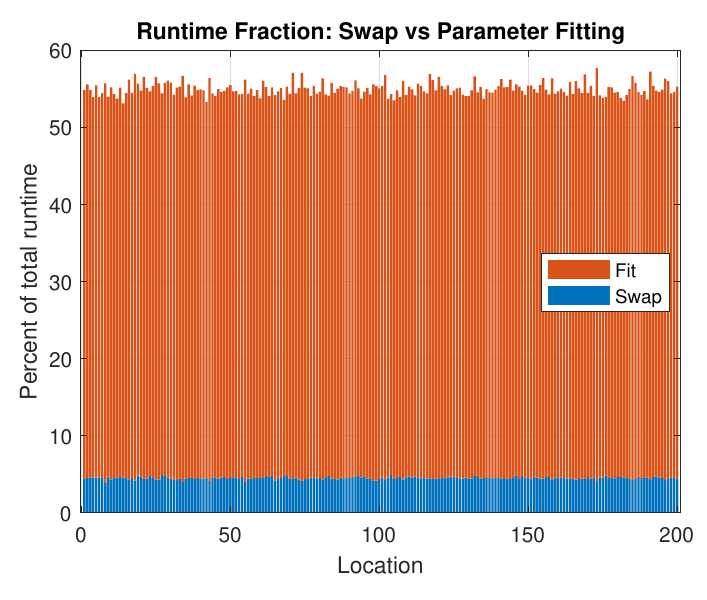}
\hspace{0.00in}
  \subfigure[RBF surrogate]{
\includegraphics[width=0.30\textwidth]{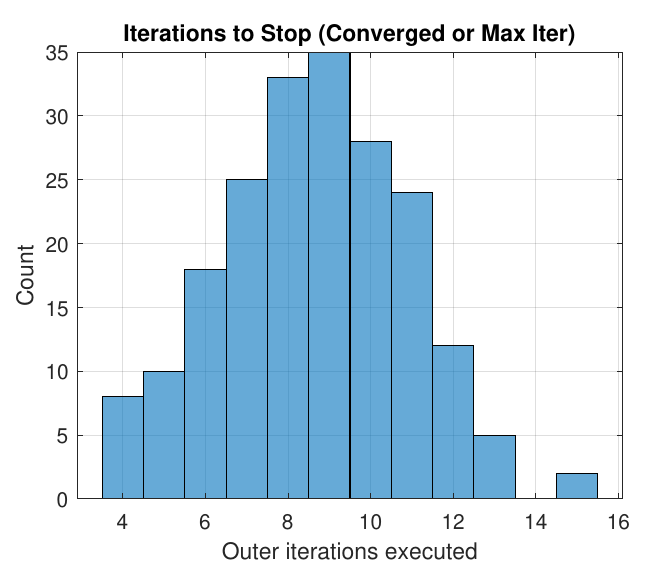}}
\hspace{0.00in}
\includegraphics[width=0.30\textwidth]{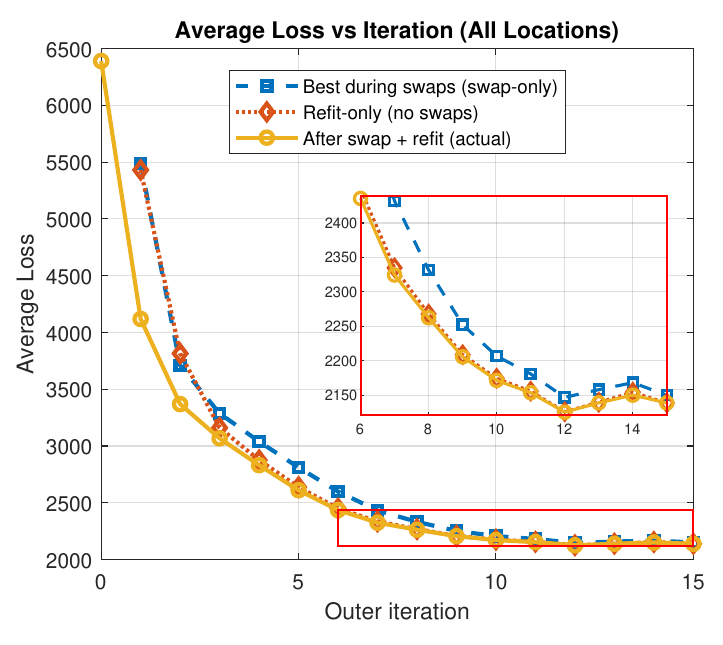}
\end{minipage}
\caption {Surrogate approximation performance with a maximum of 20 optimization iterations, showing the runtime fraction between the fitting and swap search; convergence/stopping-iteration distribution; and the average loss decay across iterations for swap-only, refit-only, and combined swap–refit updates.}
\label{fig:surrogate_iteration_20}
\vspace{-0.00in}
\end{figure*}

\begin{figure*}[t]
\centering
\begin{minipage}{0.9\textwidth}
\centering
  \subfigure{
\includegraphics[width=0.47\textwidth]{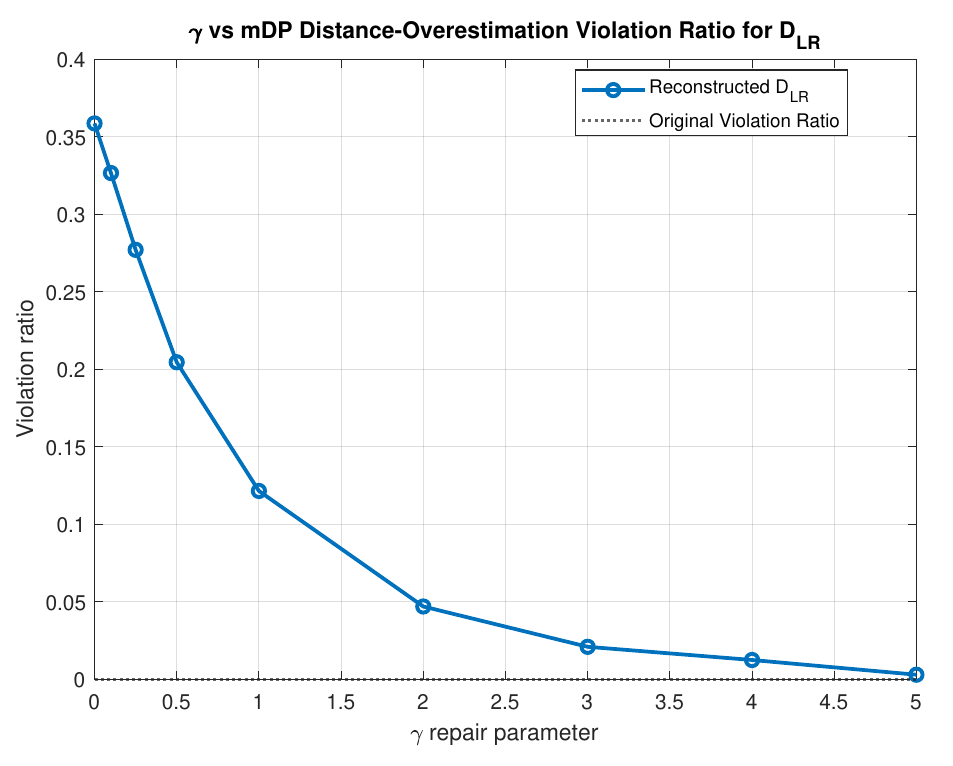}}
\hspace{0.02in}
 \subfigure{
\includegraphics[width=0.48\textwidth]{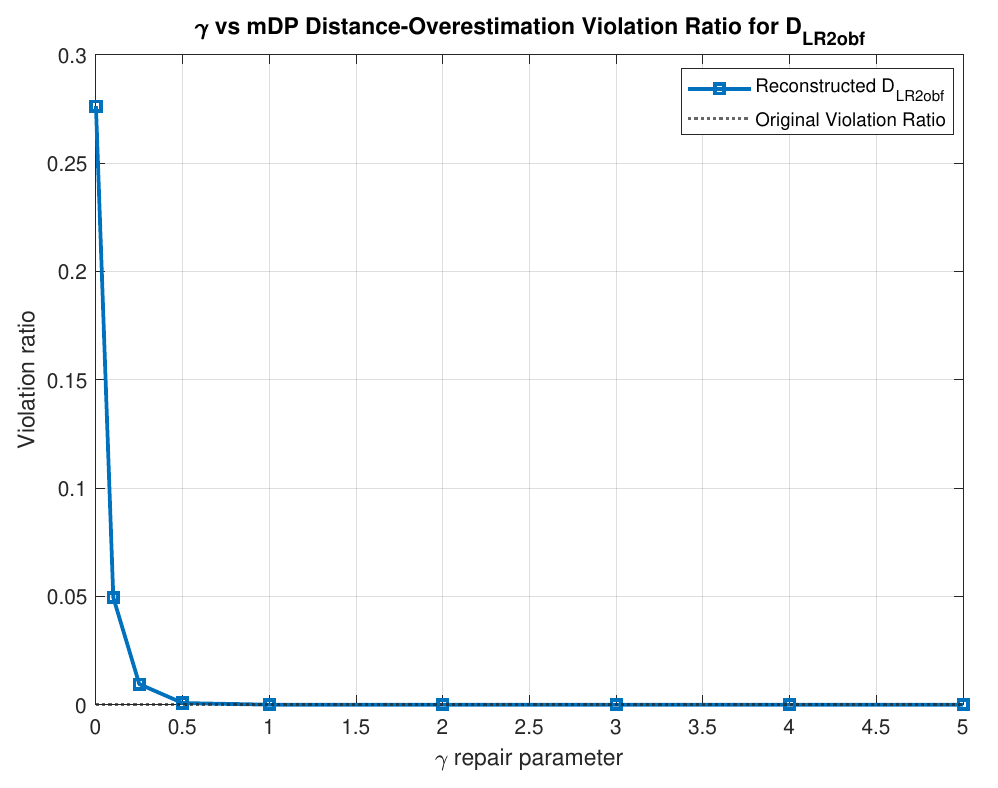}}
\caption{{\cam$\gamma$ vs mDP distance-overestimation violation ratio for both intra-neighborhood distance block and cross-domain distance block. }}
\label{fig:gamma_vs_mDP_violation}
\end{minipage}
\vspace{-0.00in}
\end{figure*}

\clearpage
\makeatletter
\global\@colroom=\@colht
\global\vsize=\@colht
\makeatother
\raggedbottom

\appendices
\section{Artifact Appendix}

\subsection{Description \& Requirements}

\subsubsection{How to access}
The artifact is available at \url{https://zenodo.org/records/21782566}. It contains the MATLAB implementation, experiment inputs, result directories, and documentation needed to evaluate DISTA-LP and its baselines. The primary entry point is
\texttt{Bendersdecomposition/main.m}; \texttt{main\_surrogate.m} evaluates coefficient approximation separately, and baseline entry points are under \texttt{Bendersdecomposition/baseline/}.

\subsubsection{Hardware dependencies}
The artifact was evaluated on a workstation with 32~GB RAM. We recommend a multicore CPU, at least 32~GB RAM, and 5--10~GB free storage. A GPU is not required. Begin with a 2,000-node experiment; centralized LP, COPT and larger PAnDA settings can require substantially more memory and time.

\subsubsection{Software dependencies}
The main framework requires MATLAB R2024b or later and Optimization Toolbox. Statistics and Machine Learning Toolbox is required by statistical and sampling functions. Parallel Computing Toolbox is recommended for parallel subproblems. Baselines were checked with MATLAB R2025b; baseline functions using \texttt{heaviside} require Symbolic Math Toolbox. The code supports Windows, macOS, and Linux when MATLAB and the required toolboxes are available.

\subsubsection{Benchmarks}
The supplied experiments use road-network data for Rome, London, and New York City (NYC), with 2,000, 4,000, or 6,000 selected records, privacy budgets $\epsilon\in\{4,7,10\}\,\mathrm{km}^{-1}$, ten users, and dataset-dependent repetitions:

\begin{center}
\small 
\begin{tabular}{lr}
\toprule
Setting & Repetitions \\
\midrule
Rome, 2,000 records & 5 \\
Rome, 4,000/6,000 records & 4 \\
London, all sizes & 6 \\
NYC, all sizes & 4 \\
\bottomrule
\end{tabular}
\end{center}

The evaluated baselines are Exponential Mechanism (EM), Bayesian remapping (RMP), coarse-grid LP approximation (LP-A), Planar Laplace, PAnDA, COPT and the original centralized metric-DP LP. The full centralized LP and COPT are infeasible at the evaluated domain sizes and are reported as unavailable.

\subsection{Artifact Installation \& Configuration}

Download the zip folder and enter the framework directory (Bendersdecomposition folder):
Start MATLAB in this directory. The main runner adds the required class and function directories to the MATLAB path.

Before a main-framework run, select a consistent dataset, target-region block, node count, uploaded-budget matrix, budget assigned, and obfuscated-location fraction.

\paragraph{Dataset} In the Bendersdecomposition folder, \texttt{func/read\_files/readCityMapInfo.m}, leave exactly one of the dataset (Rome, London, or NYC) loading blocks active depending on the dataset of consideration. From lines 14-17 for the Rome dataset, lines 20-23 for the NYC dataset, and lines 26-29 for the London dataset.

\paragraph{Target regions} In \texttt{main.m}, activate the \textit{Regions} block matching the city and record count. The region count $R=\textit{numel(Regions)}$ also determines the number of repetitions and columns in the uploaded-budget matrix. From lines 124-163, based on the dataset and node count.

\paragraph{Node count (nr\_loc\_selected) and uploaded budget} Set, for example,
\textit{env\_parameters.nr\_loc\_selected = 2000; \, budget\_assigned = budget\_assigned\_rome\_svd.}

The selected \textit{budget\_assigned} must match the city, node count, and surrogate. The main runner uses
\textit{
epsilon\_value = [10, 7, 4];
}
so matrix rows 1--3 correspond to $\epsilon=10,7,4$, respectively. With \textit{approx\_matrix=1}, the budget remaining for optimization is
\textit{
env\_parameters.EPSILON = EPSILON - budget\_assigned(ep\_idx,r);
}
Setting \textit{approx\_matrix=0} assigns the full $\epsilon$ to optimization and disables approximated-matrix privacy accounting. The uploaded budget is between lines 221-263, budget assigned at line 363, and node count at line 383.

\paragraph{Candidate outputs} In \texttt{classes/User/User.m}, set the fraction used by \textit{obf\_identifier} at line 201:
\begin{center}
\small 
\begin{tabular}{rrr}
\toprule
Records & Fraction & Count \\
\midrule
2,000 & 0.01 & 20 \\
4,000 & 0.005 & 20 \\
6,000 & 0.0033 & $\approx20$ \\
\bottomrule
\end{tabular}
\end{center}
The implementation retains the nearest candidates by Haversine distance when more locations fall inside the obfuscation range.

\paragraph{Consistency check.} Confirm that the dataset loader, \textit{Regions}, \textit{nr\_loc\_selected}, uploaded-budget definition and assignment, and candidate-output fraction all describe the same experiment. Use names and comments rather than relying only on line numbers, which may change.

\subsection{Experiment Workflow}

The main workflow loads a road network and sampled user inputs, constructs private surrogate representations of LP coefficients, accounts for privacy budget consumed during parameter upload, solves the distributed problem through Benders decomposition, and aggregates the output metrics. Run it from \texttt{Bendersdecomposition/}:
\begin{lstlisting}
run('main.m')
\end{lstlisting}
It reports utility loss, mDP violation ratio, computation time, privacy-loss statistics for uploaded parameters, and surrogate approximation errors.

The baseline workflow is independent. From
\texttt{Bendersdecomposition/baseline/}, edit the configuration near the top of each runner:
\textit{city = 'london';} \textit{node\_count = 2000;} \textit{epsilons = [4, 7, 10];} \textit{user\_ids = 1:10;}

Then run:
\begin{lstlisting}[keepspaces=true]
run('EM_EMBR.m')          % EM, RMP, LP-A
run('laplace_baseline.m') 
run('PAnDA_baseline.m')   % PAnDA and LB
run('COPT_baseline.m')
run('LP_baseline.m')
\end{lstlisting}
The baseline runners print tables to the MATLAB Command Window and do not create new MAT or CSV files. Their displayed epsilon order is $[4,7,10]$, unlike the main runner's internal budget-row order $[10,7,4]$.

\subsection{Major Claims}

\begin{itemize}
  \item \textbf{(C1) Privacy:} DISTA-LP enforces the target metric differential privacy constraints with near-zero empirical violation ratios across the evaluated cities and privacy budgets. Validate with E1 and compare with E2.
  \item \textbf{(C2) Utility:} DISTA-LP improves utility over predefined mechanisms and the PAnDA hybrid while remaining above the lower-bound solution. Validate with E1 and E2.
  \item \textbf{(C3) Scalability:} Benders decomposition makes fine-grained road-network experiments tractable where full centralized LP/COPT runs are unavailable or exceed the 1,800-second limit. Validate with E1 and E2 across node counts.
  \item \textbf{(C4) Private coefficient upload:} Compact surrogates and calibrated parameter perturbation allow coefficient information to be uploaded with measurable privacy loss and approximation error. Validate with E1.
\end{itemize}

\subsection{Evaluation}

\subsubsection{Experiment (E1): DISTA-LP main framework}
\textbf{Estimated effort:} 10--20 human-minutes for setup; compute time is machine- and configuration-dependent. Use 2,000 records for the initial functional test.

\textbf{Preparation.} Select one city in \texttt{readCityMapInfo.m}; activate its matching \textit{Regions, uploaded-budget matrix}, and \textit{budget-assigned} in \texttt{main.m}; set \textit{nr\_loc\_selected}; set \textit{approx\_matrix=1}; and set the matching fraction in \texttt{User.m}. For a Rome 2,000-record test, use five regions, \textit{budget\_assigned\_rome\_svd}, and fraction 0.01.

\textbf{Execution.} From \texttt{Bendersdecomposition/}, run \texttt{main.m}. Allow all active regions and epsilon values to finish. Repeat for each city and record count required by the paper result being checked.

\textbf{Results.} Record utility loss, violation ratio, computation time, upload privacy loss, and surrogate relative/violation errors. Interpret uploaded-budget rows in $[10,7,4]$ order. Confirm successful optimization and inspect any recorded error message before using a result. E1 evaluates C1--C4. Utility loss is divided by 10,000 before printing. Thus
\begin{lstlisting}[language={}]
34708.02 +/- 2414.59
becomes 3.470802 +/- 0.241459,
or 3.47 +/- 0.24 after rounding.
\end{lstlisting}

\subsubsection{Experiment (E2): Baselines}
\textbf{Estimated effort:} 10 human-minutes for configuration; compute time varies widely. LP-A, PAnDA, COPT, and centralized LP can be long-running or resource-limited at larger sizes.

\textbf{Preparation.} From the baseline directory, set \textit{city, node\_count, epsilons=[4,7,10]}, and \textit{user\_ids=1:10}. Both \texttt{COPT\_baseline.m} and \texttt{LP\_baseline.m} use Optimization Toolbox's \texttt{linprog}, set \textit{max\_time\_seconds=1800}. A one-user run may be used only as a functional check; restore ten users for paper reproduction.

\textbf{Execution.} Run \texttt{EM\_EMBR.m}, \texttt{laplace\_baseline.m}, \texttt{PAnDA\_baseline.m}, \texttt{COPT\_baseline.m} and \texttt{LP\_baseline.m}. Copy the printed utility-loss, violation-ratio, and computation-time tables. Repeat for the remaining city/size combinations.

\textbf{Results.} For each repetition, the scripts sum utility loss across ten users and average runtime and violation ratio across users. Then report the mean and sample standard deviation across repetitions.
EM, RMP, and Planar Laplace display zero violation by construction/post-processing; LP-A reports an empirical distance-approximation violation ratio. PAnDA prints separate PAnDA and companion lower-bound rows. E2 evaluates C1--C3.

\subsubsection{Experiment (E3): Surrogate approximation.}
\textbf{Estimated effort:} 10--20 human-minutes per city; compute time depends on configuration. This experiment reproduces the RelErr and ViolRatio columns of Table II for Gaussian, polynomial, RBF, and low-rank SVD surrogates.

\textbf{Preparation.} In \texttt{\small{func/read\_files/readCityMapInfo.m}}, activate exactly one dataset-loading block: Rome (lines 14--17), NYC (20--23), or London (26--29). In \path{classes/User/User.m}, uncomment fitting calls at lines 113--115, noisy-parameter calls at 118--120, and reconstruction calls at 123--125; keep the corresponding SVD calls enabled. Set the \textit{obf\_identifier} fraction at approximately line 201 to 0.01. In \path{main_surrogate.m}, activate the corresponding 2,000-node region block: Rome (124--128), NYC (143--146), or London (156--161). Ensure the effective \textit{nr\_loc\_selected} is 2000, with \textit{NR\_LOC=1} if the later \textit{NR\_LOC*2000} assignment remains active.

\textbf{Execution.} From \texttt{Bendersdecomposition/}, run \texttt{main\_surrogate.m}. Allow all regions and users to finish. Preserve \texttt{Summary\_Gaussian.csv}, \texttt{Summary\_Polynomial.csv}, \texttt{Summary\_RBF.csv}, and \texttt{Summary\_SVD.csv} before repeating for the other cities.

\textbf{Results.} Compare generated summaries with those in \path{rome_relErr_violRatio/}, \path{nyc_relErr_violRatio/}, and \path{london_relErr_violRatio/}, relative to \path{Bendersdecomposition/}. Match \texttt{Rel\_DLR}/\texttt{Viol\_DLR}, \texttt{Rel\_DLR2}/\texttt{Viol\_DLR2}, and \texttt{Rel\_CRL} to the coefficient blocks in Table II. Surrogate ViolRatio differs from the final mechanism's metric-DP violation ratio. Compare original and reconstructed coefficients before \textit{User.apply\_svd()} overwrites the working matrices; otherwise artificial zero errors can result. The dedicated surrogate runner avoids this replacement. A modified optimization runner must retain original matrices or use the preserved \texttt{*\_test} copies. E3 evaluates C4.

\subsubsection{Experiment (E4): End-to-end result check}
\textbf{Estimated effort:} 15 human-minutes after experiments complete.

\textbf{Preparation.} Create one table per city and node count with columns $\epsilon=4,7,10$ and rows for DISTA-LP and every available baseline.

\textbf{Execution.} Transfer E1 and E2 values without changing their epsilon ordering. Compare utility loss, violation ratio, and runtime with the corresponding paper table or figure.

\textbf{Results.} Reproduced results should preserve the paper's qualitative ranking. As a supplied reference point, the main README reports Rome, 4,000 records, $\epsilon=4$: DISTA-LP utility loss $2.88\pm0.21$ in 10,000-meter units, compared with approximately $7.16\pm2.53$ for EM/Laplace, $6.76\pm2.55$ for LP, and a lower bound near 1.7. Treat this value only as a cross-check; the paper and generated outputs are authoritative.

\subsection{Customization}

The artifact can be customized by selecting another supported city, record count, epsilon, user subset, or surrogate family. Root-level scripts test Gaussian, polynomial, RBF, and low-rank SVD surrogates. When adding a new city or size, supply compatible road-network CSVs, sampled MAT inputs, target-region definitions, graph objects where required by baselines, and a budget matrix with one column per repetition. Preserve the epsilon-to-row mapping, document all parameter changes and use separate output directories to avoid overwriting reference runs.

\subsection{Notes}

\begin{itemize}
  \item Baseline output is currently \emph{mean $\pm$ sample standard deviation}, not a 95\% confidence interval. If the paper uses 95\% normal-approximation intervals, compute $1.96s/\sqrt{R}$ and label the statistic consistently.
  \item Divide both the utility-loss mean and its deviation by 10,000.
  \item Primary baseline runners do not read or overwrite historical MAT files. Auxiliary files containing \texttt{memory} can write profiling checkpoints. Main and surrogate runnerscan overwrite summaries or sampled inputs; preserve eachrun before reconfiguration
  \item \texttt{main\_BD.m} is legacy/development code and is not a primary baseline entry point.
\end{itemize}



\end{document}